\documentclass[journal]{IEEEtran}
\usepackage{amssymb}

\usepackage{amsthm}
\usepackage{booktabs}

\newtheorem{theorem}{Theorem}[section]

\theoremstyle{definition}

\newtheorem*{Proof}{Proof}
\newtheorem{lemma}{Lemma}[section]

\theoremstyle{remark}
\newtheorem{remark}[theorem]{Remark}

\renewenvironment{proof}{\begin{Proof}}{\end{Proof}}

\usepackage{cite}

\ifCLASSINFOpdf
  \usepackage[pdftex]{graphicx}
\else
  \usepackage[dvips]{graphicx}
\fi
\usepackage{amsmath}
\usepackage{rotating}
\usepackage{multirow}
\usepackage{url}
\usepackage{color}

\usepackage{algorithm, algpseudocode}

\global\long\def\SP{SP}
\global\long\def\ASEW{ASP}
\global\long\def\SPDP{PDP}
\global\long\def\ONVW{OVDP}
\global\long\def\ONVWOne{OVDP1}
\global\long\def\ONVWTwo{OVDP2}
\global\long\def\ONVWThree{OVDP3}
\global\long\def\SPs{SP }
\global\long\def\ASEWs{ASP }
\global\long\def\SPDPs{PDP }
\global\long\def\ONVWs{OVDP }
\global\long\def\ONVWOnes{OVDP1 }
\global\long\def\ONVWTwos{OVDP2 }
\global\long\def\ONVWThrees{OVDP3 }

\newcommand{\argmin}{\mathop{\mathrm{argmin}}\limits}

\global\long\def\st{\mathrm{s.t.}}
\global\long\def\RealNumSet{\mathbb{R}}
\global\long\def\FuncIndi#1#2{\iota_{#1}(#2)}
\global\long\def\NormOp#1{\left\|#1\right\|_{\mathrm{op}}}
\global\long\def\NormOpNoResize#1{\|#1\|_{\mathrm{op}}}
\global\long\def\NormET#1{\left\|#1\right\|_{2}}
\global\long\def\NormETNoResize#1{\|#1\|_{2}}

\global\long\def\Precon{\mathbf{\Gamma}}
\global\long\def\PreconSca{\Gamma}
\global\long\def\PreconOpNorm{\mu}
\global\long\def\PreconOpNormSP{\mu_{\SP}}
\global\long\def\ConstASPSPrimal{\sigma}
\global\long\def\ConstASPSDual{\tau}
\global\long\def\IndASPS{l}
\global\long\def\ParamPDS{\gamma}
\global\long\def\ParamSPDP{\tau}
\global\long\def\ParamSPDStheta{\theta}
\global\long\def\ParamONP{\beta}
\global\long\def\VarPrimalElem{x}
\global\long\def\VarDualElem{y}
\global\long\def\VarPrimal{\mathbf{\VarPrimalElem}}
\global\long\def\VarDual{\mathbf{\VarDualElem}}
\global\long\def\FuncPrimal{f}
\global\long\def\FuncDual{g}
\global\long\def\LinOpe{\mathfrak{L}}
\global\long\def\MatLinOpe{\mathbf{L}}
\global\long\def\NumPrimal{N}
\global\long\def\NumDual{M}
\global\long\def\IndPrimal{i}
\global\long\def\IndDual{j}
\global\long\def\NumElemPrimal{n}
\global\long\def\NumElemDual{m}
\global\long\def\NumElemPrimalAll{\tilde{n}}
\global\long\def\NumElemDualAll{\tilde{m}}
\global\long\def\ZeroOpe{\mathbf{O}}

\global\long\def\prox{\mathrm{prox}}

\global\long\def\InnerIter{t}

\global\long\def\NumVerPixUnm{N_{1}}
\global\long\def\NumHorPixUnm{N_{2}}
\global\long\def\NumBandUnm{N_{3}}
\global\long\def\NumEndmember{N_{e}}
\global\long\def\IndEndmember{e}
\global\long\def\NumPixUnm{\NumVerPixUnm\NumHorPixUnm}
\global\long\def\IndPixUnm{i}

\global\long\def\MatEndmemberSymbUnm{E}
\global\long\def\MatEndmemberUnm{\mathbf{\MatEndmemberSymbUnm}}
\global\long\def\MatEndmemberExpUnm{\widetilde{\MatEndmemberUnm}}
\global\long\def\VecEmdmemberSymbUnm{e}
\global\long\def\VecEmdmemberUnm{\mathbf{\VecEmdmemberSymbUnm}}
\global\long\def\AbundanceSymbUnm{a}
\global\long\def\AbundanceUnm{\mathbf{\AbundanceSymbUnm}}
\global\long\def\NoiseUnm{\mathbf{n}}

\global\long\def\NNRealNumSet{\RealNumSet_{+}}

\global\long\def\ObsHSIUnm{\mathbf{v}}

\global\long\def\ParamFidelUnm{\varepsilon}
\global\long\def\BallFidelUnm{B_{2, \ParamFidelUnm}^{\ObsHSIUnm}}

\global\long\def\VarDualUnm{\mathbf{z}}

\global\long\def\StanDivGaussUnm{\sigma}

\global\long\def\SymbGraph{\mathcal{G}}
\global\long\def\Vertex{\mathcal{V}}
\global\long\def\Edge{\mathcal{E}}
\global\long\def\MatWeightElem{W}
\global\long\def\MatWeight{\mathbf{\MatWeightElem}}
\global\long\def\NumVertex{N_{\SymbGraph}}
\global\long\def\NumSampVertex{M_{\SymbGraph}}
\global\long\def\IndVertex{i}
\global\long\def\IndVertexCor{j}
\global\long\def\NeiVertexI#1{\mathcal{N}(#1)}

\global\long\def\VarOneElemGS{x}
\global\long\def\VarTwoElemGS{y}
\global\long\def\VarOneGS{\mathbf{\VarOneElemGS}}
\global\long\def\VarTwoGS{\mathbf{\VarTwoElemGS}}

\global\long\def\GTV#1{\|#1\|_{\mathrm{GTV}}}

\global\long\def\ObsGS{\mathbf{v}}
\global\long\def\GTGS{\mathbf{u}}
\global\long\def\NoiseGS{\mathbf{n}}

\global\long\def\MatSampElemGS{\Phi}
\global\long\def\MatSampGS{\mathbf{\MatSampElemGS}}
\global\long\def\DiffGS{\mathbf{D}_{\SymbGraph}}

\global\long\def\ParmFidelGS{\varepsilon}
\global\long\def\BallFidelGS{B_{2, \ParmFidelGS}^{\ObsGS}}

\global\long\def\VarDualGS{\mathbf{z}}

\global\long\def\StanDivGaussGS{\sigma}

\global\long\def\NumVerPixMNR{N_{1}}
\global\long\def\NumHorPixMNR{N_{2}}
\global\long\def\NumBandMNR{N_{3}}
\global\long\def\ObsMNR{\mathbf{v}}
\global\long\def\GTMNR{\mathbf{u}}
\global\long\def\SparseMNR{\mathbf{s}}
\global\long\def\StripeMNR{\mathbf{l}}
\global\long\def\NoiseMNR{\mathbf{n}}

\global\long\def\DiffvSymb{\mathfrak{D}_{v}}
\global\long\def\DiffhSymb{\mathfrak{D}_{h}}
\global\long\def\DiffbSymb{\mathfrak{D}_{b}}
\global\long\def\Diffv#1{\DiffvSymb(#1)}
\global\long\def\Diffh#1{\DiffhSymb(#1)}
\global\long\def\Diffb#1{\DiffbSymb(#1)}

\global\long\def\ZeroElem{\mathbf{0}}
\global\long\def\ParamFidelMNR{\varepsilon}
\global\long\def\ParamSparMNR{\eta_{\SparseMNR}}
\global\long\def\BallFidelMNR{B_{2, \ParamFidelMNR}^{\ObsMNR}}
\global\long\def\BallSparMNR{B_{1, \ParamSparMNR}^{\ZeroElem}}

\global\long\def\ParamBalanceMNR{\lambda}
\global\long\def\RateSparse{p_{\SparseMNR}}

\global\long\def\VarDualMNR{\mathbf{z}}

\global\long\def\StanDivGaussMNR{\sigma}

\makeatletter
\def\bstctlcite{\@ifnextchar[{\@bstctlcite}{\@bstctlcite[@auxout]}}
\def\@bstctlcite[#1]#2{\@bsphack
	\@for\@citeb:=#2\do{%
		\edef\@citeb{\expandafter\@firstofone\@citeb}%
		\if@filesw\immediate\write\csname #1\endcsname{\string\citation{\@citeb}}\fi}%
	\@esphack}
\makeatother

\begin{document}
\bstctlcite{IEEEexample:BSTcontrol}
%
\title{Variable-Wise Diagonal Preconditioning for Primal-Dual Splitting: Design and Applications}
%
%
%

\author{Kazuki~Naganuma,~\IEEEmembership{Student~Member,~IEEE,}
        Shunsuke~Ono,~\IEEEmembership{Member,~IEEE,}
\thanks{Manuscript received XXX, XXX; revised XXX XXX, XXX.}%
\thanks{K. Naganuma is with the Department of Computer Science, Tokyo Institute of Technology, Yokohama, 226-8503, Japan (e-mail: naganuma.k.aa@m.titech.ac.jp).}
\thanks{S. Ono is with the Department of Computer Science, Tokyo Institute of Technology, Yokohama, 226-8503, Japan (e-mail: ono@c.titech.ac.jp).}
\thanks{This work was supported Grant-in-Aid for JSPS Fellows under Grant 23KJ0912, in part by JST PRESTO under Grant JPMJPR21C4 and JST AdCORP under Grant JPMJKB2307, and in part by JSPS KAKENHI under Grant 22H03610, 22H00512, and 23H01415.}}

%
%

\markboth{Journal of \LaTeX\ Class Files,~Vol.~XX, No.~X, August~20XX}%
{Shell \MakeLowercase{\textit{et al.}}: Bare Demo of IEEEtran.cls for IEEE Journals}
%



\maketitle

\begin{abstract}
This paper proposes a method for designing diagonal preconditioners for a preconditioned primal-dual splitting method (P-PDS), an efficient algorithm that solves nonsmooth convex optimization problems. To speed up the convergence of P-PDS, a design method has been proposed to automatically determine appropriate preconditioners from the problem structure. However, the existing method has two limitations. One is that it directly accesses all elements of matrices representing linear operators involved in a given problem, which is inconvenient for handling linear operators implemented as procedures rather than matrices. The other is that it takes an element-wise preconditioning approach, which turns certain types of proximity operators into analytically intractable forms. To overcome these limitations, we establish an Operator norm-based design method of Variable-wise Diagonal Preconditioning (\ONVW). First, \ONVWs constructs diagonal preconditioners using only (upper bounds) of the operator norms of linear operators, thus eliminating the need for their explicit matrix representations. Furthermore, since \ONVWs takes a variable-wise preconditioning approach, it keeps any proximity operator analytically computable. We also prove that our preconditioners satisfy the convergence condition of P-PDS. Finally, we demonstrate the effectiveness and usefulness of \ONVWs through applications to mixed noise removal of hyperspectral images, hyperspectral unmixing, and graph signal recovery.
\end{abstract}

\begin{IEEEkeywords}
Primal-dual splitting method (PDS), diagonal preconditioning, convex optimization, signal estimation
\end{IEEEkeywords}

%
\IEEEpeerreviewmaketitle

\section{Introduction}
\label{sec:intro}

Many signal estimation and processing problems, such as denoising, interpolation, decomposition, and reconstruction, have been resolved by casting them as convex optimization problems~\cite{parikh2014proximal,combettes2021fixed} of the form:
\vspace{-2mm}
\begin{align}
	\label{prob:general_form_of_optimization}
	\min_{\substack{\VarPrimal_{1}, \ldots, \VarPrimal_{\NumPrimal}, \\ 
			\VarDual_{1}, \ldots, \VarDual_{\NumDual}}} \:
	& \sum_{\IndPrimal=1}^{\NumPrimal}\FuncPrimal_{\IndPrimal}(\VarPrimal_{\IndPrimal}) 
	+ \sum_{\IndDual=1}^{\NumDual}\FuncDual_{\IndDual}\left(\VarDual_{\IndDual}\right) \nonumber \\ 
	\mathrm{s.t.} \:  
	& \VarDual_{1} = \sum_{\IndPrimal=1}^{\NumPrimal}\LinOpe_{1,\IndPrimal}(\VarPrimal_{\IndPrimal}), 
	\ldots, 
	\VarDual_{\NumDual}=\sum_{\IndPrimal=1}^{\NumPrimal}\LinOpe_{\NumDual,\IndPrimal}(\VarPrimal_{\IndPrimal}),
\end{align}
where $\FuncPrimal_{\IndPrimal}:\RealNumSet^{\NumElemPrimal_{\IndPrimal}}\rightarrow(-\infty,+\infty]$  and $\FuncDual_{\IndDual}:\RealNumSet^{\NumElemDual_{\IndDual}}\rightarrow(-\infty,+\infty]$ are proximable\footnotemark ~proper lower-semicontinuous convex functions, and $ \LinOpe_{\IndDual,\IndPrimal}:\RealNumSet^{\NumElemPrimal_{\IndPrimal}}\rightarrow\RealNumSet^{\NumElemDual_{\IndDual}}$ are linear operators ($\forall \IndPrimal=1, \ldots, \NumPrimal$ and $\forall \IndDual=1, \ldots, \NumDual$). The variables $\VarPrimal_{1},\ldots,\VarPrimal_{\NumPrimal}$ represent estimated signals or components, and $\VarDual_{1},\ldots,\VarDual_{\NumDual}$ are auxiliary variables for splitting.
\footnotetext{If an efficient computation of the proximity operator (see. Eq.~\eqref{eq:proximity_operator}) of $f$ is available, we call $f$ proximable.}

As a method for solving Prob.~\eqref{prob:general_form_of_optimization}, a primal-dual splitting method (PDS)~\cite{PDS_1} has attracted attention~\cite{PDS_appl_Condat_2014,AdaptivePDS_Goldstein_2015,PDS_appl_Ono_2015,PDS_review_2015,ono2017primal,PDS_Linesearch_Malitsky_2018,PDS_appl_Kyochi_2021,GPDS_He_2022,PDS_Linesearch_gold_Chang_2022} due to its simple implementation without operator inversions.\footnote{This algorithm has been generalized by Condat~\cite{PDS_2} and Vu~\cite{PDS_3}, where smooth convex functions are optimized by using their Lipschitzian gradients.} 
To improve the convergence speed of PDS, a preconditioned PDS (P-PDS) has been studied~\cite{DP-PDS,DP_PDS_grad,pmlr-v108-ye20a,SPDP}.
P-PDS is a generalization of the standard PDS, where the scalar-valued stepsizes of PDS are replaced by (positive definite) matrix-valued preconditioners.
The theoretical convergence of P-PDS is established in a primal-dual space equipped with a skewed metric, which is determined by the linear operators involved in the optimization problem and the preconditioners used (see~\cite{PDS_2,DP-PDS,ono_2015} for details). 
Preconditioning can be viewed as the selection of an appropriate metric for optimization algorithms and is a crucial long-standing issue not only in P-PDS but also in various proximal algorithms~\cite{COMBETTES201317,BStephen_quasiNewton_2019}.

The appropriate preconditioners that accelerate the convergence of P-PDS vary greatly depending on the structure of the target optimization problem (see Section~\ref{sec:experiments} for detailed examples). 
To automatically determine such preconditioners, the authors in~\cite{DP-PDS} have proposed a diagonal-preconditioner design method.
The elements of the diagonal preconditioners consist of the row/column absolute sum of the elements of the explicit matrices representing the linear operators $\LinOpe_{j,i}$ in~\eqref{prob:general_form_of_optimization}, and thus the resulting diagonal elements of the preconditioners can be different for each element in one variable. 

Although this design method determines reasonable diagonal preconditioners, there exist two limitations that are considerable in real-world applications. First, the method is difficult to apply in the case where (some of) the linear operators $\LinOpe_{j,i}$ in Prob.~\eqref{prob:general_form_of_optimization} are not implemented as explicit matrices because it requires access to the entire elements of the matrices to construct the preconditioners. We often encounter such situations, especially in imaging applications, where the linear operators are implemented not as explicit matrices but as procedures that compute forward and adjoint operations in an efficient manner, e.g., difference operators~\cite{chambolle2010introduction,difference_operator_1} and frame transforms~\cite{Frame1,Frame2,Frame3}. Second, some proximable functions $\FuncPrimal_{\IndPrimal}$ and $\FuncDual_{\IndDual}$ are not completely separable for each element of the input variables $\VarPrimal_{\IndPrimal}$ and $\VarDual_{\IndDual}$, e.g., mixed norms and the indicator functions of norm balls~\cite{Proximity_Operator_Repository}. For such functions, the element-wise preconditioning might make the functions non-proximable.

To address the above issues, this paper proposes an Operator-norm-based design method of Variable-wise Diagonal Preconditioning (\ONVW). 
Specifically, we introduce a new general form of P-PDS preconditioners, and then propose specific preconditioners based on this general form. 
We also prove that the sequence generated by P-PDS with \ONVWs converges to an optimal solution of Prob.~\eqref{prob:general_form_of_optimization}.
 
Our method has two features preferred in many real-world applications. 
First, our preconditioners can be computed from (upper bounds of) the operator norms of the linear operators $\LinOpe_{j,i}$, meaning that our method does not need their explicit matrix representations. 
This is because (upper bounds of) the operator norms are often known or can be estimated without matrix implementation for typical linear operators used in signal processing applications, including the ones mentioned above. 
Second, the elements of the diagonal preconditioners obtained by our method take the same value for all the elements of each variable, i.e., variable-wise preconditioning. 
This maintains the proximablity of the functions $\FuncPrimal_{\IndPrimal}$ and $\FuncDual_{\IndDual}$ in Prob.~\eqref{prob:general_form_of_optimization}.

Comprehensive experiments are conducted by applying our method to three signal estimation problems: mixed noise removal of hyperspectral images, hyperspectral unmixing, and graph signal recovery. 
By discussing the convergence in these three optimization problems, which have very different structures, we demonstrate the effectiveness and usefulness of our method.

This paper is organized as follows.
Section~\ref{sec:preliminaries} gives preliminaries on mathematical tools, the description of P-PDS, and reviews of existing preconditioner design methods. 
In Section~\ref{sec:proposed}, we present \ONVWs and prove the convergence theorem of P-PDS with \ONVW. 
Their applications to mixed noise removal hyperspectral images, hyperspectral unmixing, and graph signal recovery are given in Section~\ref{sec:experiments}. 
Finally, we conclude the paper in Section~\ref{sec:conclusion}.

The preliminary version of this work, without the generalization of our method, applications to various signal estimation tasks, or deeper discussion, has appeared in conference proceedings~\cite{P-PDS_naganuma_2022}.

\section{Preliminaries}
\label{sec:preliminaries}
\subsection{Notations}
\label{ssec:Notations}
In this paper, 
vectors and matrices are denoted by lowercase and uppercase bold letters, for example, $\mathbf{x}$ and $\mathbf{X}$, respectively.
For a vector $\mathbf{x} = [x_{1}, \ldots, x_{N}]^{\top} \in \RealNumSet^{N}$, each scalar value $x_{i}$ $(1\leq i\leq N)$ is called the $i$th element of $\mathbf{x}$ and the $\ell_{p}$ norm of $\mathbf{x}$ is defined by $\|\mathbf{x}\|_{p} = (\sum_{i = 1}^{N} |x_{i}|^{p})^{1/p}$ for $p\geq 1$.
Similarly, for a matrix $\mathbf{X} = [x_{j,i}]_{1\leq j \leq M, 1\leq i \leq N}$, each scalar value $x_{j,i}$ $(1\leq j \leq M, 1\leq i \leq N)$ is called the $(j, i)$th element of $\mathbf{X}$.
We denote a matrix $\mathbf{X}\in \RealNumSet^{\NumElemDualAll\times \NumElemPrimalAll}$ ($\NumElemDualAll = \sum_{\IndDual = 1}^{\NumDual} \NumElemDual_{\IndDual}, \NumElemPrimalAll = \sum_{\IndPrimal = 1}^{\NumPrimal} \NumElemPrimal_{\IndPrimal}$) consisting of block matrices $\mathbf{X}_{j,i} \in \RealNumSet^{\NumElemDual_{\IndDual} \times \NumElemPrimal_{\IndPrimal}}$ ($j = 1, \ldots , M$ and $i = 1, \ldots , N$) by $\mathbf{X} = [\mathbf{X}_{j,i}]_{1\leq j \leq M, 1\leq i \leq N}$.

Let $\mathfrak{L}:\RealNumSet^{N}\rightarrow \RealNumSet^{M}$ be a linear operator. 
We denote the adjoint operator of $\mathfrak{L}$ as $\mathfrak{L}^*$, which satisfies  $\langle\mathfrak{L}(\mathbf{x}),\mathbf{y}\rangle=\langle\mathbf{x},\mathfrak{L}^{*}(\mathbf{y})\rangle$ for any $\mathbf{x}\in\mathbb{R}^{N}$ and  $\mathbf{y}\in\mathbb{R}^{M}$.

\subsection{Mathematical Tools}
\label{ssec:Mathematical_Tools}

Let $f : \RealNumSet^{N} \rightarrow (-\infty,\infty]$ be a proximable proper lower-semicontinuous convex function and $\mathbf{G} \in \RealNumSet^{N\times N}$ be a symmetric and
positive definite matrix. 
The proximity operator of $f$ relative to the metric induced by $\mathbf{G}$ is defined as 
\begin{equation}
	\label{eq:skewed_proximity_operator}
	\prox_{\mathbf{G},f}(\mathbf{x}):=\argmin_{\mathbf{y}}\frac{1}{2}\langle \mathbf{x - y}, \mathbf{G(x-y)}\rangle + f(\mathbf{y}),
\end{equation}
where $\langle \cdot ,\cdot \rangle$ is the Euclidean inner product. If $\mathbf{G}$ is a positive scalar matrix, i.e., $\mathbf{G}=\alpha\mathbf{I}$ $(\alpha>0)$, the proximity operator is identical to the standard proximity operator:
\begin{equation}
	\label{eq:proximity_operator}
	\prox_{\mathbf{G},f}(\mathbf{x}) = \argmin_{\mathbf{y}}\frac{1}{2}\|\mathbf{x-y}\|_{2}^{2}+\frac{1}{\alpha}f(\mathbf{y}).
\end{equation}
In this paper, the proximity operator relative to the metric induced by a positive matrix that is not scalar matrix is called the \textit{skewed proximity operator}. We would like to note that the standard proximity operators of some popular convex functions, such as the mixed $\ell_{1,2}$-norm and the indicator functions of norm balls, have analytic solutions but their computations are not completely separable element by element. 
In such cases, even if $\mathbf{G}$ is diagonal (with different elements), the computation of the skewed proximity operator becomes difficult.

\textit{The Fenchel--Rockafellar conjugate function} of $f$ is defined as
\begin{equation}
	\label{eq:def_FRconjugate_function}
	f^{*}(\mathbf{x}) := \max_{\mathbf{y}} \langle \mathbf{x}, \mathbf{y} \rangle  - f(\mathbf{y}).
\end{equation}
Thanks to the generalization of Moreau's Identity~\cite[Theorem 3.1 (ii)]{G_Moreau_Identity_2013}, the skewed proximity operator of $f^{*}$ is calculated as
\begin{equation}
	\prox_{\mathbf{G}, f^{*}}(\mathbf{x}) = \mathbf{x} - \mathbf{G}^{-1}\prox_{\mathbf{G}^{-1}, f}(\mathbf{G}\mathbf{x}).
\end{equation}

For a given nonempty closed convex set $C\subset \RealNumSet^{N}$, the indicator function of $C$ is defined by 
\begin{equation}
	\label{eq:indicator_function}
\iota_{C}(\mathcal{X})
:=
\begin{cases}
	0, & \mathrm{if} \: \mathcal{X}\in C; \\
	\infty, & \mathrm{otherwise}.
\end{cases} 
\end{equation}
The proximity operator of the indicator function $\iota_{C}$ is equivalent to the convex projection onto $C$.
The following convex sets are useful in signal processing applications.
\begin{itemize}
	\item For $\mathbf{c} \in \RealNumSet^{N}$, the $\mathbf{c}$-centered $\ell_{p}$-ball ($p=1$ or $2$) with the radius $\alpha > 0$ defined by
	\begin{equation}
		\label{eq:norm_ball}
		B_{p, \alpha}^{\mathbf{c}} := \{ \mathbf{x} \in \RealNumSet^{N} \: | \: \|\mathbf{x} - \mathbf{c} \|_{p} \leq \alpha \}.
	\end{equation}
	\item The nonnegative orthant $\NNRealNumSet^{N} := [0, +\infty)^{N}$.
\end{itemize}

For a linear operator $\mathfrak{L}$, the operator norm $\NormOpNoResize{\mathfrak{L}}$ is defined by
\begin{equation}
	\label{eq:operator_norm}
	\NormOpNoResize{\mathfrak{L}} := \sup_{\mathbf{x} \neq \ZeroElem} \frac{\|\mathfrak{L}(\mathbf{x})\|_{2}}{\|\mathbf{x}\|_{2}}.
\end{equation}
For a matrix $\mathbf{A}$, its operator norm satisfies
\begin{equation}
	\NormOp{\mathbf{A}}:=\sup_{\mathbf{x}\neq\ZeroElem}\frac{\|\mathbf{Ax}\|_{2}}{\|\mathbf{x}\|_{2}}=\sigma_{1}(\mathbf{A}),
\end{equation}
where $\sigma_{1}(\mathbf{A})$ is the maximum singular value of $\mathbf{A}$.
Let $\mathfrak{L}_{1}\circ\mathfrak{L}_{2}$ be the composition of linear operators $\mathfrak{L}_{1}$ and $\mathfrak{L}_{2}$.
The operator norm of $\mathfrak{L}_{1}\circ\mathfrak{L}_{2}$ satisfies that
\begin{equation}
	\label{eq:submultiplicity}
	\NormOpNoResize{\mathfrak{L}_{1}\circ\mathfrak{L}_{2}} \leq \NormOpNoResize{\mathfrak{L}_{1}}\NormOpNoResize{\mathfrak{L}_{2}}.
\end{equation}
This property is called \textit{the submultiplicity}.

\subsection{Preconditioned PDS (P-PDS)}
For Prob.~\eqref{prob:general_form_of_optimization}, let $\VarPrimal = [\VarPrimal_{1}^{\top}, \ldots, \VarPrimal_{\NumPrimal}^{\top}]^{\top} \in \RealNumSet^{\NumElemPrimalAll}$ ($\NumElemPrimalAll = \sum_{\IndPrimal = 1}^{\NumPrimal} \NumElemPrimal_{\IndPrimal}$), 
$\VarDual = [\VarDual_{1}^{\top}, \ldots, \VarDual_{\NumDual}^{\top}]^{\top} \in \RealNumSet^{\NumElemDualAll}$ ($\NumElemDualAll = \sum_{\IndDual = 1}^{\NumDual} \NumElemDual_{\IndDual}$),
$\FuncPrimal(\VarPrimal) = \sum_{\IndPrimal = 1}^{\NumPrimal} \FuncPrimal_{\IndPrimal}(\VarPrimal_{\IndPrimal})$, 
$\FuncDual(\VarDual) = \sum_{\IndDual = 1}^{\NumDual} \FuncDual_{\IndDual}(\VarDual_{\IndDual})$,
and 
\begin{equation}
	\label{eq:jointed_linear_operator}
	\LinOpe:=
	\begin{bmatrix} 
		\LinOpe_{1,1} & \LinOpe_{1,2} & \cdots & \LinOpe_{1,\NumPrimal} \\ 
		\LinOpe_{2,1} & \LinOpe_{2,2} & \cdots & \LinOpe_{2,\NumPrimal} \\ 
		\vdots & \vdots & \ddots & \vdots \\ 
		\LinOpe_{\NumDual, 1} & \LinOpe_{\NumDual, 2} & \cdots & \LinOpe_{\NumDual,\NumPrimal} 
	\end{bmatrix}.
\end{equation}
P-PDS~\cite{DP-PDS} computes an optimal solution of Prob.~\eqref{prob:general_form_of_optimization} by the following iterative procedures:
\begin{align}
	\left\lfloor
	\begin{array}{l}
		\VarPrimal^{(\InnerIter + 1)} \leftarrow \prox_{\Precon_{1}^{-1}, \FuncPrimal} (\VarPrimal^{(\InnerIter)} - \Precon_{1}\LinOpe^{*}(\VarDual^{(\InnerIter)})), \\
		\VarDual^{(\InnerIter + 1)} \leftarrow \prox_{\Precon_{2}^{-1}, \FuncDual^{*}} (\VarDual^{(\InnerIter)} + \Precon_{2}\LinOpe(2\VarPrimal^{(\InnerIter + 1)} - \VarPrimal^{(\InnerIter)})),
	\end{array}
	\right.
	\label{eq:P_PDS}
\end{align}
where $\Precon_{1} \in \RealNumSet^{\NumElemPrimalAll\times \NumElemPrimalAll}$ and $\Precon_{2} \in \RealNumSet^{\NumElemDualAll \times \NumElemDualAll}$ are symmetric and positive definite matrices called \textit{preconditioners}.

If $\Precon_{1}$ and $\Precon_{2}$ are block-diagonal matrices, that is, $\Precon_{1} = \mathrm{diag}(\Precon_{1,1}, \ldots, \Precon_{1,\NumPrimal})$ and $\Precon_{2} = \mathrm{diag}(\Precon_{2,1}, \ldots, \Precon_{2,\NumDual})$ for matrices $\Precon_{1,1}, \ldots, \Precon_{1,\NumPrimal}, \Precon_{2,1}, \ldots, \Precon_{2,\NumDual}$ corresponding to $\VarPrimal_{1}, \ldots, \VarPrimal_{\NumPrimal}, \VarDual_{1}, \ldots, \VarDual_{\NumDual}$, the procedures in~\eqref{eq:P_PDS} can be rewritten as the following equivalent form:
\begin{align}
\left\lfloor
\begin{array}{l}
	\VarPrimal_{1}^{(\InnerIter + 1)} 
	\leftarrow 
	\prox_{\Precon_{1,1}^{-1}, \FuncPrimal_{1}} 
	(\VarPrimal_{1}^{(\InnerIter)} 
	- \Precon_{1,1} \sum_{\IndDual = 1}^{\NumDual} \LinOpe_{\IndDual, 1}^{*}(\VarDual_{\IndDual}^{(\InnerIter)})), \\
	\vdots \\
	\VarPrimal_{\NumPrimal}^{(\InnerIter + 1)} 
	\leftarrow 
	\prox_{\Precon_{1,\NumPrimal}^{-1}, \FuncPrimal_{\NumPrimal}} 
	(\VarPrimal_{\NumPrimal}^{(\InnerIter)} 
	- \Precon_{1,\NumPrimal} \sum_{\IndDual = 1}^{\NumDual} \LinOpe_{\IndDual, \NumPrimal}^{*}(\VarDual_{\IndDual}^{(\InnerIter)})), \\
	\VarDual_{1}^{(\InnerIter + 1)} 
	\leftarrow \\ 
	~~~ \prox_{\Precon_{2,1}^{-1}, \FuncDual_{1}^{*}}(\VarDual_{1}^{(\InnerIter)} + \Precon_{2,1} \sum_{\IndPrimal = 1}^{\NumPrimal}\LinOpe_{1, \IndPrimal}(2\VarPrimal_{\IndPrimal}^{(\InnerIter + 1)} - \VarPrimal_{\IndPrimal}^{(\InnerIter)})), \\
	\vdots \\
	\VarDual_{\NumDual}^{(\InnerIter + 1)} 
	\leftarrow \\ 
	~~~ \prox_{\Precon_{2,\NumDual}^{-1}, \FuncDual_{\NumDual}^{*}}(\VarDual_{\NumDual}^{(\InnerIter)} + \Precon_{2,\NumDual} \sum_{\IndPrimal = 1}^{\NumPrimal}\LinOpe_{\NumDual, \IndPrimal}(2\VarPrimal_{\IndPrimal}^{(\InnerIter + 1)} - \VarPrimal_{\IndPrimal}^{(\InnerIter)})). 
\end{array}
\right.
\label{eq:P_PDS_BD}
\end{align}
Compared with~\eqref{eq:P_PDS}, the procedures in~\eqref{eq:P_PDS_BD} can easily be calculated because it avoids the computations of the skewed proximity operators and linear operators over the entire variables.

Here, we introduce the convergence theorem of P-PDS.

\begin{theorem}{ \cite[Theorem~1]{DP-PDS} }
	\label{theo:convergence_property_of_DP_PDS}
	Let $\Precon_{1}$ and $\Precon_{2}$ be symmetric and
	positive definite matrices satisfying
	\begin{equation}
		\label{eq:convergence_equation}
		\NormOp{\Precon_{2}^{\frac{1}{2}}\circ\LinOpe\circ\Precon_{1}^{\frac{1}{2}}}^{2} < 1.
	\end{equation}
	Then, the sequence $(\VarPrimal_{1}^{(\InnerIter)},\ldots,\VarPrimal_{\NumPrimal}^{(\InnerIter)},\VarDual_{1}^{(\InnerIter)},\ldots,\VarDual_{\NumDual}^{(\InnerIter)})$ generated by~\eqref{eq:P_PDS} converges to an optimal solution $(\VarPrimal_{1}^{*},\ldots,\VarPrimal_{\NumPrimal}^{*},\VarDual_{1}^{*},\ldots,$ $\VarDual_{\NumDual}^{*})$ of Prob.~\eqref{prob:general_form_of_optimization}.
\end{theorem}

\subsection{Existing Preconditioner Design Methods}
\label{ssec:existing_preconditioners}
\subsubsection{Scalar Preconditioning (\SP)}
The standard PDS~\cite{PDS_1} can be recovered by setting the preconditioners to be scalar matrices, i.e., 
\begin{equation}
	\label{eq:precon_SP}
	\Precon_{1}=\ParamPDS_{1}\mathbf{I}, 
	\Precon_{2}=\ParamPDS_{2}\mathbf{I}.
\end{equation}
The parameters $\ParamPDS_{1}$ and $\ParamPDS_{2}$ are positive scalars that satisfy~\eqref{eq:convergence_equation}, that is, 
\begin{equation}
	\label{eq:PDS_condition}
	\ParamPDS_{1}\ParamPDS_{2}\NormOp{\LinOpe}^{2} < 1.
\end{equation}
In practice, the parameter $\ParamPDS_{2}$ is often set as
\begin{equation}
	\label{eq:SP_expriment}
	\ParamPDS_{2} = \frac{1}{\PreconOpNormSP^2\gamma_{1}},
\end{equation}
where $\PreconOpNormSP$ is an upper bound of $\NormOpNoResize{\LinOpe}$.
Since $\NormOpNoResize{\LinOpe} < \PreconOpNormSP$, the parameters $\ParamPDS_{1}$ and $\ParamPDS_{2}$ in~\eqref{eq:SP_expriment} satisfy the inequality in~\eqref{eq:PDS_condition}.
We note that the parameter $\ParamPDS_{1}$ needs to be manually adjusted for accelerating the convergence of P-PDS.

\subsubsection{Row/Column Absolute Sum-Based Element-Wise Preconditioning (\ASEW)}
Let $\MatLinOpe_{\IndDual, \IndPrimal}$ be the representation matrix of $\LinOpe_{\IndDual, \IndPrimal}$. The authors of~\cite{DP-PDS} present a design method of constructing the preconditioners $\Precon_{1} = \mathrm{diag}(\Precon_{1,1}, \ldots, \Precon_{1,\NumPrimal})$ and $\Precon_{2} = \mathrm{diag}(\Precon_{2,1}, \ldots, \Precon_{2,\NumDual})$ as follows: 
\begin{align}
	\Precon_{1, \IndPrimal} = 
	& \mathrm{diag} \left(
	\frac{1}{\ConstASPSPrimal_{\IndPrimal,1}}, 
	\ldots, 
	\frac{1}{\ConstASPSPrimal_{\IndPrimal,\NumElemPrimal_{\IndPrimal}}} \right), 
	~ (\forall \IndPrimal = 1, \ldots, \NumPrimal), \nonumber \\
	\Precon_{2, \IndDual} 
	= & \mathrm{diag} \left(
	\frac{1}{\ConstASPSDual_{\IndDual,1}}, 
	\ldots,
	\frac{1}{\ConstASPSDual_{\IndDual,\NumElemDual_{\IndDual}}} \right),
	~ (\forall \IndDual = 1, \ldots, \NumDual),
	\label{eq:precon_ASEW}
\end{align}
where 
\begin{align}
	\ConstASPSPrimal_{\IndPrimal,\IndASPS} = 
	& \sum_{\IndDual = 1}^{\NumDual} \sum_{k = 1}^{\NumElemDual_{\IndDual}}|[\MatLinOpe_{\IndDual, \IndPrimal}]_{k, \IndASPS}|, 
	~ (\forall \IndASPS = 1, \ldots, \NumElemPrimal_{\IndPrimal}), \nonumber \\
	\ConstASPSDual_{\IndDual,\IndASPS} = 
	& \sum_{\IndPrimal = 1}^{\NumPrimal} \sum_{k = 1}^{\NumElemPrimal_{\IndPrimal}} |[\MatLinOpe_{\IndDual, \IndPrimal}]_{\IndASPS, k}|,
	~ (\forall \IndASPS = 1, \ldots, \NumElemDual_{\IndDual}).
\end{align}
Each $\Precon_{1,\IndPrimal}$ (or $\Precon_{2,\IndDual}$) is a diagonal matrix consisting of the row/column absolute sums of the elements of $\MatLinOpe_{\IndDual,\IndPrimal}$ (see~\cite[Lemma 2]{DP-PDS}). This means that the diagonal elements of one $\Precon_{1,\IndPrimal}$ (and $\Precon_{2,\IndDual}$) may take different values, i.e., the diagonal elements of the preconditioners will be different for each element for one variable in~\eqref{prob:general_form_of_optimization}.

\subsubsection{Positive-Definite Preconditioning (\SPDP)}
The authors in~\cite{SPDP} proposed to determine the preconditioners as
\begin{equation}
	\label{eq:precon_SPDP}
	\Precon_{1}
	= \ParamSPDP\mathbf{I}, 
	\Precon_{2} 
	= \frac{1}{\ParamSPDP} (\MatLinOpe\MatLinOpe^{\top} + \ParamSPDStheta\mathbf{I})^{-1},
\end{equation}
where $\MatLinOpe$ is the representation matrix of $\LinOpe$ and $\ParamSPDP > 0$ is a parameter.
Since the preconditioners in~\eqref{eq:precon_SPDP} are not block-diagonal matrices in general, P-PDS with them results in the procedures given in~\eqref{eq:P_PDS}.
 
If the number of dual variables is two ($\NumDual = 2$), the preconditioners are set as
\begin{equation}
	\label{eq:precon_SPDP_2}
	\Precon_{1} 
	= \frac{\ParamSPDP}{2}\mathbf{I}, \:
	\Precon_{2} = 
	\begin{bmatrix}
		\Precon_{2, 1} & \mathbf{O} \\
		\mathbf{O} & \Precon_{2, 2}
	\end{bmatrix},
\end{equation}
where
\begin{equation}
	\Precon_{2, \IndDual} 
	= \frac{1}{\ParamSPDP} \left(\sum_{\IndPrimal = 1}^{\NumPrimal} \MatLinOpe_{\IndDual, \IndPrimal}\MatLinOpe_{\IndDual, \IndPrimal}^{\top} + \ParamSPDStheta\mathbf{I}\right)^{-1}, ~ (\forall \IndDual = 1, 2).
\end{equation}
Since $\Precon_{1}$ and $\Precon_{2}$ in~\eqref{eq:precon_SPDP_2} are block-diagonal matrices, P-PDS with them can solve the Prob.~\eqref{prob:general_form_of_optimization} by the procedures given in~\eqref{eq:P_PDS_BD}.

We note that the parameters $\ParamSPDP$ and $\ParamSPDStheta$ affect the convergence speed of P-PDS. 
Therefore, the parameters $\ParamSPDP$ and $\ParamSPDStheta$ need to be manually adjusted.

\section{Proposed Operator Norm-Based Variable-Wise Diagonal Preconditioning (\ONVW)}
\label{sec:proposed}
This section is devoted to the establishment of a novel diagonal preconditioning method, \ONVW, for P-PDS.
First, we introduce a general form of our preconditioners as follows: for all $\IndPrimal=1, \ldots, \NumPrimal$ and $\IndDual=1,\ldots,\NumDual$
\begin{align}
	\label{eq:ONP_gen}
	\Precon_{1, \IndPrimal} & = \PreconSca_{1, \IndPrimal} \mathbf{I} = \frac{1}{\sum_{\IndDual=1}^{\NumDual}\PreconOpNorm_{\IndDual, \IndPrimal}^{2 - \ParamONP}}\mathbf{I}, \nonumber \\
	\Precon_{2,\IndDual} & = \PreconSca_{2, \IndDual} \mathbf{I} = \frac{1}{\sum_{\IndPrimal=1}^{\NumPrimal}\PreconOpNorm_{\IndDual,\IndPrimal}^{\ParamONP}}\mathbf{I},
	\: (\ParamONP \in [0, 2])
\end{align}
where each  $\PreconOpNorm_{\IndDual,\IndPrimal}$ is an upper bound of the operator norm of each $\LinOpe_{\IndDual,\IndPrimal}$, i.e., 
\begin{equation}
	\label{assump:gamma}
	\PreconOpNorm_{\IndDual,\IndPrimal}\in[\NormOp{\LinOpe_{\IndDual,\IndPrimal}}, \infty).
\end{equation} 
By changing the choice of $\ParamONP$, \ONVWs gives three design ways.
\begin{itemize}
	\item If we choose $\ParamONP = 0$, the preconditioners by \ONVWs (\ONVWOne) become
	\begin{equation}
		\label{eq:ONP1}
		\Precon_{1, \IndPrimal}=\frac{1}{\sum_{\IndDual=1}^{\NumDual}\PreconOpNorm_{\IndDual, \IndPrimal}^{2}}\mathbf{I}, \: \Precon_{2,\IndDual}=\frac{1}{\NumPrimal}\mathbf{I}.
	\end{equation}
	\item If we choose $\ParamONP = 1$, the preconditioners by \ONVWs (\ONVWTwo) become
	\begin{equation}
		\label{eq:ONP2}
		\Precon_{1,\IndPrimal} = \frac{1}{\sum_{\IndDual=1}^{\NumDual}\PreconOpNorm_{\IndDual, \IndPrimal}}\mathbf{I}, \:
		\Precon_{2,\IndDual} = \frac{1}{\sum_{\IndPrimal=1}^{\NumPrimal}\PreconOpNorm_{\IndDual,\IndPrimal}}\mathbf{I}. 
	\end{equation}
	\item If we choose $\ParamONP = 2$, the preconditioners by \ONVWs (\ONVWThree) become
	\begin{equation}
		\label{eq:ONP3}
		\Precon_{1,\IndPrimal} = \frac{1}{\NumDual}\mathbf{I}, \:
		\Precon_{2,\IndDual} = \frac{1}{\sum_{\IndPrimal=1}^{\NumPrimal}\PreconOpNorm_{\IndDual,\IndPrimal}^{2}}\mathbf{I}, 
	\end{equation}
\end{itemize}

\begin{remark}[Two Features of Our Method] \hfill
	\begin{itemize}
		\item Our preconditioners can be calculated by only using (upper bounds of) the operator norms of the linear operators $\LinOpe_{\IndDual,\IndPrimal}$. This implies that \ONVWs does not require direct access to the elements of the explicit matrices representing $\LinOpe_{\IndDual,\IndPrimal}$ as long as some $\PreconOpNorm_{\IndPrimal,\IndDual}$ are available. 
		\item In addition, the diagonal elements of one $\Precon_{1,\IndPrimal}$ take the same value ($\Precon_{2,\IndDual}$ as well), i.e., our method is a variable-wise preconditioning method, which maintains the proximability of the functions in Prob.~\eqref{prob:general_form_of_optimization}.
	\end{itemize}
\end{remark}

Before showing the convergence theorem of P-PDS with \ONVWs defined in~\eqref{eq:ONP_gen}, we give the following lemma on matrix decomposition.

\begin{algorithm}[t]
	\caption{P-PDS with \ONVWs for solving~\eqref{prob:general_form_of_optimization}}
	\label{algo:P_PDS}
	\begin{algorithmic}[1]
		\Require{$\VarPrimal_{1}^{(0)},\ldots,\VarPrimal_{\NumPrimal}^{(0)},\VarDual_{1}^{(0)},\ldots,\VarDual_{\NumDual}^{(0)}$ }
		\Ensure{$\VarPrimal_{1}^{(\InnerIter)},\ldots,\VarPrimal_{\NumPrimal}^{(\InnerIter)},\VarDual_{1}^{(\InnerIter)},\ldots,\VarDual_{\NumDual}^{(\InnerIter)}$}
		\State Initialize $\InnerIter = 0$;
		\State Set $\Precon_{1,1},\ldots,\Precon_{1,\NumPrimal},\Precon_{2,1},\ldots,\Precon_{2,\NumDual}$ as in~\eqref{eq:ONP_gen};
		\While {A stopping criterion is not satisfied}
		
		\For{$\IndPrimal=1,\cdots,\NumPrimal$}
		
		\State $\VarPrimal_{\IndPrimal}^{\prime} \leftarrow \sum_{\IndDual = 1}^{\NumDual}\LinOpe_{\IndDual,\IndPrimal}^{*}(\VarDual_{\IndDual}^{(\InnerIter)})$;
		\State $\VarPrimal_{\IndPrimal}^{(\InnerIter + 1)} \leftarrow \prox_{\Precon_{1,\IndPrimal}^{-1}, \FuncPrimal_{\IndPrimal}} (\VarPrimal_{\IndPrimal}^{(\InnerIter)}-\Precon_{1,\IndPrimal}\VarPrimal_{\IndPrimal}^{\prime})$;
		
		\EndFor
		
		\For{$\IndDual = 1, \cdots, \NumDual$}
		\State $\VarDual_{\IndDual}^{\prime} \leftarrow \sum_{\IndPrimal = 1}^{\NumPrimal} \LinOpe_{\IndDual,\IndPrimal}(2\VarPrimal_{\IndPrimal}^{(\InnerIter + 1)} - \VarPrimal_{\IndPrimal}^{(\InnerIter)})$;
		\State $\VarDual_{\IndDual}^{(\InnerIter+1)} \leftarrow \prox_{\Precon_{2,\IndDual}^{-1}, \FuncDual_{\IndDual}^{*}}(\VarDual_{\IndDual}^{(\InnerIter)}+\Precon_{2,\IndDual}\VarDual_{\IndDual}^{\prime})$;
		\EndFor
		\State $\InnerIter \leftarrow \InnerIter + 1$;
		\EndWhile
	\end{algorithmic}
\end{algorithm}

\begin{lemma}
	\label{lem:matrix_decomp}
	An arbitrary matrix $\mathbf{A}\in\RealNumSet^{m\times n}$ can be decomposed into matrices $\mathbf{B}$ and $\mathbf{C}$ (i.e., $\mathbf{A} = \mathbf{BC}$) that satisfy for any $\ParamONP \in [0,1]$
	\begin{align}
		\label{prpty:opnorm}
		& \NormOp{\mathbf{B}} = \NormOp{\mathbf{A}}^{1 - \ParamONP} (= \sigma_{1}(\mathbf{A})^{1 - \ParamONP}), \nonumber \\ 
		& \NormOp{\mathbf{C}} = \NormOp{\mathbf{A}}^{\ParamONP} (= \sigma_{1}(\mathbf{A})^{\ParamONP}).
	\end{align}
\end{lemma}

The proof is in Appendix.

Then, the following theorem guarantees the convergence of P-PDS with \ONVW.

\begin{theorem}
	\label{theo:new-preconditioner_gen}
	If the preconditioners are set as~\eqref{eq:ONP_gen}, then the following inequality holds: 
	\begin{equation}
		\label{eq:convergence_equation_ours}
		\NormOp{\Precon_{2}^{\frac{1}{2}}\circ\LinOpe\circ\Precon_{1}^{\frac{1}{2}}}^{2} \leq 1.
	\end{equation}
\end{theorem}

\begin{proof}
	Since $\Precon_{1}$ and $\Precon_{2}$ are positive-definite and diagonal, their powers of one-half are 
	\begin{align}
		\Precon_{1}^{\frac{1}{2}}=\mathrm{diag}\left(\Precon_{1,1}^{\frac{1}{2}},\ldots,\Precon_{1,\NumPrimal}^{\frac{1}{2}}\right), \nonumber \\
		\Precon_{2}^{\frac{1}{2}}=\mathrm{diag}\left(\Precon_{2,1}^{\frac{1}{2}}\ldots,\Precon_{2,\NumDual}^{\frac{1}{2}}\right).
		\label{eq:half_power_gen}
	\end{align}
	By matrix multiplication and Eq.~\eqref{eq:half_power_gen}, we have
	\begin{equation}
			\label{eq:transform_step_1_gen}
			\Precon_{2}^{\frac{1}{2}}\circ\LinOpe\circ\Precon_{1}^{\frac{1}{2}} = \left[\,\Precon_{2,\IndDual}^{\frac{1}{2}}\circ\mathfrak{L}_{\IndDual,\IndPrimal}\circ\Precon_{1,\IndPrimal}^{\frac{1}{2}}\,\right]_{1 \leq \IndPrimal \leq \NumPrimal, 1 \leq \IndDual \leq \NumDual}. 
		\end{equation}
	For all $\VarPrimal = [\VarPrimal_{1}^{\top}, \ldots, \VarPrimal_{\NumPrimal}^{\top}]^{\top} \in \RealNumSet^{\NumElemPrimalAll}$, the triangle inequality yields
	\begin{equation}
		\left\| \Precon_{2}^{\frac{1}{2}}\circ\LinOpe\circ\Precon_{1}^{\frac{1}{2}} \VarPrimal \right\|_{2}^{2} 
		\leq  \sum_{\IndDual=1}^{\NumDual} \sum_{\IndPrimal=1}^{\NumPrimal}  \left\| \Precon_{2,\IndDual}^{\frac{1}{2}}\circ\LinOpe_{\IndDual,\IndPrimal}\circ\Precon_{1,\IndPrimal}^{\frac{1}{2}} \VarPrimal_{\IndPrimal} \right\|_{2}^{2} .
		\label{eq:transform_step_2_gen}
	\end{equation}
	Since $\LinOpe_{\IndDual,\IndPrimal}$ $(\IndPrimal = 1, \ldots, \NumPrimal, \IndDual = 1, \ldots, \NumDual)$ can be represented by matrices, from Lemma~\ref{lem:matrix_decomp}, there exist linear operators $\LinOpe_{\IndDual,\IndPrimal}^{\tfrac{\ParamONP}{2}}$ and  $\LinOpe_{\IndDual,\IndPrimal}^{1 - \tfrac{\ParamONP}{2}}$ that satisfy for any $\ParamONP \in [0,2]$,
	\begin{align}
		\LinOpe_{\IndDual,\IndPrimal} 
		& =  \LinOpe_{\IndDual,\IndPrimal}^{1 - \tfrac{\ParamONP}{2}} \circ \LinOpe_{\IndDual,\IndPrimal}^{\tfrac{\ParamONP}{2}}, \nonumber \\
		\NormOpNoResize{\LinOpe_{\IndDual,\IndPrimal}^{1 - \tfrac{\ParamONP}{2}}} 
		& = \NormOp{\LinOpe_{\IndDual,\IndPrimal}}^{1 - \tfrac{\ParamONP}{2}}, \nonumber \\
		\NormOpNoResize{\LinOpe_{\IndDual,\IndPrimal}^{\tfrac{\ParamONP}{2}}} 
		& = \NormOp{\LinOpe_{\IndDual,\IndPrimal}}^{\tfrac{\ParamONP}{2}}.
		\label{eq:transform_step_3_gen}
	\end{align}
	Thus, it follows, from Eq.~\eqref{eq:transform_step_3_gen} and the definition and the submultiplicity of operator norms, that
	\begin{align}
		\mathrm{Eq.}~(32) 
		& = \sum_{\IndDual=1}^{\NumDual} \sum_{\IndPrimal=1}^{\NumPrimal} \left\| \Precon_{2,\IndDual}^{\frac{1}{2}} \circ \LinOpe_{\IndDual,\IndPrimal}^{1 - \tfrac{\ParamONP}{2}} \circ \LinOpe_{\IndDual,\IndPrimal}^{\tfrac{\ParamONP}{2}} \circ \Precon_{1,\IndPrimal}^{\frac{1}{2}} \VarPrimal_{\IndPrimal} \right\|_{2}^{2}\nonumber \\
		& \leq \sum_{\IndDual=1}^{\NumDual} \PreconSca_{2, \IndDual} \sum_{\IndPrimal=1}^{\NumPrimal} \PreconSca_{1, \IndPrimal} \NormOp{\LinOpe_{\IndDual,\IndPrimal}}^{2 - \ParamONP} \NormOp{\LinOpe_{\IndDual,\IndPrimal}}^{\ParamONP} \NormET{\VarPrimal_{\IndPrimal}}^{2}.
		\label{eq:transform_step_4_gen}
	\end{align}
	By applying the inequality $\sum_{\IndDual = 1}^{\NumDual} x_{\IndDual}^{2}  \leq (\sum_{\IndDual = 1}^{\NumDual} x_{\IndDual})^{2}$ for any positive real numbers $x_{1}, \ldots, x_{\NumDual}$ and the Cauchy-Schwarz inequality to the right hand side of Eq.~\eqref{eq:transform_step_4_gen}, we obtain
	\begin{align}
		& \mathrm{Eq.}~\eqref{eq:transform_step_4_gen} \nonumber \\
		& \leq 
		\sum_{\IndDual=1}^{\NumDual} \PreconSca_{2, \IndDual} 
		\left( \sum_{\IndPrimal=1}^{\NumPrimal} \sqrt{\PreconSca_{1, \IndPrimal}} \NormOp{\LinOpe_{\IndDual,\IndPrimal}}^{1 - \tfrac{\ParamONP}{2}} \NormOp{\LinOpe_{\IndDual,\IndPrimal}}^{\tfrac{\ParamONP}{2}} \NormET{\VarPrimal_{\IndPrimal}} \right)^{2}  \nonumber \\
		& \leq 
		\sum_{\IndDual=1}^{\NumDual} \PreconSca_{2, \IndDual}
		\left(\sum_{\IndPrimal=1}^{\NumPrimal}  \NormOp{\LinOpe_{\IndDual,\IndPrimal}}^{\ParamONP}\right)
		\left(\sum_{\IndPrimal=1}^{\NumPrimal} \PreconSca_{1, \IndPrimal} \NormOp{\LinOpe_{\IndDual,\IndPrimal}}^{2 - \ParamONP} \NormET{\VarPrimal_{\IndPrimal}}^{2} \right).
		\label{eq:transform_step_5_gen}
	\end{align}
	Then, from the definitions of $\PreconSca_{2, \IndDual}$ and $\PreconSca_{1, \IndPrimal}$ in~\eqref{eq:ONP_gen}, we have $\PreconSca_{2, \IndDual} \sum_{\IndPrimal=1}^{\NumPrimal}\NormOp{\LinOpe_{\IndDual,\IndPrimal}}^{\ParamONP}\leq 1$ for any $\IndDual=1,\ldots,\NumDual$ and $\PreconSca_{1, \IndPrimal} \sum_{\IndDual=1}^{\NumDual}\NormOp{\LinOpe_{\IndDual,\IndPrimal}}^{2 - \ParamONP} \leq 1$ for any $\IndPrimal=1,\ldots,\NumPrimal$, which yields 
	\begin{align}
		\mathrm{Eq.}~\eqref{eq:transform_step_5_gen} 
		& \leq  \sum_{\IndDual=1}^{\NumDual} \sum_{\IndPrimal=1}^{\NumPrimal}  \PreconSca_{1, \IndPrimal} \NormOp{\LinOpe_{\IndDual,\IndPrimal}}^{2 - \ParamONP} \NormET{\VarPrimal_{\IndPrimal}}^{2} \nonumber \\
		& = \sum_{\IndPrimal=1}^{\NumPrimal}  \PreconSca_{1, \IndPrimal}\left( \sum_{\IndDual=1}^{\NumDual} \NormOp{\LinOpe_{\IndDual,\IndPrimal}}^{2 - \ParamONP}\right) \NormET{\VarPrimal_{\IndPrimal}}^{2}\nonumber \\
		& \leq \sum_{\IndPrimal=1}^{\NumPrimal}  \NormET{\VarPrimal_{\IndPrimal}}^{2} 
		= \NormET{\VarPrimal}^{2}.
		\label{eq:transform_step_6}
	\end{align}
	Therefore, we finally obtain
	\begin{equation*}
		\NormOp{\Precon_{2}^{\frac{1}{2}}\circ\LinOpe\circ\Precon_{1}^{\frac{1}{2}}}^{2} 
		= \sup_{\mathbf{x} \neq \ZeroElem} \frac{\NormETNoResize{\Precon_{2}^{\frac{1}{2}}\circ\LinOpe\circ\Precon_{1}^{\frac{1}{2}} \VarPrimal}^{2} }{\|\mathbf{x}\|_{2}^{2}}
		\leq \frac{\|\mathbf{x}\|_{2}^{2}}{\|\mathbf{x}\|_{2}^{2}}
		= 1.
	\end{equation*}
	\begin{flushright} $\square$ \end{flushright}
\end{proof}

\begin{remark}
To guarantee the convergence of P-PDS, inequality~\eqref{eq:convergence_equation} has to be strict, but inequality~\eqref{eq:convergence_equation_ours} is not.
However, we do not observe any convergence issue of P-PDS with our preconditioners in the experiments (see Section~IV). 
This is because, our method separates $\LinOpe$ variable by variable and sums up upper bounds of the operator norms, resulting in setting preconditioners such that $\NormOpNoResize{\Precon_{2}^{\frac{1}{2}}\circ\LinOpe\circ\Precon_{1}^{\frac{1}{2}}}^{2} < 1$ in almost all real-world applications.
\end{remark}

Theorem~\ref{theo:new-preconditioner_gen} asserts that our preconditioners defined in~\eqref{eq:ONP1},~\eqref{eq:ONP2}, and~\eqref{eq:ONP3} satisfy the convergence condition of P-PDS in~\eqref{eq:convergence_equation}. Therefore, P-PDS with \ONVWs generates sequences that converge to an optimal solution of Prob.~\eqref{prob:general_form_of_optimization}. 

Here, each $\PreconOpNorm_{j,i}$ is determined in the following manner.
{
	\setlength{\leftmargini}{15pt}         
	\begin{itemize}
		\setlength{\parskip}{2pt}      
		\item If the operator norm $\NormOp{\LinOpe_{\IndDual,\IndPrimal}}$ is known, we set  $\PreconOpNorm_{\IndDual,\IndPrimal}$ to $\NormOp{\LinOpe_{\IndDual,\IndPrimal}}$.
		\item If $\NormOp{\LinOpe_{\IndDual,\IndPrimal}}$ is unknown, we set $\PreconOpNorm_{\IndDual,\IndPrimal}$ to some known or computatble upper bound of $\NormOp{\LinOpe_{\IndDual,\IndPrimal}}$.
		\item If the linear operator is the composition of two linear operators $\mathfrak{A}$ and $\mathfrak{B}$ whose operator norms (or their upper bounds) are known $(\NormOp{\mathfrak{A}} \leq \alpha_{\mathfrak{A}}, \NormOp{\mathfrak{B}} \leq \alpha_{\mathfrak{B}})$, we set $\PreconOpNorm_{\IndDual,\IndPrimal}$ to $\alpha_{\mathfrak{A}}\alpha_{\mathfrak{B}}$, which is an upper bound of $\NormOp{\mathfrak{A} \circ \mathfrak{B}}$ due to the submultiplicity in~\eqref{eq:submultiplicity}.
	\end{itemize}
}

Finally, we show the detailed procedures of P-PDS with \ONVWs in Algorithm~\ref{algo:P_PDS}.

\begin{table}[t]
	\begin{center}
		\caption{Features of Existing Methods \\ and Our Method (Highlighted in Bold).}
		\label{tab:existing_methods}
		\vspace{-2mm}
		\scalebox{0.95}{
		\begin{tabular}{cccc}
			\toprule
			\multirow{2}{*}{Methods} & Parameters requiring & Maintaining & Avoiding access to\\
			& manual adjustment & proximability & representation matrices\\
			\cmidrule(lr){1-1}\cmidrule(lr){2-4}
			\SP~\cite{PDS_1} & $\ParamPDS_{1}$ & $\checkmark$ & $\checkmark$ \\
			\ASEW~\cite{DP-PDS} & None. & $\times$ & $\times$  \\
			\SPDP~\cite{SPDP} & $\ParamSPDP$ & $\times$ & $\checkmark$  \\
			\textbf{\ONVWOne} & None. & $\checkmark$ & $\checkmark$  \\
			\textbf{\ONVWTwo} & None. & $\checkmark$ & $\checkmark$  \\
			\textbf{\ONVWThree} & None. & $\checkmark$ & $\checkmark$  \\
			\bottomrule
		\end{tabular}
	}
	\end{center}
	\vspace{-3mm}
\end{table}

\begin{table}[t]
	\begin{center}
		\caption{Stopping Criteria.}
		\label{tab:convergence_conditions}
			\vspace{-2mm}
			\begin{tabular}{cc}
				\toprule
				Applications & Stopping criteria \\
				\cmidrule(lr){1-1}\cmidrule(lr){2-2}
				Mixed noise removal & $\mathrm{RMSE} < 0.005$ \\
				Unmixing & $\mathrm{RMSE} < 0.01$ \\ 
				Graph signal recovery & $\mathrm{RMSE} < 0.001$ \\
				\bottomrule
			\end{tabular}
	\end{center}
	\vspace{-3mm}
\end{table}

\begin{figure*}[!t]
	\begin{center}
		\begin{minipage}{0.97\hsize}
			\centerline{\includegraphics[width=\hsize]{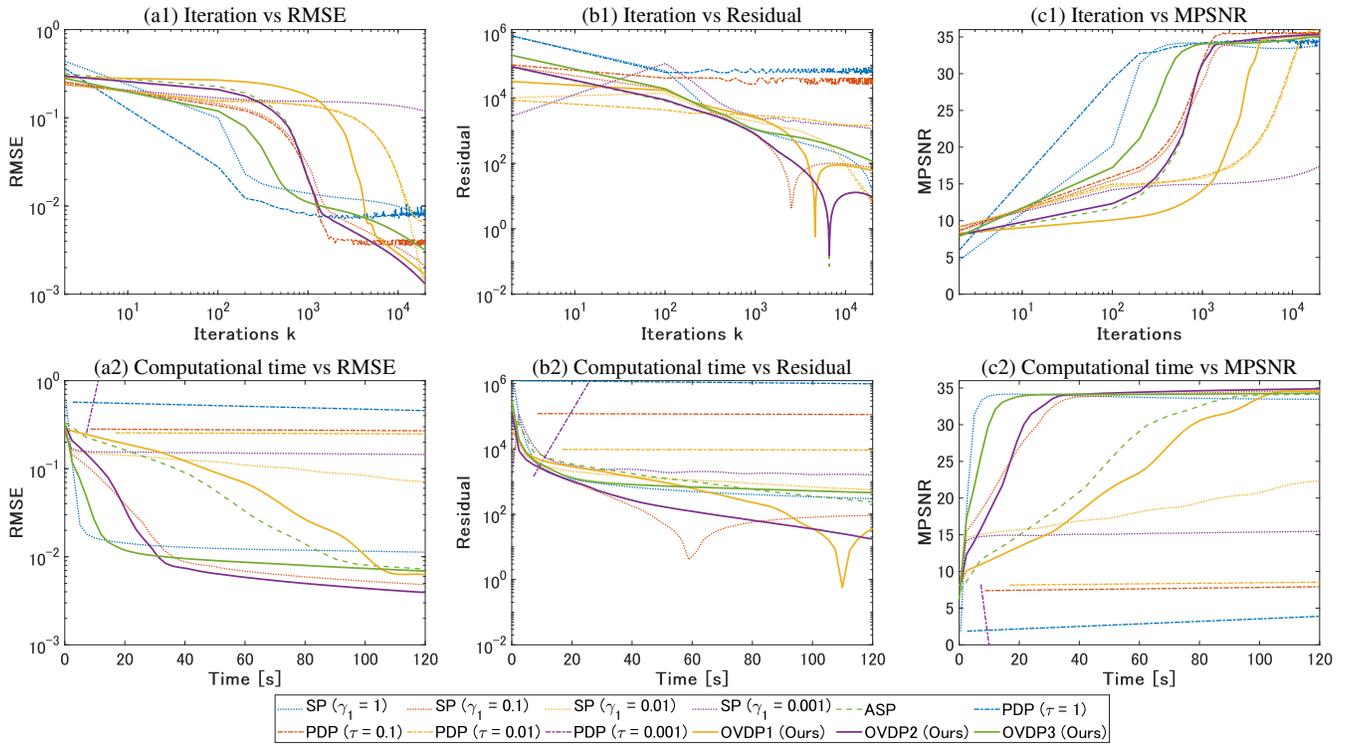}} 
		\end{minipage}
	\end{center}
	
	\vspace{-2mm}
	
	\caption{Convergence profiles of the mixed noise removal experiments. (a): Iterations/computational time versus RMSE. (b): Iterations/computational time versus Residual. (c): Iterations/computational time versus MPSNR. Note that applying P-PDS with \ASEWs (green dotted line) to Prob.~\eqref{prob:sstv_restoration_formulation} is not practical in terms of implementation (the linear operators $\DiffvSymb$, $\DiffhSymb$, and $\DiffbSymb$ are not usually implemented as explicit matrices).}
	\label{fig:graph_MNR}
	\vspace{-2mm}
\end{figure*}

\section{Experiments and Discussion}
\label{sec:experiments}
In this section, we apply our \ONVWs to three signal estimation problems: mixed noise removal of hyperspectral images, hyperspectral unmixing, and graph signal recovery.
Through these applications, we illustrate the effectiveness and usefulness of our method as follows:
\begin{itemize}
	\item P-PDS with \ONVWs is fast on average to obtain an optimal solution of the target optimization problem.
	\item The preconditioners by \ONVWs can be easily calculated by using operator norms even if the target optimiztion problem involves linear operators implemented not as explicit matrices.
	\item P-PDS with \ONVWs is efficiently computed by avoiding the computations of skewed proximity operators.
\end{itemize}

\subsection{Experimental Setup}
We compared \ONVWs with three existing preconditioner design methods (see Tab.~\ref{tab:existing_methods}): the Scalar Preconditioning (\SP)~\cite{PDS_1} in~\eqref{eq:precon_SP}, the row/column Absolute Sum-based element-wise Preconditioning (\ASEW)~\cite{DP-PDS} in~\eqref{eq:precon_ASEW}, and the Positive-Definite Preconditioning (\SPDP)~\cite{SPDP} in~\eqref{eq:precon_SPDP} and in~\eqref{eq:precon_SPDP_2}.
Note that the preconditioners by \SPs and \SPDPs have parameters ($\ParamPDS_{1}$, $\ParamSPDP$, $\ParamSPDStheta$) to be adjusted manually.
For \SP , we set $\ParamPDS_{1}$ and $\ParamPDS_{2}$ in~\eqref{eq:precon_SP} as $\ParamPDS_{1} = 1, 0.1, 0.01, 0.001$, and as in~\eqref{eq:SP_expriment}.
The parameter $\ParamSPDP$ in~\eqref{eq:precon_SPDP} and in~\eqref{eq:precon_SPDP_2} was set as $\ParamSPDP = 1, 0.1, 0.01, 0.001$.
The parameter $\ParamSPDStheta$ in~\eqref{eq:precon_SPDP} and in~\eqref{eq:precon_SPDP_2} was set as $\ParamSPDStheta = 0.01$, which is recommended in~\cite{SPDP}.
To calculate skewed proximity operators in the iterations of P-PDSs with \ASEWs and \SPDP, we used the Fast Iterative Shrinkage-Thresholding Algorithm (FISTA)~\cite{FISTA} initialized with a zero vector.

To check the convergence of P-PDS, we used the Root Mean Square Error (RMSE):
\begin{equation}
	\label{eq:NRMSE}
	\mbox{RMSE}(\VarPrimal_{1}^{(\InnerIter)},\ldots,\VarPrimal_{\NumPrimal}^{(\InnerIter)}):=
	\sqrt{\frac{\sum_{\IndPrimal=1}^{\NumPrimal}\|\VarPrimal_{\IndPrimal}^{(\InnerIter)} - \VarPrimal_{\IndPrimal}^{*}\|_{2}^{2} }
	{\sum_{\IndPrimal=1}^{\NumPrimal} \NumElemPrimal_{\IndPrimal}}},
\end{equation}
and the residual of the function values:
\begin{align}
	& \mbox{Residual}(\VarPrimal_{1}^{(\InnerIter)},\ldots,\VarPrimal_{\NumPrimal}^{(\InnerIter)}) \nonumber \\
	& :=  \left| \left( \sum_{\IndPrimal=1}^{\NumPrimal}\FuncPrimal_{\IndPrimal}(\VarPrimal_{\IndPrimal}^{(\InnerIter)}) 
	+ \sum_{\IndDual=1}^{\NumDual}\FuncDual_{\IndDual} \left(\sum_{\IndPrimal=1}^{\NumPrimal}\LinOpe_{\IndDual,\IndPrimal}(\VarPrimal_{\IndPrimal}^{(\InnerIter)}) \right) \right) \right. \nonumber \\ 
	& \left.
	- \left( \sum_{\IndPrimal=1}^{\NumPrimal}\FuncPrimal_{\IndPrimal}(\VarPrimal_{\IndPrimal}^{*}) 
	+ \sum_{\IndDual=1}^{\NumDual}\FuncDual_{\IndDual} \left(\sum_{\IndPrimal=1}^{\NumPrimal}\LinOpe_{\IndDual,\IndPrimal}(\VarPrimal_{\IndPrimal}^{*})\right)\right) \right|,
	\label{eq:function_diff}
\end{align}
where $\VarPrimal_{1}^{*},\ldots,\VarPrimal_{\NumPrimal}^{*}$ are oracle solutions. 
However, such oracle solutions are not available in the experiments, and therefore, we generated pseudo-oracle solutions by the following procedures. 
We calculated the results through $100,000$ iterations of P-PDS with all the methods in advance, and then selected the best ones among them.

Tab.~\ref{tab:convergence_conditions} shows the stopping criteria with RMSE as the threshold used in the experiments.
Since convergence speeds are different depending on problems, reasonable criteria are also different. 
To determine reasonable criteria, we employed normalized error ($\|\mathbf{x}^{(\InnerIter + 1)} - \mathbf{x}^{(\InnerIter)}\|_{2}/\|\mathbf{x}^{(\InnerIter)}\|_{2}$), which is often used as stopping criteria in real-world applications. Based on the normalized error, we set the stopping criteria as the RMSE values such that $\|\mathbf{x}^{(\InnerIter + 1)} - \mathbf{x}^{(\InnerIter)}\|_{2}/\|\mathbf{x}^{(\InnerIter)}\|_{2} < 10^{-5}$. 

\begin{figure*}[!t]
	\begin{center}
		\begin{minipage}{0.20\hsize}
			\centerline{\includegraphics[width=\hsize]{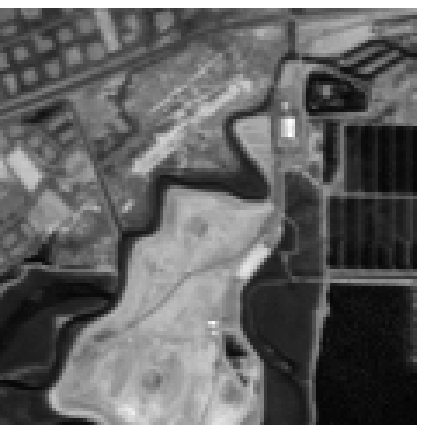}} 
		\end{minipage}
		\begin{minipage}{0.20\hsize}
			\centerline{\includegraphics[width=\hsize]{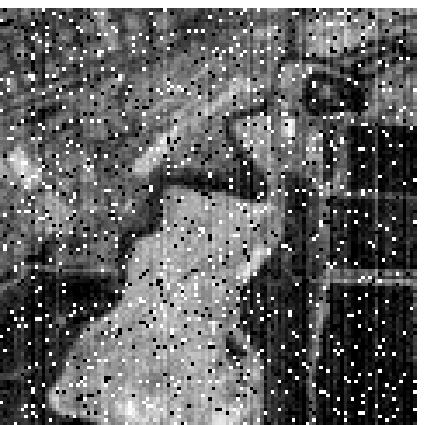}} 
		\end{minipage}
		\begin{minipage}{0.20\hsize}
			\centerline{\includegraphics[width=\hsize]{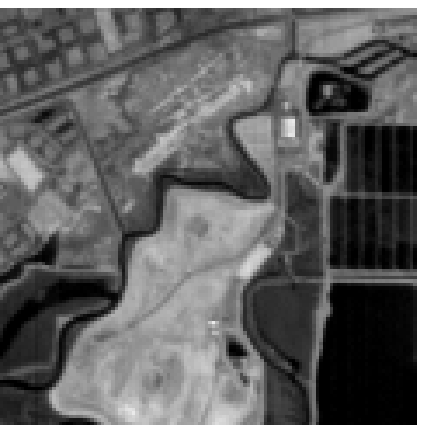}} 
		\end{minipage}
		\begin{minipage}{0.20\hsize}
			\centerline{\includegraphics[width=\hsize]{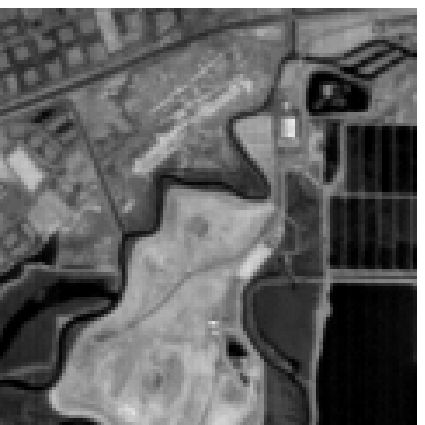}} 
		\end{minipage}
		
		\vspace{1mm}
		
		\begin{minipage}{0.20\hsize}
			\centerline{(a)}
		\end{minipage}
		\begin{minipage}{0.20\hsize}
			\centerline{(b) MPSNR=$14.37$ [dB]}
		\end{minipage}
		\begin{minipage}{0.20\hsize}
			\centerline{(c) MPSNR=$34.60$ [dB]}
		\end{minipage}
		\begin{minipage}{0.20\hsize}
			\centerline{(d) MPSNR=$34.62$ [dB]}
		\end{minipage}
		
		\vspace{1mm}

		\begin{minipage}{0.20\hsize}
			\centerline{\includegraphics[width=\hsize]{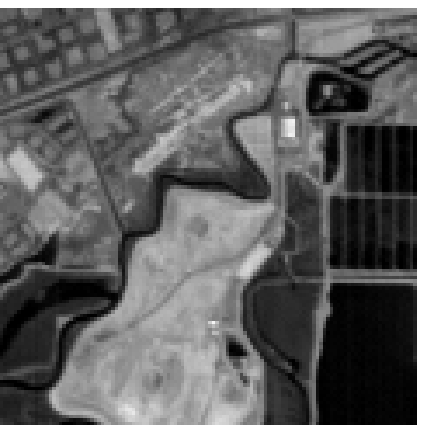}} 
		\end{minipage}
		\begin{minipage}{0.20\hsize}
			\centerline{\includegraphics[width=\hsize]{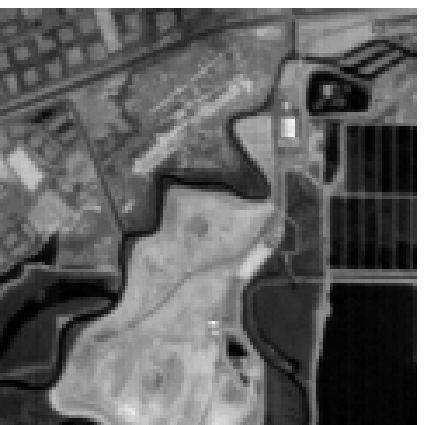}} 
		\end{minipage}
		\begin{minipage}{0.20\hsize}
			\centerline{\includegraphics[width=\hsize]{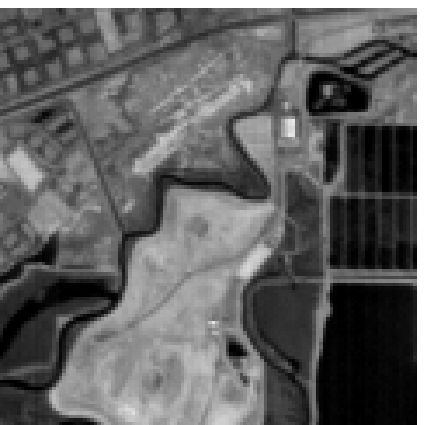}} 
		\end{minipage}
		\begin{minipage}{0.20\hsize}
			\centerline{\includegraphics[width=\hsize]{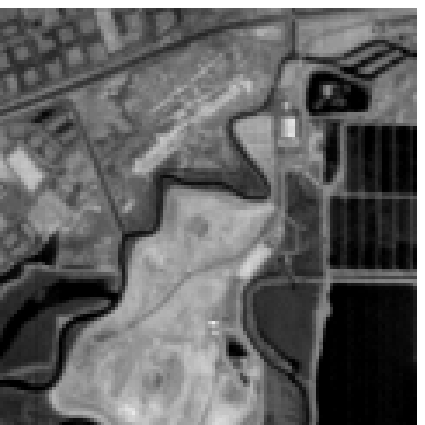}} 
		\end{minipage}
		
		\vspace{1mm}
		
		\begin{minipage}{0.20\hsize}
			\centerline{(e) MPSNR=$35.45$ [dB]}
		\end{minipage}
		\begin{minipage}{0.20\hsize}
			\centerline{(f) MPSNR=$34.81$ [dB]}
		\end{minipage}
		\begin{minipage}{0.20\hsize}
			\centerline{(g) MPSNR=$34.66$ [dB]}
		\end{minipage}	
		\begin{minipage}{0.20\hsize}
			\centerline{(h) MPSNR=$34.62$ [dB]}
		\end{minipage}

	\end{center}

	\vspace{-4mm}
	
	\caption{Mixed noise removal results. (a): The ground truth HS image. (b): The observed HS image. (c): The HS image estimated by P-PDS with \SP~\cite{PDS_1} ($\ParamPDS_{1} = 0.1$). (d): The HS image estimated by P-PDS with \ASEW~\cite{DP-PDS}. (e): The HS image estimated by P-PDS with \SPDP~\cite{SPDP} ($\ParamSPDP = 0.1$). (f): The HS image estimated by P-PDS with \ONVWOnes (Ours). (g): The HS image estimated by P-PDS with \ONVWTwos (Ours). (h): The HS image estimated by P-PDS with \ONVWThrees (Ours).}
	\label{fig:results_MNR}
\end{figure*}

\subsection{Application to Mixed Noise Removal of Hyperspectral Images}
Hyperspectral (HS) images often suffer from various noises, such as random noise, outliers, missing values, and stripe noise, due to environmental and sensor issues~\cite{MNR_LFRR_Zheng2020,MNR_LRTF_Zheng2020,MNR_Zhang2022}.
These noises seriously degrade the performance of subsequent processing, such as HS unmixing~\cite{ghamisi2017advances}, classification~\cite{HSClassification_review_Audebert2019}, and anomaly detection~\cite{HSAnomaly_review_Su2022}.
Therefore, removing mixed noise from HS images is a crucial preprocessing.
Popular mixed noise removal methods adopt the Spatio-Spectral Total Variation (SSTV) regularization~\cite{SSTV,SSTV_LRTD,LRMRSSTV,GLSSTV,SSTV_hybrid,FC_naganuma_2022,GSSTV_Takemoto2022}, which models the spatial piecewise-smoothness and the spectral correlations of HS images.

\subsubsection{Problem Formulation}
Consider that an observed HS image (of size $\NumVerPixMNR \times \NumHorPixMNR \times \NumBandMNR$) $\mathbf{v} \in \mathbb{R}^{\NumVerPixMNR \NumHorPixMNR \NumBandMNR}$ is given by 
\begin{equation}
	\label{eq:observation_model}
	\ObsMNR = \bar{\GTMNR} + \bar{\SparseMNR} + \bar{\StripeMNR} + \NoiseMNR,
\end{equation}
where $\bar{\GTMNR}$, $\bar{\SparseMNR}$, $\bar{\StripeMNR}$, and $\mathbf{n}$ are the true HS image of interest, sparsely distributed noise (e.g. outliers and missing values), stripe noise, and random noise, respectively. 
Based on this observation model, the SSTV-regularized mixed noise removal problem is formulated as the following convex optimization problem:
\begin{align}
	\min_{\GTMNR, \SparseMNR, \StripeMNR} \: & 
	\|\Diffv{\Diffb{\GTMNR}}\|_{1}
	+ \|\Diffh{\Diffb{\GTMNR}}\|_{1}
	+ \ParamBalanceMNR\|\mathbf{\StripeMNR}\|_{1} \nonumber \\
	\st \: &
		\Diffv{\StripeMNR} = \ZeroElem,
		\SparseMNR \in \BallSparMNR, 
		\GTMNR + \SparseMNR + \StripeMNR \in \BallFidelMNR, 
	\label{prob:sstv_restoration_formulation} 
\end{align}
where $\DiffvSymb$, $\DiffhSymb$, and $\DiffbSymb$ are the vertical, horizontal, and spectral difference operators, respectively, with the Neumann boundary condition. 
To reduce computing resources, these difference operators are usually implemented not as matrices but as the following procedures:
\begin{equation}
	[\DiffvSymb(\mathbf{x})]_{i, j, k} 
	:=
	\begin{cases}
		[\mathbf{x}]_{i, j, k} - [\mathbf{x}]_{i + 1, j, k}, & \mathrm{if} \: i < \NumVerPixMNR; \\
		0, & \mathrm{otherwise},
	\end{cases}
\end{equation}
\begin{equation}
	[\DiffhSymb(\mathbf{x})]_{i, j, k} 
	:=
	\begin{cases}
		[\mathbf{x}]_{i, j, k} - [\mathbf{x}]_{i, j + 1, k}, & \mathrm{if} \: j < \NumHorPixMNR; \\
		0, & \mathrm{otherwise},
	\end{cases} 
\end{equation}
\begin{equation}
	[\DiffbSymb(\mathbf{x})]_{i, j, k} 
	:=
	\begin{cases}
		[\mathbf{x}]_{i, j, k} - [\mathbf{x}]_{i, j, k + 1}, & \mathrm{if} \: k < \NumBandMNR; \\
		0, & \mathrm{otherwise},
	\end{cases}
\end{equation}
where $[\mathbf{x}]_{i_{1},i_{2},i_{3}}$ is the value of $\mathbf{x}$ at a location $(i_{1},i_{2},i_{3})$.
Here, $\|\cdot\|_{1}$ is the $\ell_{1}$ norm, and $\BallFidelMNR$ and $\BallSparMNR$ are the $\ell_{2}$ and $\ell_{1}$ norm balls, respectively given by
\begin{align}
	\BallFidelMNR & 
	:= 
	\left\{\mathbf{x}\in\mathbb{R}^{\NumVerPixMNR\NumHorPixMNR\NumBandMNR} \: \middle| \: \| \ObsMNR - \mathbf{x}\|_{2}\leq\varepsilon\right\}, \nonumber \\
	\BallSparMNR & 
	:= 
	\left\{\mathbf{x}\in\mathbb{R}^{\NumVerPixMNR\NumHorPixMNR\NumBandMNR} \: \middle| \: \| \mathbf{x} \|_{1} \leq \ParamSparMNR \right\}.
\end{align}
The term $\|\Diffv{\Diffb{\GTMNR}}\|_{1} + \|\Diffh{\Diffb{\GTMNR}}\|_{1}$ is the SSTV regularization. 
The positive value $\ParamBalanceMNR$ is a balancing parameter between the SSTV regularization and the sparse noise term. The hard constraint guarantees the $\ell_{2}$ data-fidelity to $\ObsMNR$ with the radius $\ParamFidelMNR \geq 0$.\footnote{The original SSTV-regularized denoising formulation proposed in~\cite{SSTV} incorporates an $\ell_2$ data-fidelity term as a part of the objective function, whereas the formulation in~\eqref{prob:sstv_restoration_formulation} imposes data fidelity as an $\ell_{2}$-ball constraint. These two formulations are essentially the same with appropriate hyperparameters, but constrained formulation like~\eqref{prob:sstv_restoration_formulation} is preferred in experimental comparison and real-world applications because it facilitates hyperparameter settings as adopted, e.g., in Refs.~\cite{CSALSA,EPIpre,ono_2015,ono2017primal,ono_2019}}

\begin{figure*}[!t]
	\begin{center}
		\begin{minipage}{0.97\hsize}
			\centerline{\includegraphics[width=\hsize]{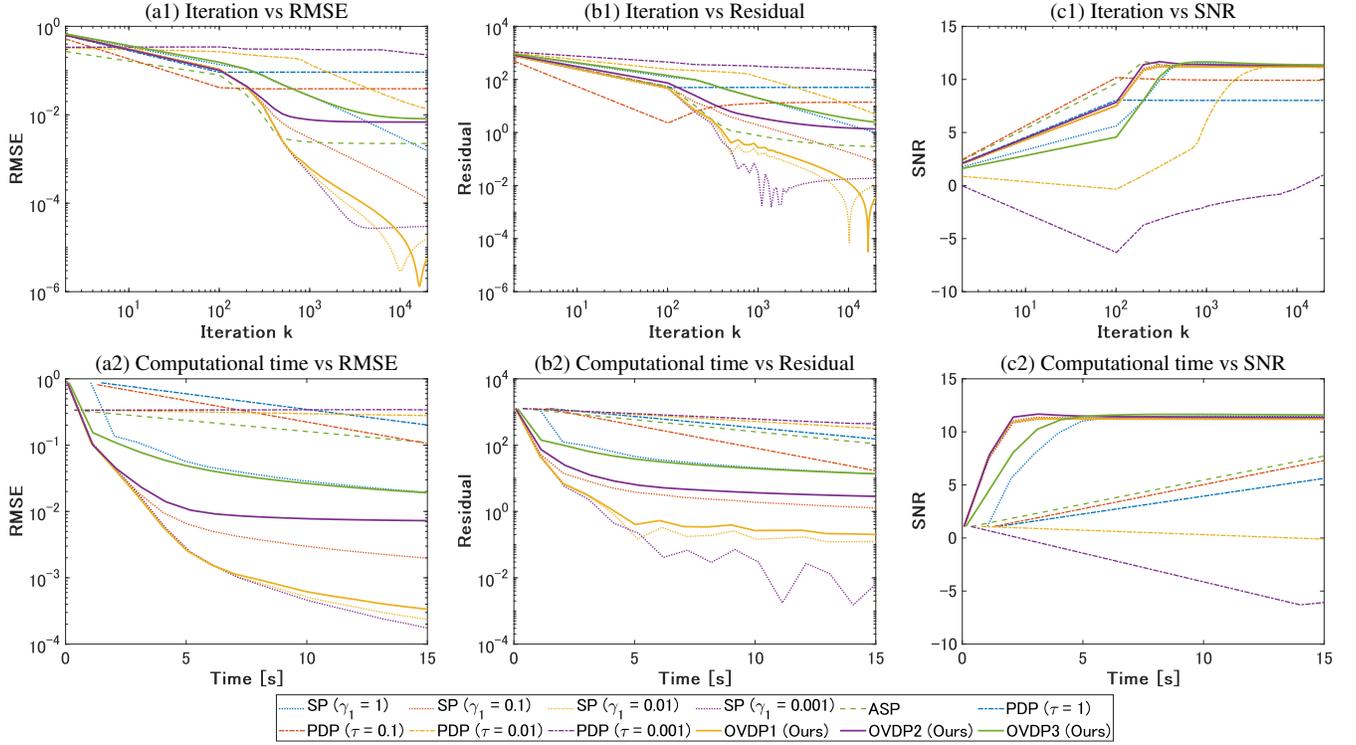}} 
		\end{minipage}
	\end{center}
	
	\vspace{-2mm}
	
	\caption{Convergence profiles of the unmixing experiments. (a): Iterations/computational time versus RMSE. (b): Iterations/computational time versus Residual. (c): Iterations/computational time versus SNR.}
	\label{fig:graph_Unm}
\end{figure*}

By using the indicator function (see Eq.~\eqref{eq:indicator_function}) of $\BallFidelMNR$, Prob.~\eqref{prob:sstv_restoration_formulation} is reduced to Prob.~\eqref{prob:general_form_of_optimization} through the following reformulation:
\begin{align}
	\min_{\substack{\GTMNR, \SparseMNR, \StripeMNR,\\ 
			\VarDualMNR_{1}, \VarDualMNR_{2}, \VarDualMNR_{3}, \VarDualMNR_{4}}} \: & 
	\FuncIndi{\BallSparMNR}{\SparseMNR}
	+ \ParamBalanceMNR\|\StripeMNR\|_{1} \nonumber \\
	& + \|\VarDualMNR_{1}\|_{1}
	+ \|\VarDualMNR_{2}\|_{1}
	+ \FuncIndi{\{\ZeroElem\}}{\VarDualMNR_{3}}
	+ \FuncIndi{\BallFidelMNR}{\VarDualMNR_{4}} \nonumber \\
	\mathrm{s.t.} \: & 
	\begin{cases}
		\VarDualMNR_{1} = \Diffv{\Diffb{\GTMNR}}, \\
		\VarDualMNR_{2} = \Diffh{\Diffb{\GTMNR}}, \\
		\VarDualMNR_{3} = \Diffv{\StripeMNR}, \\	
		\VarDualMNR_{4} = \GTMNR + \SparseMNR + \StripeMNR.
	\end{cases}
	\label{prob:transformed_sstv_restoration_formulation}
\end{align}
Applying Algorithm~\ref{algo:P_PDS} to Prob.~\eqref{prob:transformed_sstv_restoration_formulation}, we can compute an optimal solution of Prob.~\eqref{prob:sstv_restoration_formulation}. Here, since it is satisfied that $\NormOpNoResize{\DiffvSymb\circ\DiffbSymb} \leq 4$, $\NormOpNoResize{\DiffhSymb\circ\DiffbSymb} \leq 4$,\footnote{These are derived from $\NormOpNoResize{\DiffvSymb} \leq 2$, $\NormOpNoResize{\DiffvSymb} \leq 2$, $\NormOpNoResize{\DiffvSymb} \leq 2$~\cite{OP_Diff_Chambolle_2004}, and the submultiplicity of operator norms (Eq.~\eqref{eq:submultiplicity})} and $\NormOpNoResize{\mathbf{I}}=1$, the preconditioners designed by \ONVWs are given in Tab.~\ref{tab:OVDP_MNR}.

\begin{table}[t]
	\begin{center}
		\caption{The Preconditioners by \ONVWs for Mixed Noise Removal.}
		\label{tab:OVDP_MNR}
		\vspace{-2mm}
		\begin{tabular}{cccccccc}
			\toprule
			& $\Precon_{1,1}$ & $\Precon_{1,2}$ & $\Precon_{1,3}$ 
			& $\Precon_{2,1}$ & $\Precon_{2,2}$ & $\Precon_{2,3}$ & $\Precon_{2,4}$ \\
			\cmidrule(lr){2-8}
			
			\vspace{1mm}
			\ONVWOne & $\frac{1}{33}\mathbf{I}$ & $\mathbf{I}$ & $\frac{1}{5}\mathbf{I}$ & $\frac{1}{3}\mathbf{I}$ & $\frac{1}{3}\mathbf{I}$ & $\frac{1}{3}\mathbf{I}$ & $\frac{1}{3}\mathbf{I}$ \\
			
			\vspace{1mm}
			\ONVWTwo & $\frac{1}{9}\mathbf{I}$ & $\mathbf{I}$ & $\mathbf{I}$ & $\frac{1}{4}\mathbf{I}$ & $\frac{1}{4}\mathbf{I}$ & $\frac{1}{33}\mathbf{I}$ & $\frac{1}{3}\mathbf{I}$ \\ 
			
			\ONVWThree & $\frac{1}{33}\mathbf{I}$ & $\frac{1}{33}\mathbf{I}$ & $\frac{1}{33}\mathbf{I}$ & $\frac{1}{33}\mathbf{I}$ & $\frac{1}{33}\mathbf{I}$ & $\frac{1}{33}\mathbf{I}$ & $\frac{1}{33}\mathbf{I}$ \\
			\bottomrule
		\end{tabular}
	\end{center}
	\vspace{-3mm}
\end{table}

\subsubsection{Experimental Results and Discussion}
For \SP, $\PreconOpNormSP$ in~\eqref{eq:SP_expriment} was set as 
\begin{equation}
	\label{eq:PDS_stepsizes}
	\PreconOpNormSP = \sqrt{39},
\end{equation}
because the following inequality holds due to the inequality of the operator norms of block matrices~\cite{norm_inequality}:
\begin{align}
	& \NormOp{
		\begin{bmatrix} 
			\DiffvSymb\circ\DiffbSymb & \ZeroOpe & \ZeroOpe \\ 
			\DiffhSymb\circ\DiffbSymb & \ZeroOpe & \ZeroOpe \\ 
			\ZeroOpe & \ZeroOpe & \DiffvSymb \\ 
			\mathbf{I} & \mathbf{I} & \mathbf{I} 
	\end{bmatrix}}^{2} \nonumber \\
	& \leq 
	\NormOp{\DiffvSymb\circ\DiffbSymb}^{2} 
	+ \NormOp{\DiffhSymb\circ\DiffbSymb}^{2} 
	+ \NormOp{\DiffvSymb}^{2}
	+ 3\NormOp{\mathbf{I}}^{2} \nonumber \\
	& < 4^{2} + 4^{2} + 2^{2} + 3\times 1^{2} = 39,
\end{align}
where $\ZeroOpe$ is a zero operator. 

We also derived the preconditioners in~\eqref{eq:precon_ASEW}, for~\eqref{prob:transformed_sstv_restoration_formulation}. 
Let us remark that since $\DiffvSymb$, $\DiffhSymb$, and $\DiffbSymb$ in~\eqref{prob:transformed_sstv_restoration_formulation} are not usually implemented as explicit matrices, applying \ASEWs to~\eqref{prob:transformed_sstv_restoration_formulation} is not practical in real-world applications. 
Let $\mathbf{x}\in\mathbb{R}^{n_{1}n_{2}n_{3}}$ be a vectorized data cube and $[\mathbf{x}]_{i_{1},i_{2},i_{3}}$ be the value of $\mathbf{x}$ at a location $(i_{1},i_{2},i_{3})$. Then the preconditioners are
\begin{align}
	& \Precon_{1,1} = \mathrm{diag}(\mathbf{g}_{1}),
	\Precon_{1,2} = \mathbf{I},
	\Precon_{1,3} = \mathrm{diag}({\mathbf{g}_{2}}), \nonumber \\ 
	& \Precon_{2,1}=\Precon_{2,2}= \frac{1}{4}, \Precon_{2,3}=\frac{1}{2}, \Precon_{2,4}=\frac{1}{3}\mathbf{I}.
	\label{eq:preconditioner_RCASEW}
\end{align}
Here, $\mathbf{g}_{1}\in\mathbb{R}^{\NumVerPixMNR\NumHorPixMNR\NumBandMNR}$ and $\mathbf{g}_{2}\in\mathbb{R}^{\NumVerPixMNR\NumHorPixMNR\NumBandMNR}$ are given as follows:
\begin{equation}
\label{eq:g}
[\mathbf{g}_{1}]_{i_{1},i_{2},i_{3}} =
\begin{cases}
	\frac{1}{9}, & \mathrm{if} \: i_{1} \in I_{1} \: \mathrm{and}\: i_{2} \in I_{2} \: \mathrm{and}\: i_{3} \in I_{3}; \\
	\frac{1}{3}, & \mathrm{if} \: i_{1} \in E_{1} \: \mathrm{and}\: i_{2} \in E_{2} \: \mathrm{and}\: i_{3} \in E_{3}; \\
	\frac{1}{4}, & \mathrm{if} \: i_{3} \in E_{3} \: \mathrm{and}\: 
	\begin{cases}
		(i_{1} \in E_{1} \: \mathrm{and}\: i_{2} \in I_{2}) \\ 
		\mathrm{or} \\
		(i_{1} \in I_{1} \: \mathrm{and} \: i_{2} \in E_{2}); 
	\end{cases} \\
	\frac{1}{5}, & \mathrm{if} \: i_{1} \in E_{1} \: \mathrm{and}\: i_{2} \in E_{2} \: \mathrm{and}\: i_{3} \in I_{3}; \\
	\frac{1}{7}, & \mathrm{otherwise},
\end{cases} 
\end{equation}
\begin{equation}
	\label{eq:g2}
[\mathbf{g}_{2}]_{i_{1},i_{2},i_{3}} =
\begin{cases}
	\frac{1}{3}, & \mathrm{if} \: i_{1} \in I_{1};  \\
	\frac{1}{2}, & \mathrm{otherwise},
\end{cases}
\end{equation}
where $I_{m}$ and $E_{m}$ for $m = 1,2,3$ are $\{2,\ldots, n_{m} - 1\}$ and $\{1,n_{m}\}$, respectively.
In this case, the skewed proximity operators are separable and thus have analytical solutions. This indicates that P-PDS with \ASEWs does not require FISTA.

As the ground truth HS image, we used Moffett Field~\cite{MoffettField} of size $120 \times 120 \times 176$. 
The observed image was generated by adding white Gaussian noise with the standard deviation $\StanDivGaussMNR = 0.05$ and Salt \& Pepper noise with the ratio $\RateSparse = 0.1$. 
The parameters $\ParamBalanceMNR$, $\ParamSparMNR$, and $\ParamFidelMNR$ were set to $0.005$, $0.5*0.95*\RateSparse*\NumVerPixMNR\NumHorPixMNR\NumBandMNR$, and $0.95\StanDivGaussMNR\sqrt{(1 - \RateSparse)\NumVerPixMNR\NumHorPixMNR\NumBandMNR}$, respectively. 
For the quantitative evaluation of image qualities, we used the Mean Peak Signal-to-Noise Ratio (MPSNR):
\begin{equation}
	\label{eq:PSNR}
	\mbox{MPSNR}(\mathbf{u}^{(\InnerIter)})
	:= \frac{1}{\NumBandMNR} \sum_{b=1}^{\NumBandMNR} 10\log_{10} \left(\frac{\NumVerPixMNR\NumHorPixMNR}{\|\bar{\mathbf{u}}_{b}-\mathbf{u}_{b}^{(\InnerIter)}\|_{2}^{2}}\right),
\end{equation}
where $\bar{\mathbf{u}}_{b}$ and $\mathbf{u}_{b}^{(t)}$ are the $b$th band of the ground-truth image $\bar{\mathbf{u}}$ and the estimated image $\mathbf{u}^{(t)}$.

\begin{figure*}[!t]
	\begin{center}
		\begin{minipage}{0.01\hsize}
			\centerline{\rotatebox{90}{grass}} 
		\end{minipage}
		\begin{minipage}{0.13\hsize}
			\centerline{\includegraphics[width=\hsize]{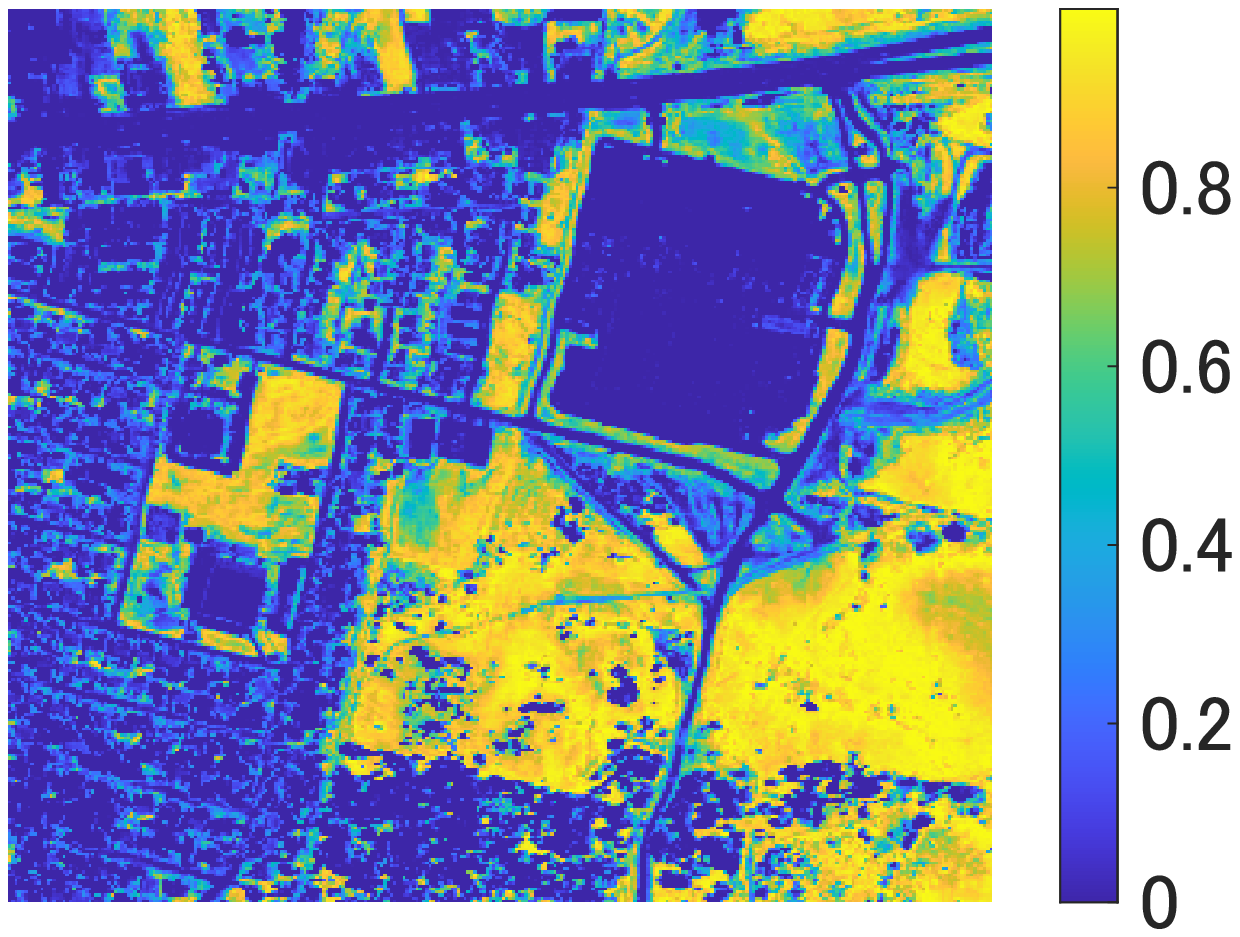}} 
		\end{minipage}
		\begin{minipage}{0.13\hsize}
			\centerline{\includegraphics[width=\hsize]{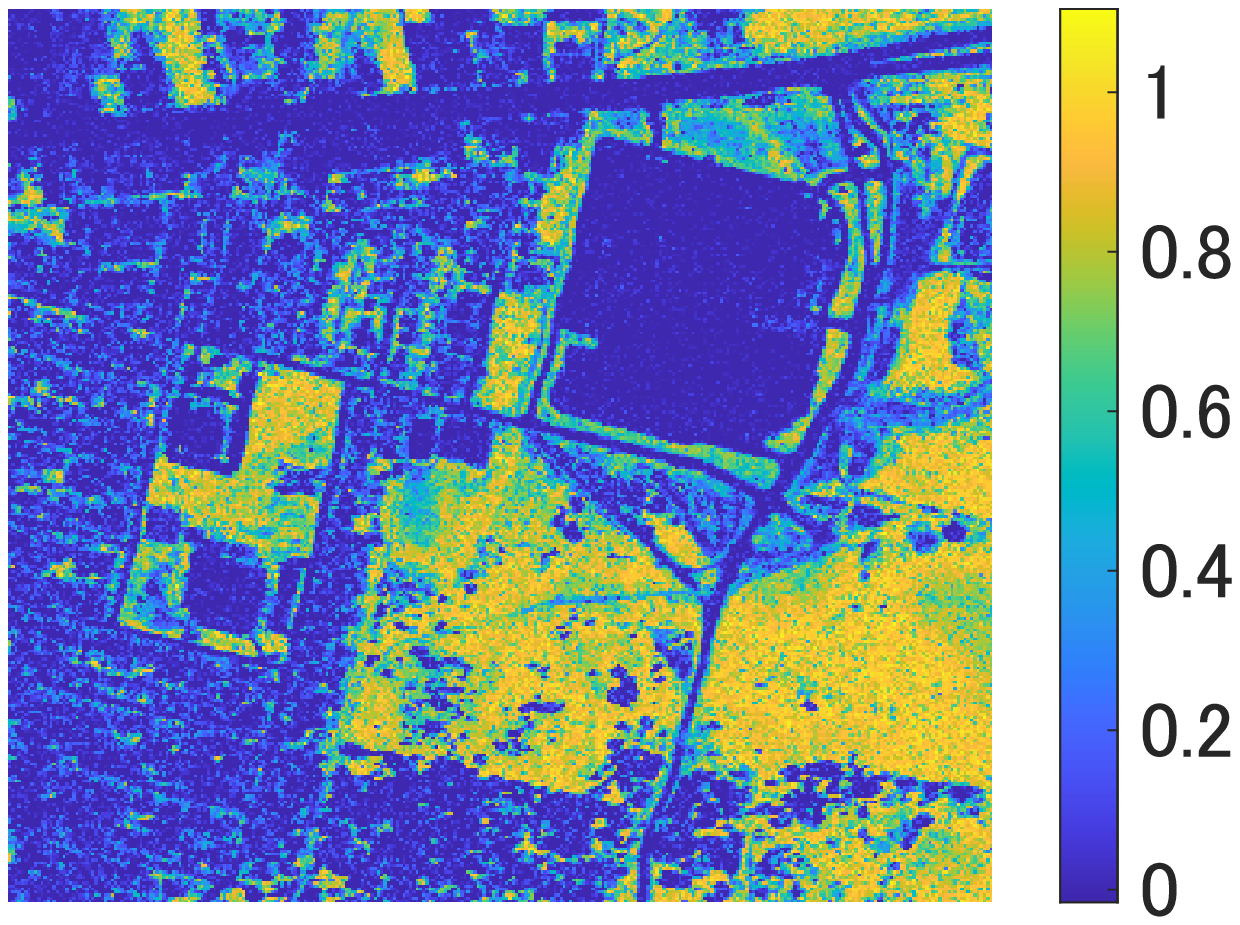}} 
		\end{minipage}
		\begin{minipage}{0.13\hsize}
			\centerline{\includegraphics[width=\hsize]{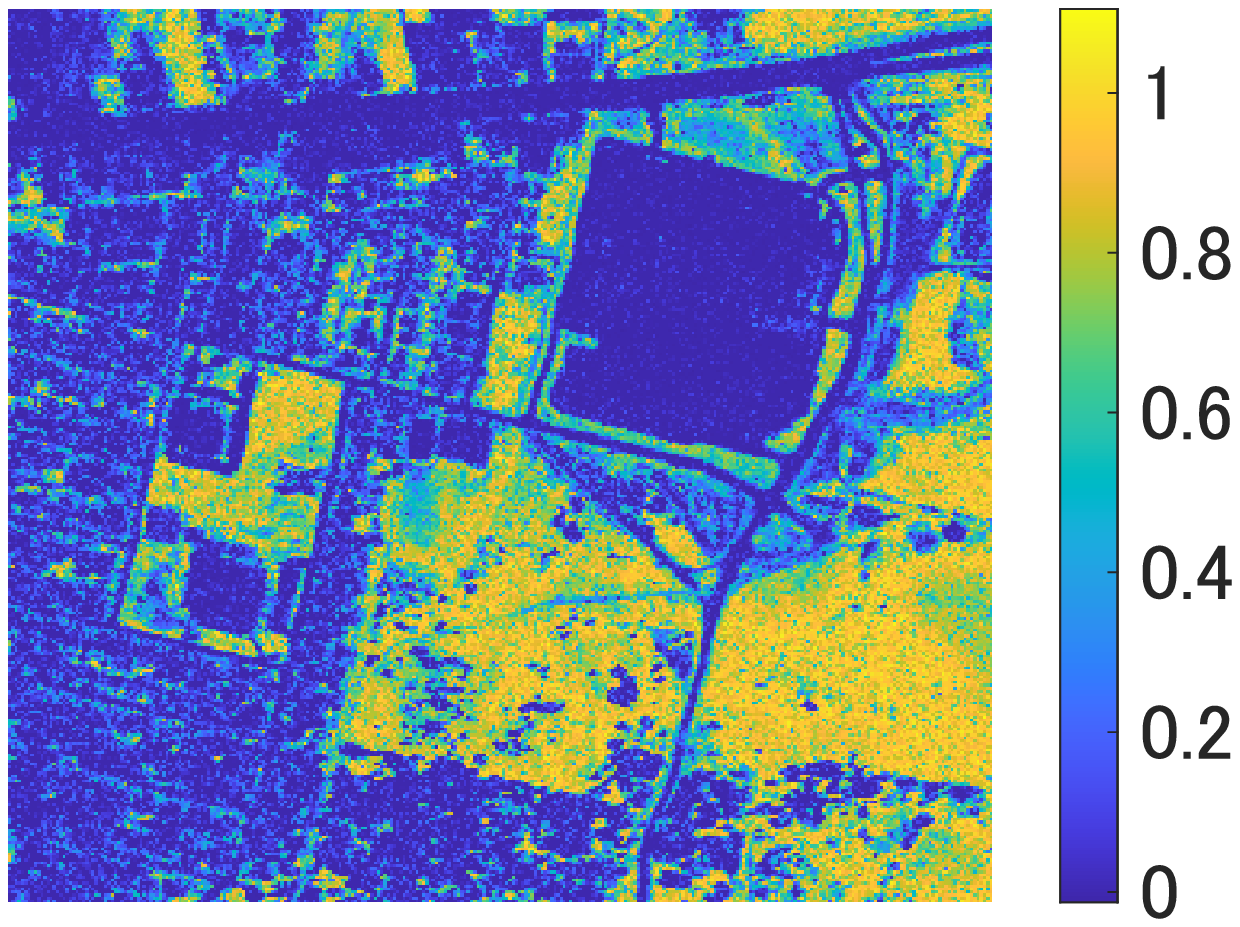}} 
		\end{minipage}
		\begin{minipage}{0.13\hsize}
			\centerline{\includegraphics[width=\hsize]{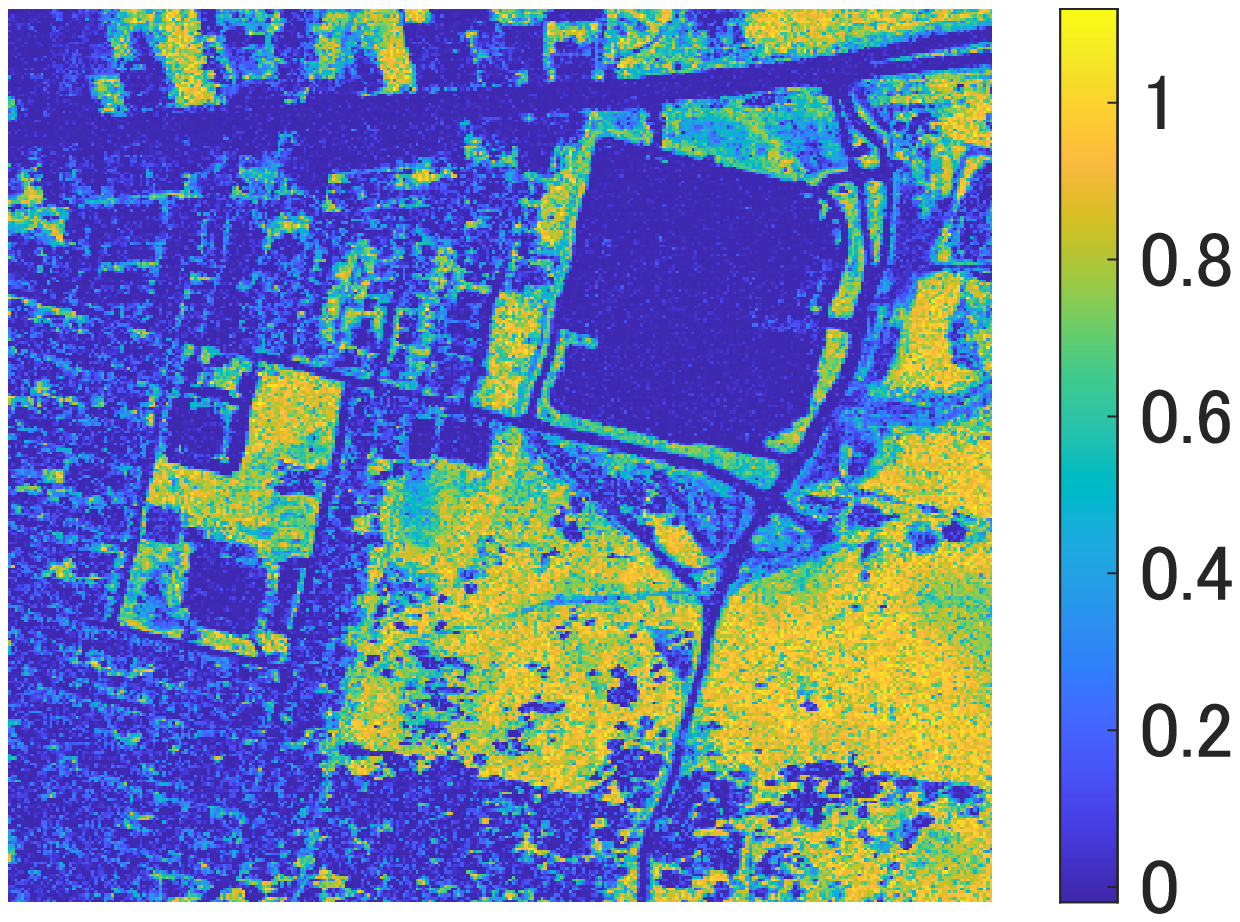}} 
		\end{minipage}
		\begin{minipage}{0.13\hsize}
			\centerline{\includegraphics[width=\hsize]{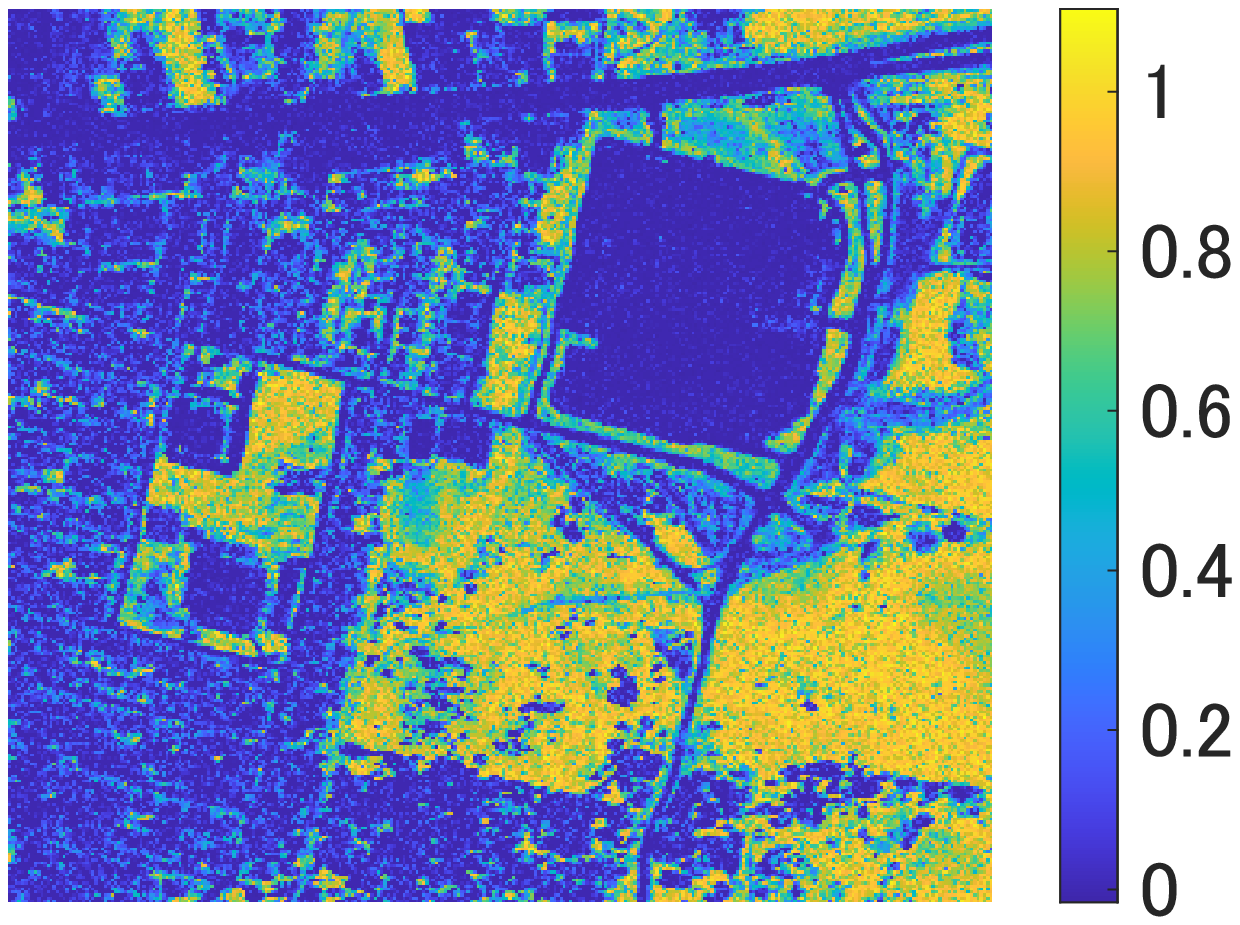}} 
		\end{minipage}
		\begin{minipage}{0.13\hsize}
			\centerline{\includegraphics[width=\hsize]{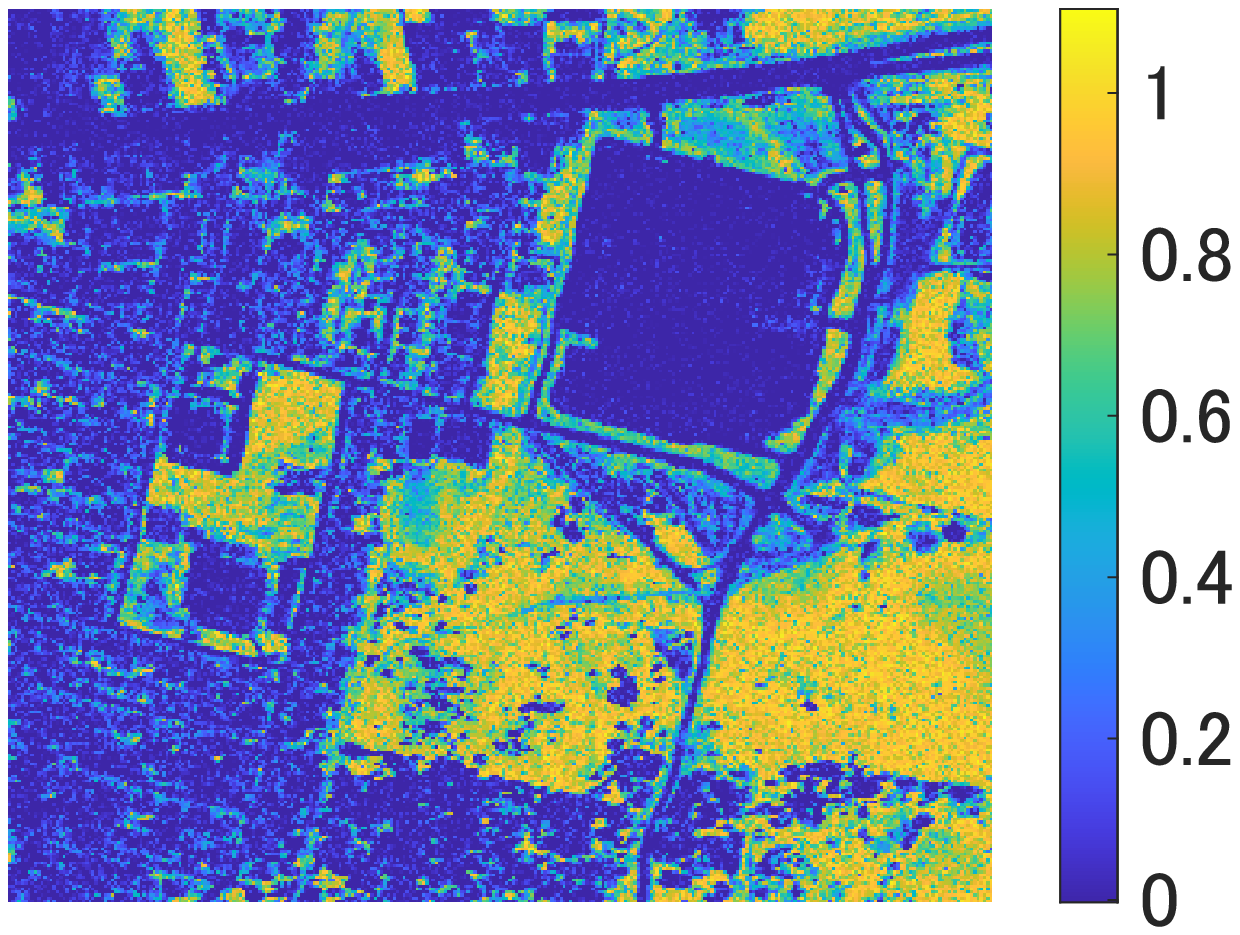}} 
		\end{minipage}
		\begin{minipage}{0.13\hsize}
			\centerline{\includegraphics[width=\hsize]{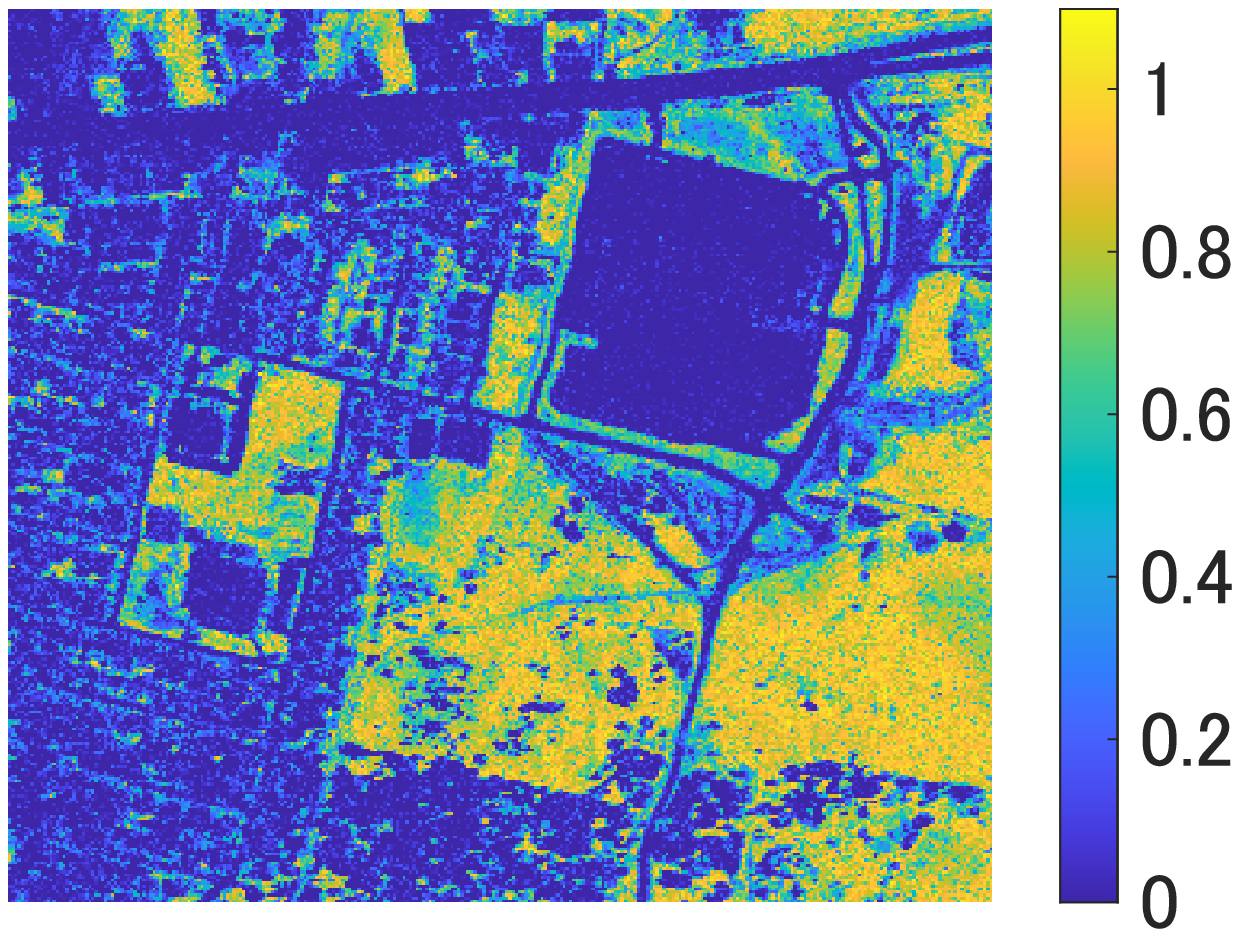}} 
		\end{minipage}
	
		\begin{minipage}{0.01\hsize}
			\centerline{\rotatebox{90}{roof}} 
		\end{minipage}
		\begin{minipage}{0.13\hsize}
			\centerline{\includegraphics[width=\hsize]{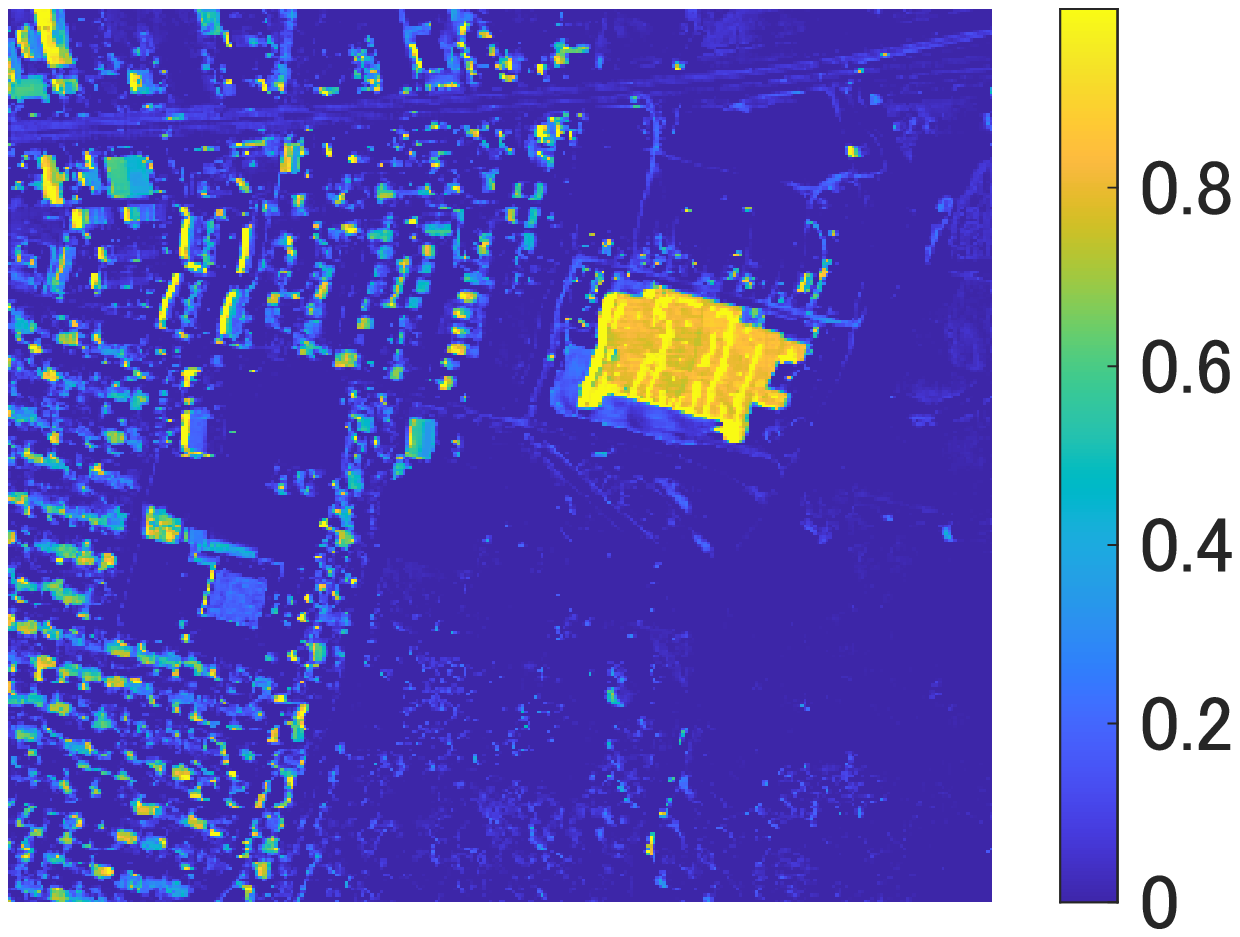}} 
		\end{minipage}
		\begin{minipage}{0.13\hsize}
			\centerline{\includegraphics[width=\hsize]{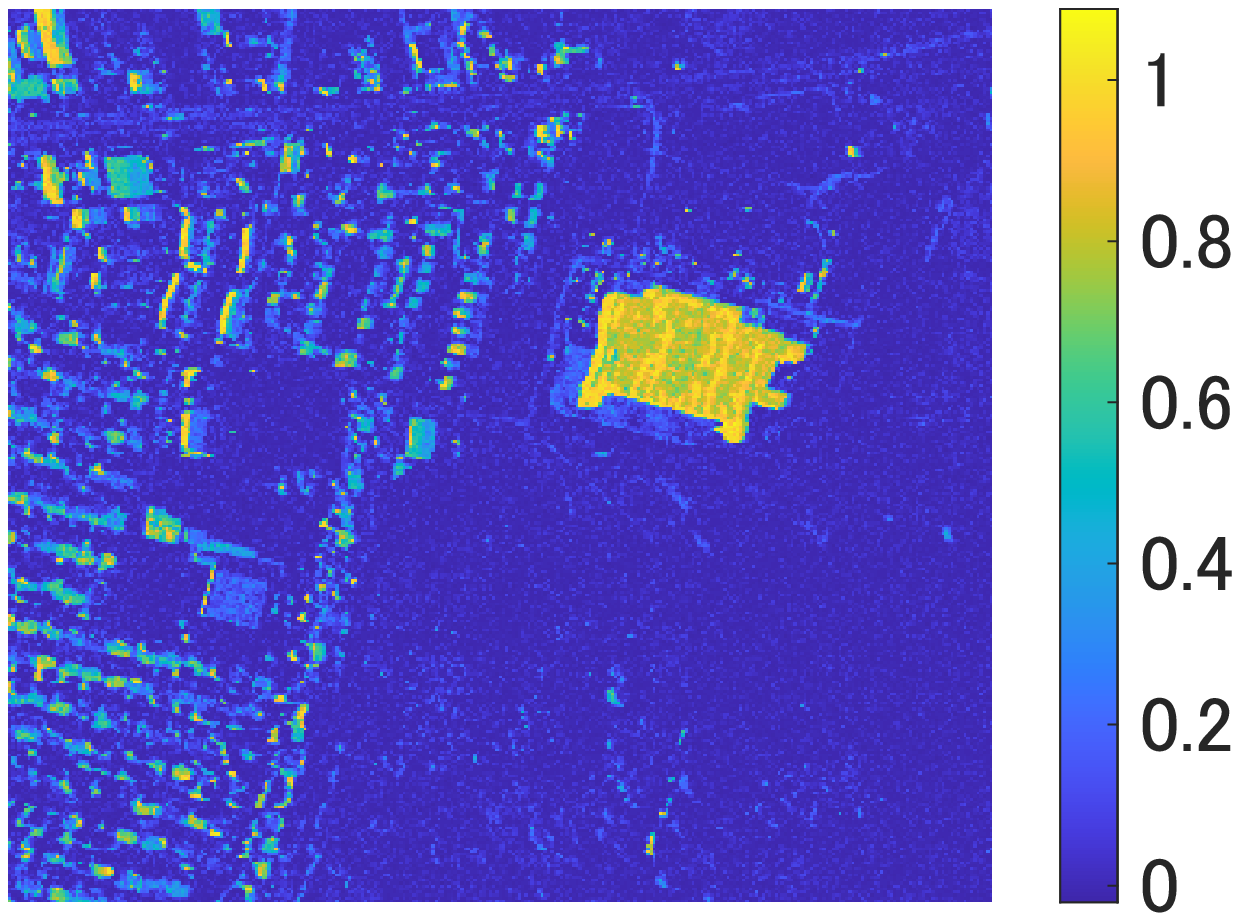}} 
		\end{minipage}
		\begin{minipage}{0.13\hsize}
			\centerline{\includegraphics[width=\hsize]{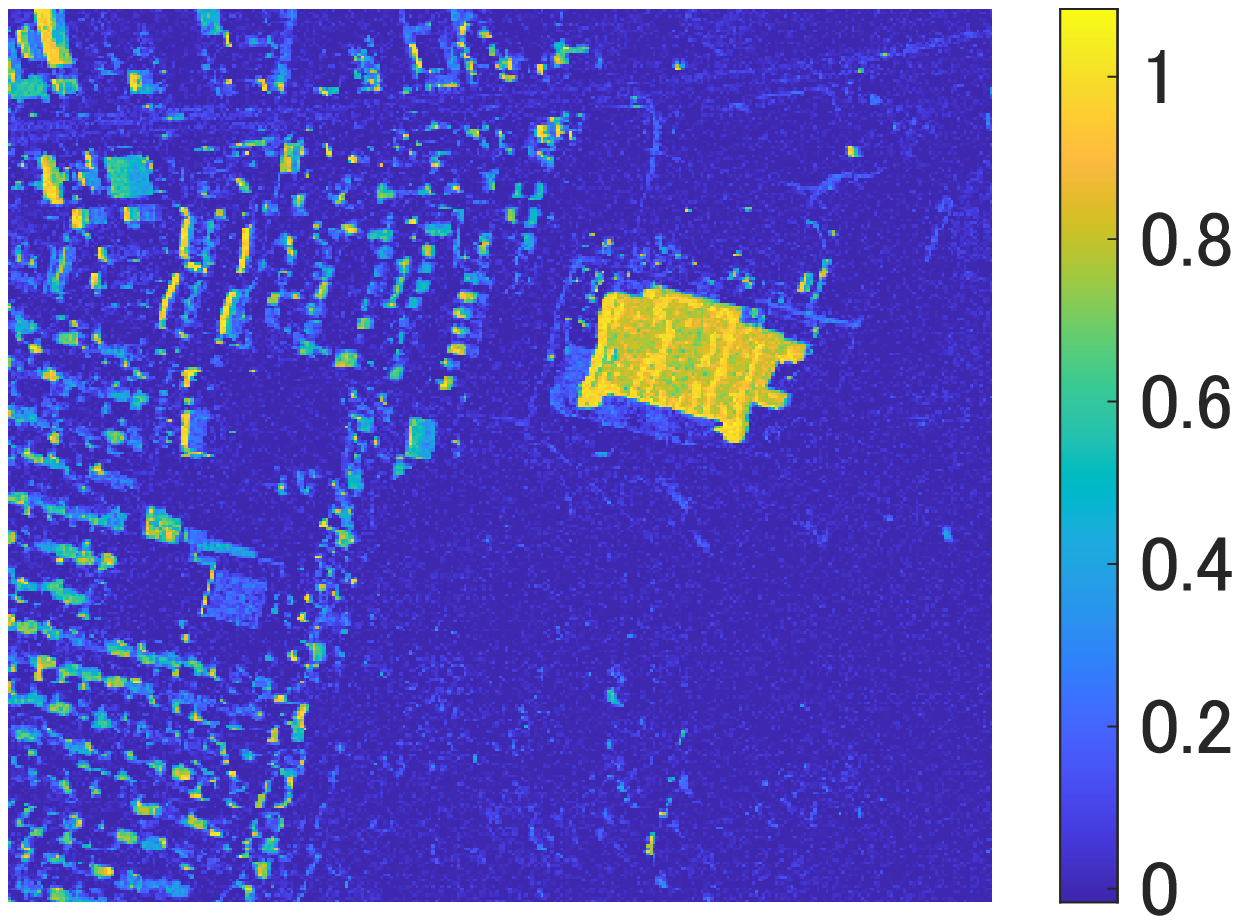}} 
		\end{minipage}
		\begin{minipage}{0.13\hsize}
			\centerline{\includegraphics[width=\hsize]{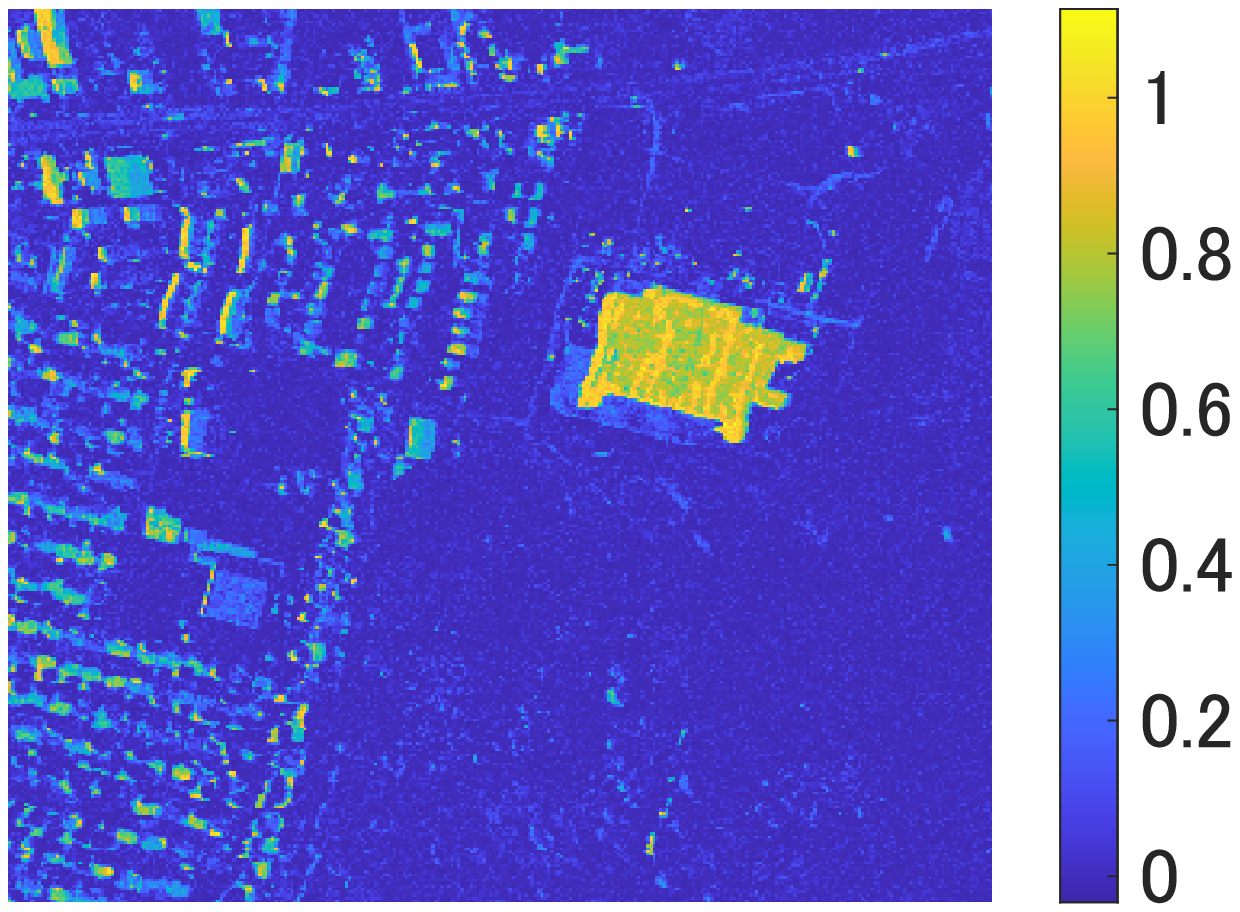}} 
		\end{minipage}
		\begin{minipage}{0.13\hsize}
			\centerline{\includegraphics[width=\hsize]{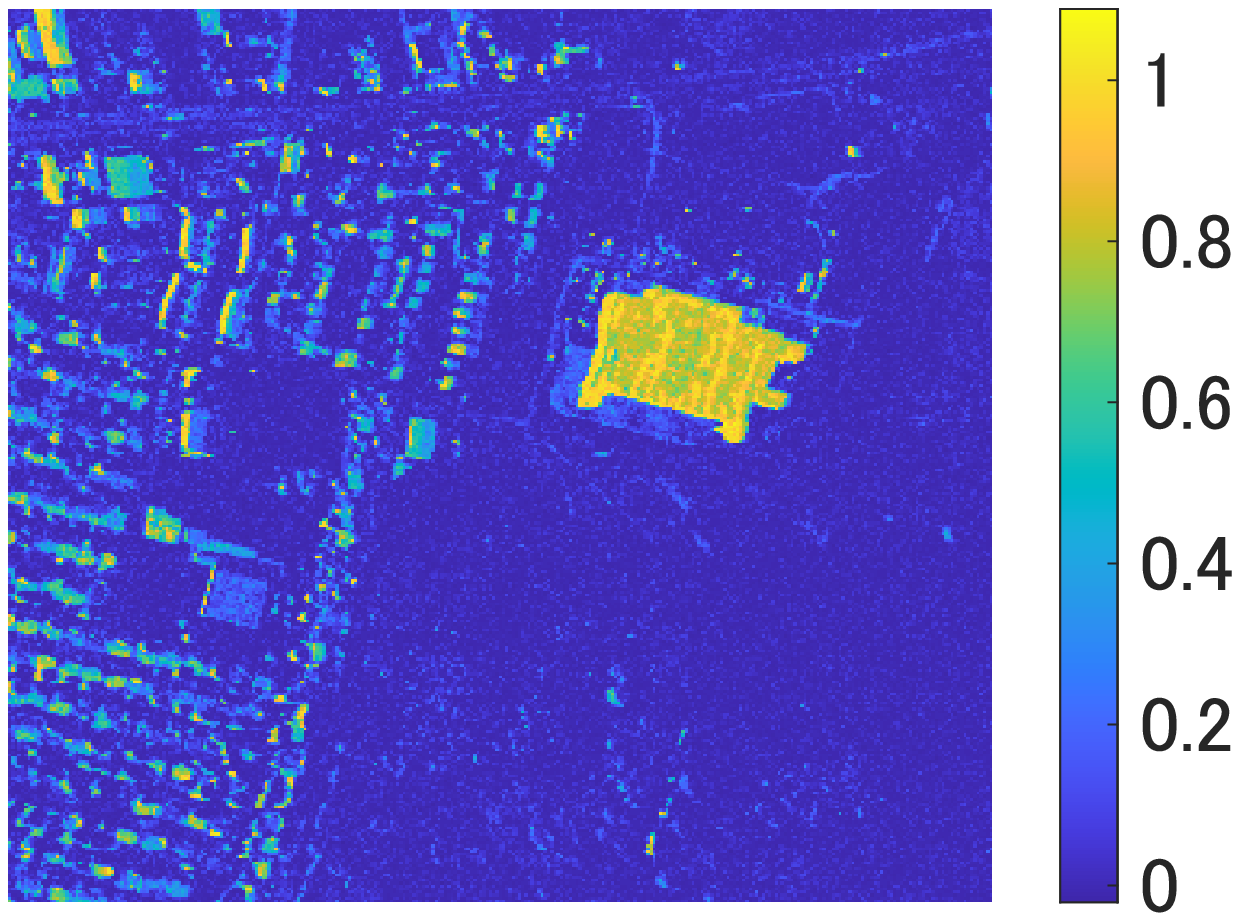}} 
		\end{minipage}
		\begin{minipage}{0.13\hsize}
			\centerline{\includegraphics[width=\hsize]{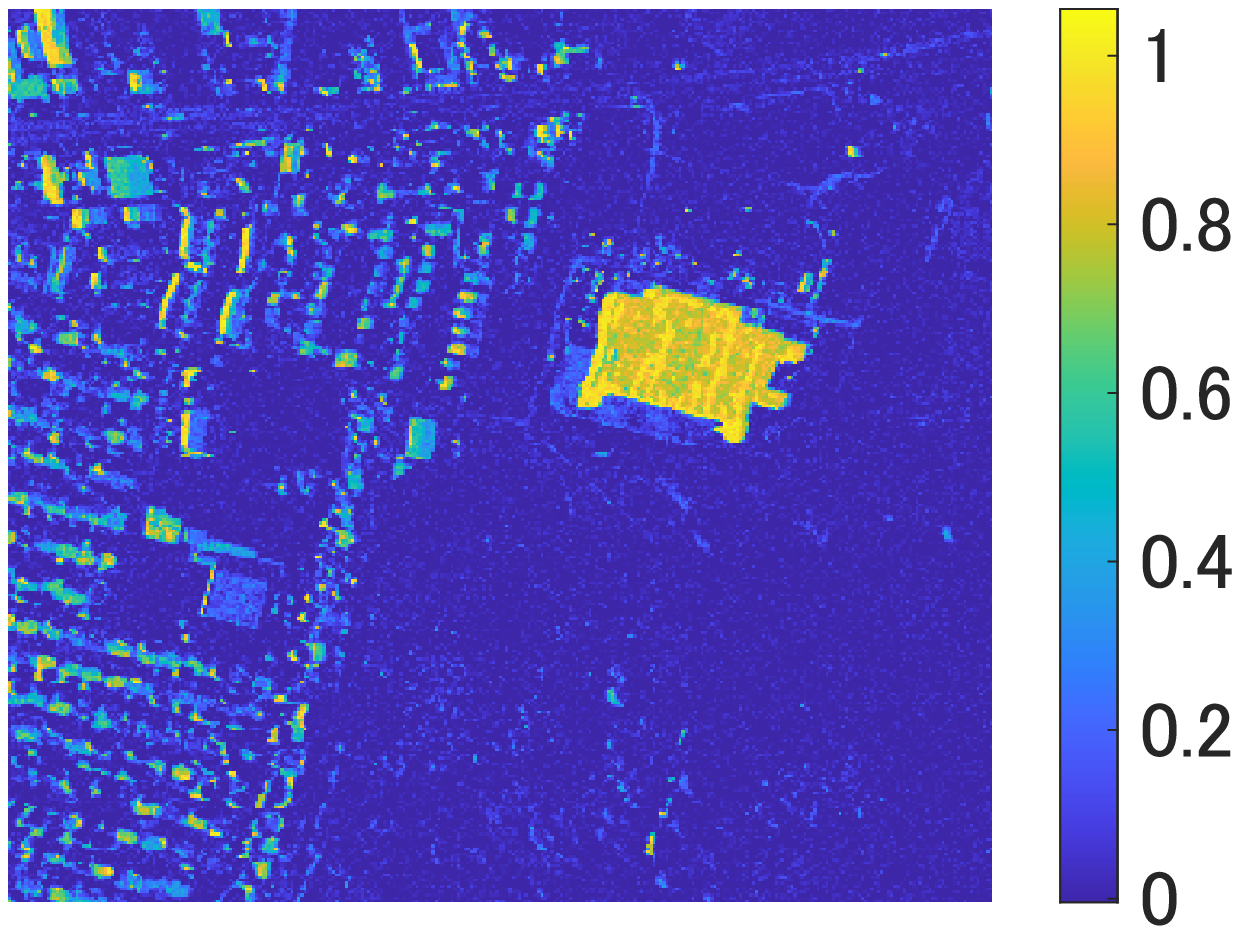}} 
		\end{minipage}
		\begin{minipage}{0.13\hsize}
			\centerline{\includegraphics[width=\hsize]{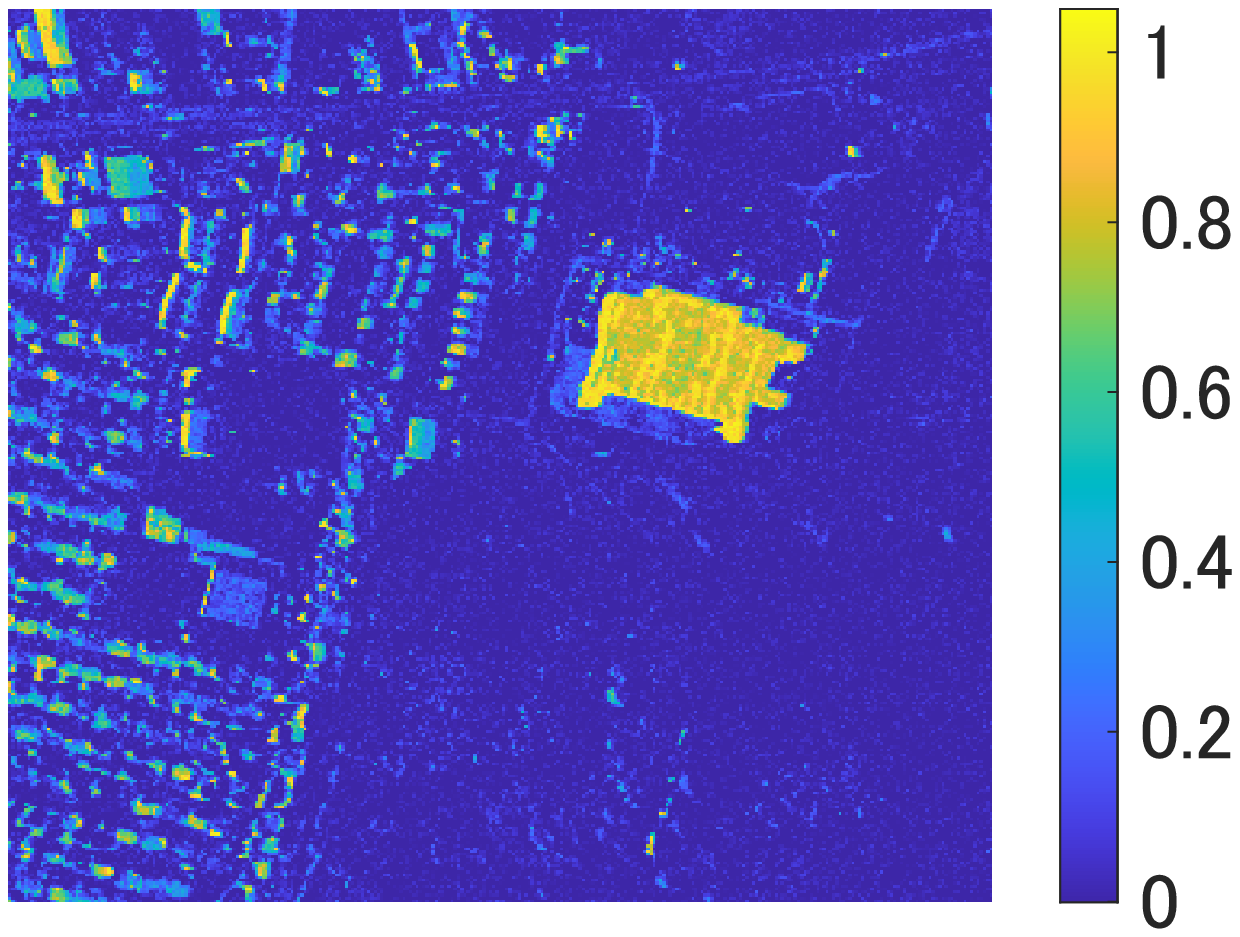}} 
		\end{minipage}
	
		\vspace{1mm}
		
		\begin{minipage}{0.01\hsize}
			\centerline{~} 
		\end{minipage}
		\begin{minipage}{0.13\hsize}
			\centerline{(a) SNR [dB]}
		\end{minipage}
		\begin{minipage}{0.13\hsize}
			\centerline{(b) $11.19$ [dB]}
		\end{minipage}
		\begin{minipage}{0.13\hsize}
			\centerline{(c) $11.43$ [dB]}
		\end{minipage}
		\begin{minipage}{0.13\hsize}
			\centerline{(d) $11.18$ [dB]}
		\end{minipage}
		\begin{minipage}{0.13\hsize}
			\centerline{(e) $11.22$ [dB]}
		\end{minipage}
		\begin{minipage}{0.13\hsize}
			\centerline{(f) $11.46$ [dB]}
		\end{minipage}
		\begin{minipage}{0.13\hsize}
			\centerline{(g) $11.45$ [dB]}
		\end{minipage}
	
		\vspace{1mm}

	\end{center}

	\vspace{-4mm}
	
	\caption{Abundance maps of HS unmixing results. (a): The ground truth abundance maps. (b): The abundance maps estimated by P-PDS with \SP~\cite{PDS_1} ($\ParamPDS_{1} = 0.001$). (c): The abundance maps estimated by P-PDS with \ASEW~\cite{DP-PDS}. (d): The abundance maps estimated by P-PDS with \SPDP~\cite{SPDP} ($\ParamSPDP = 0.01$). (e): The abundance maps estimated by P-PDS with \ONVWOnes (Ours). (f): The abundance maps estimated by P-PDS with \ONVWTwos (Ours). (g): The abundance maps estimated by P-PDS with \ONVWThrees (Ours).}
	\label{fig:results_Unm}
\end{figure*}

Fig.~\ref{fig:graph_MNR} plots iterations versus RMSE, Residual, and MPSNR and computational time versus RMSE, Residual, and MPSNR, respectively. 
In terms of iterations (Figs.~\ref{fig:graph_MNR} (a1), (b1), and (c1)), P-PDSs with \SPs ($\ParamPDS_{1}=0.01$), \SPs ($\ParamPDS_{1}=0.001$), \SPDPs ($\ParamSPDP = 0.01$), and \SPDPs ($\ParamSPDP = 0.001$) were very slow, and P-PDSs with \SPs ($\ParamPDS_{1}=1$), \SPs ($\ParamPDS_{1}=0.1$), \ASEW, \SPDPs ($\ParamSPDP = 1$), \SPDPs ($\ParamSPDP = 0.1$), \ONVWTwo, and \ONVWThrees were fast. 
For P-PDS with \ONVWOnes, the evolution of the MPSNR values was slightly slow, but the convergence of the RMSE and Residual values was not.
In terms of computational time (Figs.~\ref{fig:graph_MNR} (a2), (b2), and (c2)), although P-PDSs with \SP, \ASEW, and \ONVWs have the same computational complexity per iteration in $O$-notation, P-PDS with \ASEWs took longer than P-PDSs with \SPs and \ONVW.
When computing the analytic solutions of the proximity operators, P-PDSs with \SPs and \ONVWs require the multiplication of a scalar and a vector, while P-PDS with \ASEWs requires the element-wise multiplication of two vectors. 
Since the latter takes longer to run than the former, P-PDS with \ASEWs was longer in running time. 
P-PDSs with \SPDPs were very slow because they require the iterative algorithm to calculate the skewed proximity operator.

Fig.~\ref{fig:results_MNR} shows the denoising results and the MPSNR values [dB] obtained by P-PDS with \SPs ($\ParamPDS_{1}=0.1$), \ASEW, \SPDPs ($\ParamSPDP = 0.1$), \ONVWOne, \ONVWTwo, and \ONVWThree. 
The algorithm was run until satisfying the stopping criterion or reaching $10000$ iterations.
We can see that all results are almost the same in terms of the MPSNR and the visual qualities.

\subsection{Application to Hyperspectral Unmixing}
An HS image is a three-dimensional data cube that consists of two-dimensional spatial information and one-dimensional spectral information.
Compared to grayscale or RGB images, HS images offer more than several hundred bands, each of which contains specific unique wavelength characteristics of materials such as minerals, soils, and liquids.
Due to the trade-off between spatial resolution and wavelength resolution, HS sensors do not have a sufficient spatial resolution, resulting in containing multiple components (called endmembers) in a pixel \cite{keshava2002spectral}, which refers to as a mixel. 
The process of decomposing the mixels into endmembers and their abundance maps is called unmixing.
Unmixing has been actively studied in the remote sensing field because of its indispensability for analyzing HS images~\cite{ma2013signal,ghamisi2017advances}.
One of the popular unmixing methods is the constrained collaborative sparse regression problem~\cite{CLSUnSAL_iordache2013}, which has attracted attention as an optimization-based strategy for HS unmixing ~\cite{JSTV_Aggarwal_2016,RSSUn_TV_Wang_2019,ICoNMF_TV_Yuan_2020}.

\begin{figure*}[!t]
	\begin{center}
		\begin{minipage}{0.97\hsize}
			\centerline{\includegraphics[width=\hsize]{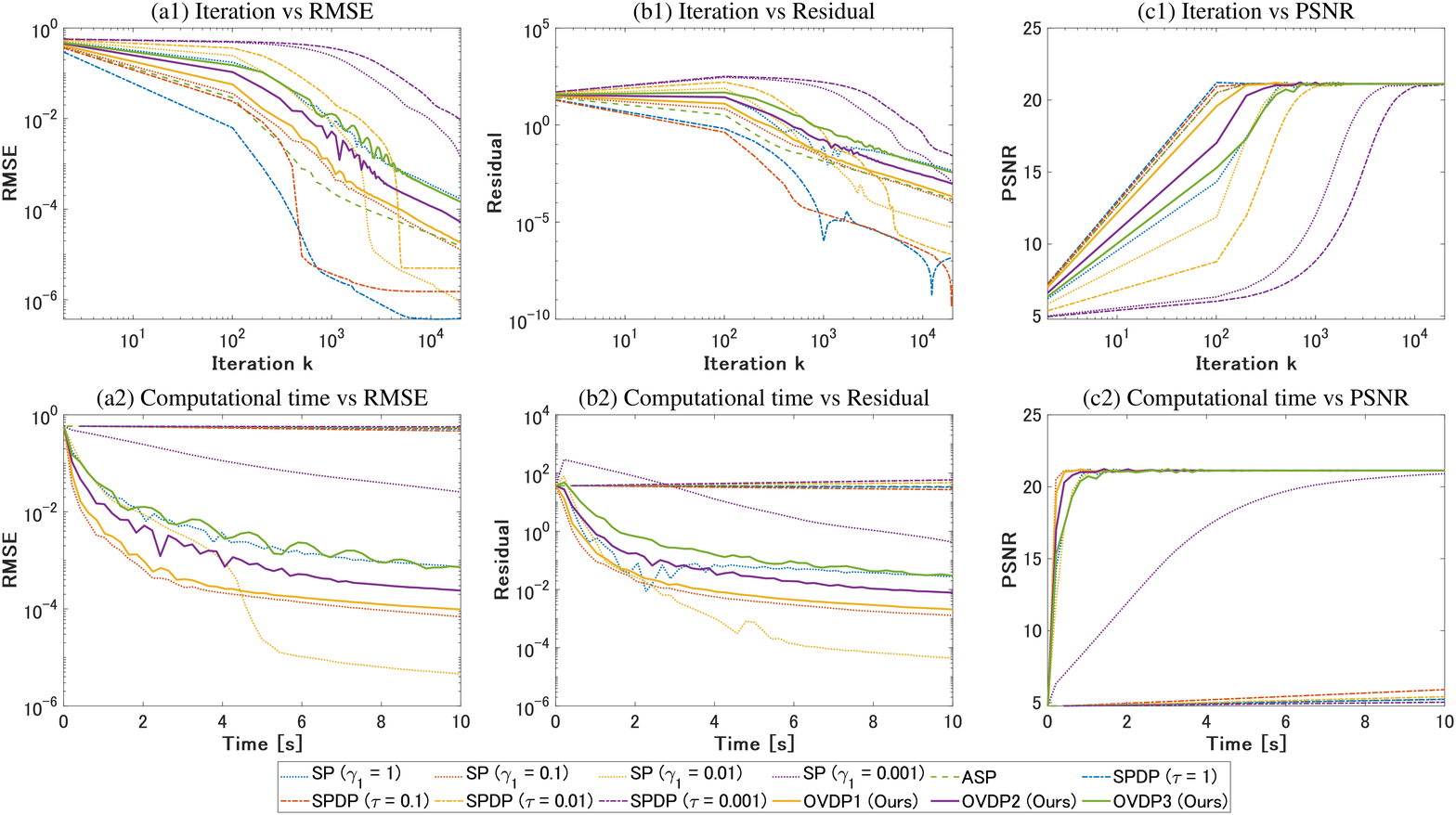}} 
		\end{minipage}
	\end{center}

	\vspace{-2mm}
	
	\caption{Convergence profiles of the graph singal recovery experiments. (a): Iterations/computational time versus RMSE. (b): Iterations/computational time versus Residual. (c): Iterations/computational time versus PSNR.}
	\label{fig:graph_GSR}
\end{figure*}

\subsubsection{Problem Formulation}
Let $\ObsHSIUnm_{\IndPixUnm} \in \RealNumSet^{\NumBandUnm \times 1}$ represent an  $\NumBandUnm$-dimensional $i$th pixel vector of an HS image with $\NumBandUnm$ spectral bands and $\MatEndmemberUnm = [\VecEmdmemberUnm_{1}, \ldots, \VecEmdmemberUnm_{\NumEndmember}] \in \RealNumSet^{\NumBandUnm \times \NumEndmember}$ be an endmember matrix that denotes a spectral library with $\NumEndmember$ spectral signatures.
The pixel $\ObsHSIUnm_{\IndPixUnm}$ can be modeled as the following form of linear combination: 
\begin{equation}
	\ObsHSIUnm_{\IndPixUnm} = \MatEndmemberUnm \AbundanceUnm_{\IndPixUnm} + \NoiseUnm_{\IndPixUnm},
\end{equation}
where $\AbundanceUnm_{\IndPixUnm}\in \RealNumSet^{\NumBandUnm \times 1}$ is an abundance map.
Introducing the extended endmember matrix $\MatEndmemberExpUnm = \mathrm{diag}(\MatEndmemberUnm, \ldots, \MatEndmemberUnm) \in \RealNumSet^{\NumPixUnm\NumBandUnm \times \NumPixUnm\NumEndmember}$, we can express an observed HS image $\ObsHSIUnm = [\ObsHSIUnm_{1}^{\top}, \ldots, \ObsHSIUnm_{\NumPixUnm}^{\top}]^{\top}$ as
\begin{equation}
	\ObsHSIUnm = \MatEndmemberExpUnm \AbundanceUnm + \NoiseUnm.
\end{equation}
Based on the above model, the constrained collaborative sparse regression problem of unmixing is formulated as the following convex optimization problem:
\begin{equation}
	\label{prob:CLSUnSAL}
	\min_{\AbundanceUnm} 
	\| \AbundanceUnm \|_{1, 2} \: 
	\st \: 
		\MatEndmemberExpUnm\AbundanceUnm \in \BallFidelUnm, 
		\AbundanceUnm \in \NNRealNumSet^{\NumVerPixUnm\NumHorPixUnm\NumBandUnm}.
\end{equation}
The first term is the mixed $\ell_{1, 2}$ norm, which is defined by
\begin{equation}
	\label{eq:l21norm_Unm}
	\| \AbundanceUnm \|_{1, 2} = \sum_{\IndEndmember = 1}^{\NumEndmember} \sqrt{\sum_{\IndPixUnm = 1}^{\NumPixUnm} [\AbundanceUnm_{\IndPixUnm}]_{\IndEndmember}^{2}}.
\end{equation}
The first constraint serves as data-fidelity with the $\ObsHSIUnm$-centered $\ell_{2}$-ball of the radius $\ParamFidelUnm > 0$.\footnote{The original constrained collaborative sparse regression formulation proposed in~\cite{CLSUnSAL_iordache2013} incorporates an $\ell_{2}$ data-fidelity term as a part of the objective function, whereas the formulation in~\eqref{prob:CLSUnSAL} imposes data fidelity as an $\ell_{2}$-ball constraint. The reason is similar to the case of the mixed noise removal experiment.} 
The second constraint enforces $\AbundanceUnm$ to belong to the nonnegative orthant $\NNRealNumSet^{\NumVerPixUnm\NumHorPixUnm\NumBandUnm}$.

By using the indicator functions (see Eq.~\eqref{eq:indicator_function}) of $\BallFidelUnm$ and $\NNRealNumSet^{\NumVerPixUnm\NumHorPixUnm\NumBandUnm}$, Prob.~\eqref{prob:CLSUnSAL} is reduced to Prob.~\eqref{prob:general_form_of_optimization} via the following reformulation:
\begin{align}
	\min_{\AbundanceUnm, \VarDualUnm_{1}, \VarDualUnm_{2}} \:
	& \| \AbundanceUnm \|_{1, 2}
	+ \FuncIndi{\BallFidelUnm}{\VarDualUnm_{1}} 
	+ \FuncIndi{\NNRealNumSet^{\NumVerPixUnm\NumHorPixUnm\NumBandUnm}}{\VarDualUnm_{2}} \nonumber \\
	\st \: 
	& \begin{cases}
		\VarDualUnm_{1} = \MatEndmemberExpUnm\AbundanceUnm, \\
		\VarDualUnm_{2} = \AbundanceUnm.
	\end{cases}
	\label{prob:CLSUnSAL_PDS}
\end{align}
Applying Algorithm~\ref{algo:P_PDS} to Prob.~\eqref{prob:CLSUnSAL_PDS}, we can obtain an optimal solution of Prob.~\eqref{prob:CLSUnSAL}. 
Since the functions $\| \cdot \|_{1, 2}$ and $\iota_{\BallFidelUnm}$ are not separable for each element of the input variable, an iterative algorithm is needed for the calculation of their skewed proximity operators relative to the metric induced by the preconditioners of \ASEWs and \SPDP.
Here, the preconditioners designed by \ONVWs are as in Tab.~\ref{tab:OVDP_Unm}.

\begin{table}[t]
	\begin{center}
		\caption{The Preconditioners by \ONVWs for Unmixing.}
		\label{tab:OVDP_Unm}
		\vspace{-2mm}
		\begin{tabular}{cccc}
			\toprule
			& $\Precon_{1,1}$ 
			& $\Precon_{2,1}$ & $\Precon_{2,2}$ \\
			\cmidrule(lr){2-4}
			
			\vspace{1mm}
			\ONVWOne & $\frac{1}{\NormOpNoResize{\MatEndmemberExpUnm}^{2} + 1^{2}}\mathbf{I}$ & $\mathbf{I}$ & $\mathbf{I}$ \\
			
			\vspace{1mm}
			\ONVWTwo & $\frac{1}{\NormOpNoResize{\MatEndmemberExpUnm} + 1}\mathbf{I}$ & $\frac{1}{\NormOpNoResize{\MatEndmemberExpUnm}}\mathbf{I}$ & $\mathbf{I}$ \\ 
			
			\ONVWThree & $\frac{1}{2}\mathbf{I}$ & $\frac{1}{\NormOpNoResize{\MatEndmemberExpUnm}^{2}}\mathbf{I}$ & $\mathbf{I}$ \\
			\bottomrule
		\end{tabular}
	\end{center}
	\vspace{-3mm}
\end{table}

\subsubsection{Experimental Results and Discussion}
For \SP, $\PreconOpNormSP$ in~\eqref{eq:SP_expriment} was set as 
\begin{equation}
	\label{eq:SP_GS_param}
	\PreconOpNormSP = \sqrt{\NormOpNoResize{\MatEndmemberExpUnm}^{2} + 1},
\end{equation}
because the following inequality holds due to the inequality of the operator norms of block matrices~\cite{norm_inequality}:
\begin{equation}
	\NormOp{\begin{bmatrix} \MatEndmemberExpUnm \\ \mathbf{I} \end{bmatrix}}^{2} 
	\leq \NormOpNoResize{\MatEndmemberExpUnm}^{2} + \NormOpNoResize{\mathbf{I}}^{2} 
	= \NormOpNoResize{\MatEndmemberExpUnm}^{2} + 1.
\end{equation}
For \SPDP, the preconditioners in~\eqref{eq:precon_SPDP_2} were used since the number of dual variables is two.

As the ground truth HS image, we used the urban dataset\footnote{http://www.tec.army.mil/Hypercube}, which has been widely used in the field of HS unmixing. 
The image consists of $307\times 307$ pixels with $210$ spectral bands.
In the image, six main endmembers can be observed in the scene: asphalt road, grass, tree, roof, metal, and dirt.
The observed data was generated by adding white Gaussian noise with the standard deviation $\StanDivGaussUnm = 0.05$. 
The parameter  $\ParamFidelUnm$ was set to $0.9\StanDivGaussMNR\sqrt{\NumVerPixMNR\NumHorPixMNR\NumBandMNR}$. 
For the quantitative evaluation of image qualities, we used the Signal-to-Noise Ratio (SNR)~\footnote{This evaluation metric is often referred to as the signal to reconstruction error in the leterature of HS unmixing (e.g.,~\cite{CLSUnSAL_iordache2013,RSSUn_TV_Wang_2019,ICoNMF_TV_Yuan_2020}).}:
\begin{equation}
	\label{eq:SRE}
	\mathrm{SNR}(\AbundanceUnm^{(\InnerIter)}) := 10 \log_{10}\left(\frac{\|\bar{\AbundanceUnm}\|_{2}}{\| \AbundanceUnm^{(\InnerIter)} - \bar{\AbundanceUnm}\|_{2}}\right),
\end{equation}
where $\AbundanceUnm^{(\InnerIter)}$ and  $\bar{\AbundanceUnm}$ are the estimated and ground true abundance maps, respectively.

\begin{figure*}[!t]
	\begin{center}
		\begin{minipage}{0.23\hsize}
			\centerline{\includegraphics[width=\hsize]{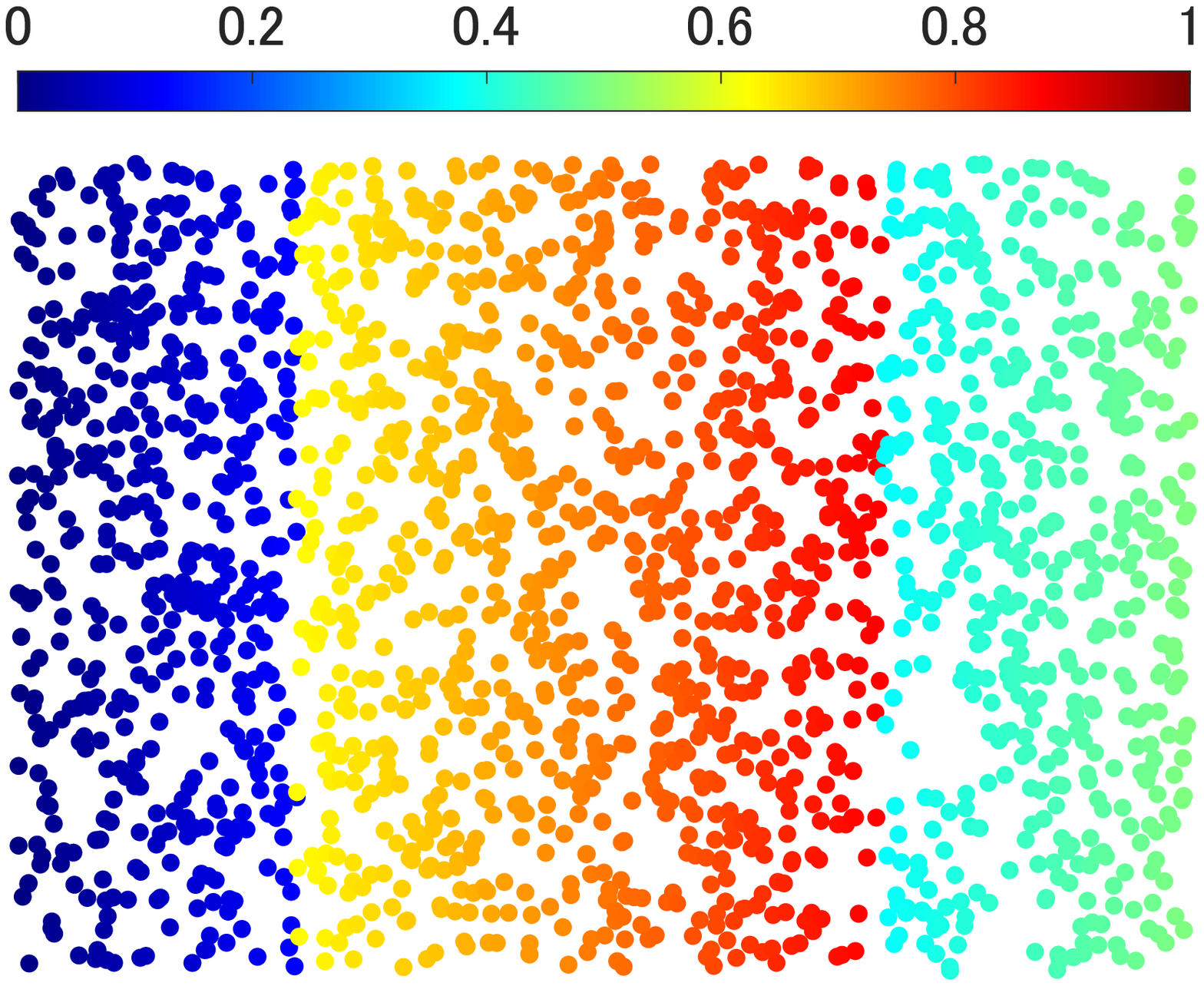}} 
		\end{minipage}
		\begin{minipage}{0.23\hsize}
			\centerline{\includegraphics[width=\hsize]{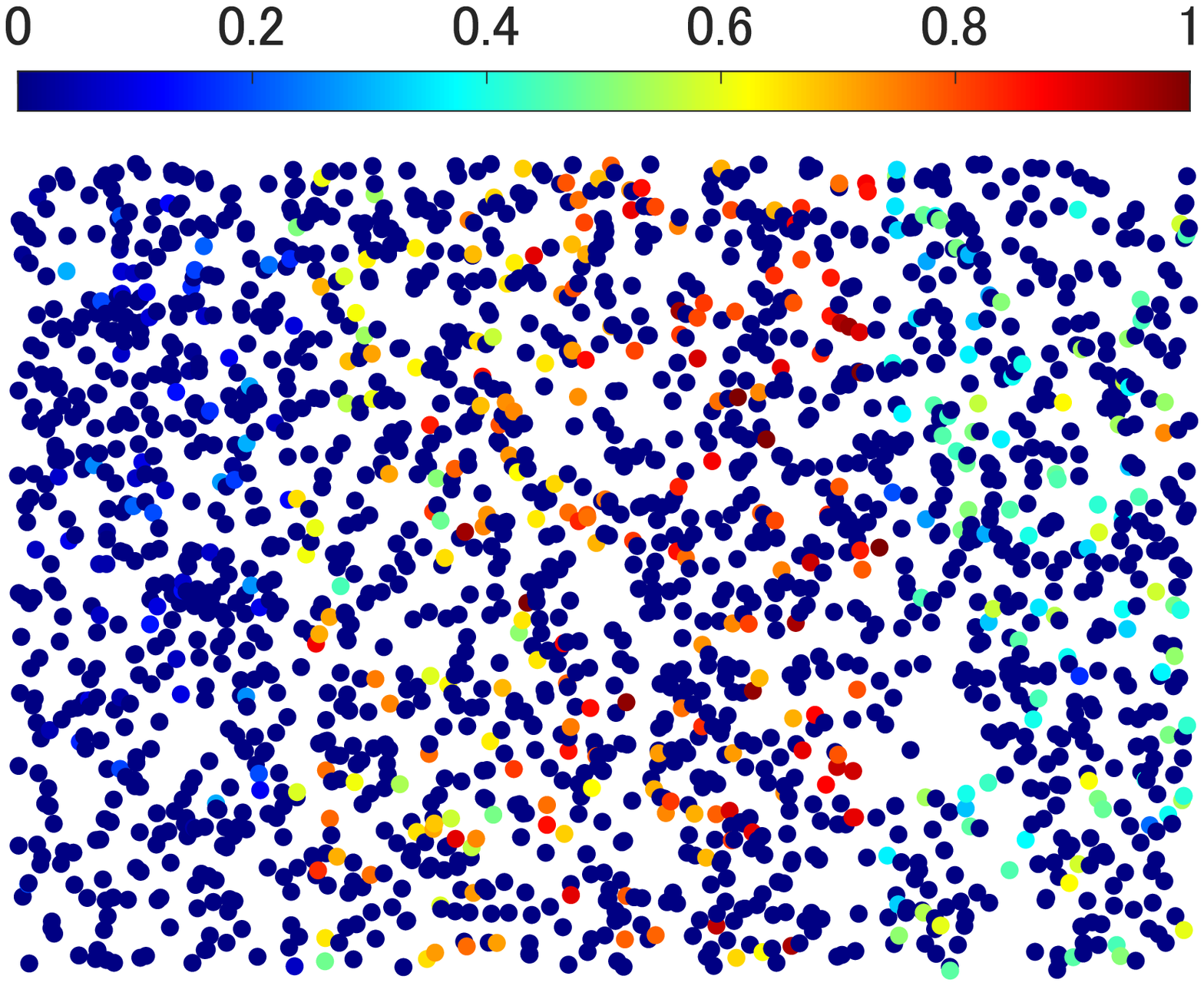}} 
		\end{minipage}
		\begin{minipage}{0.23\hsize}
			\centerline{\includegraphics[width=\hsize]{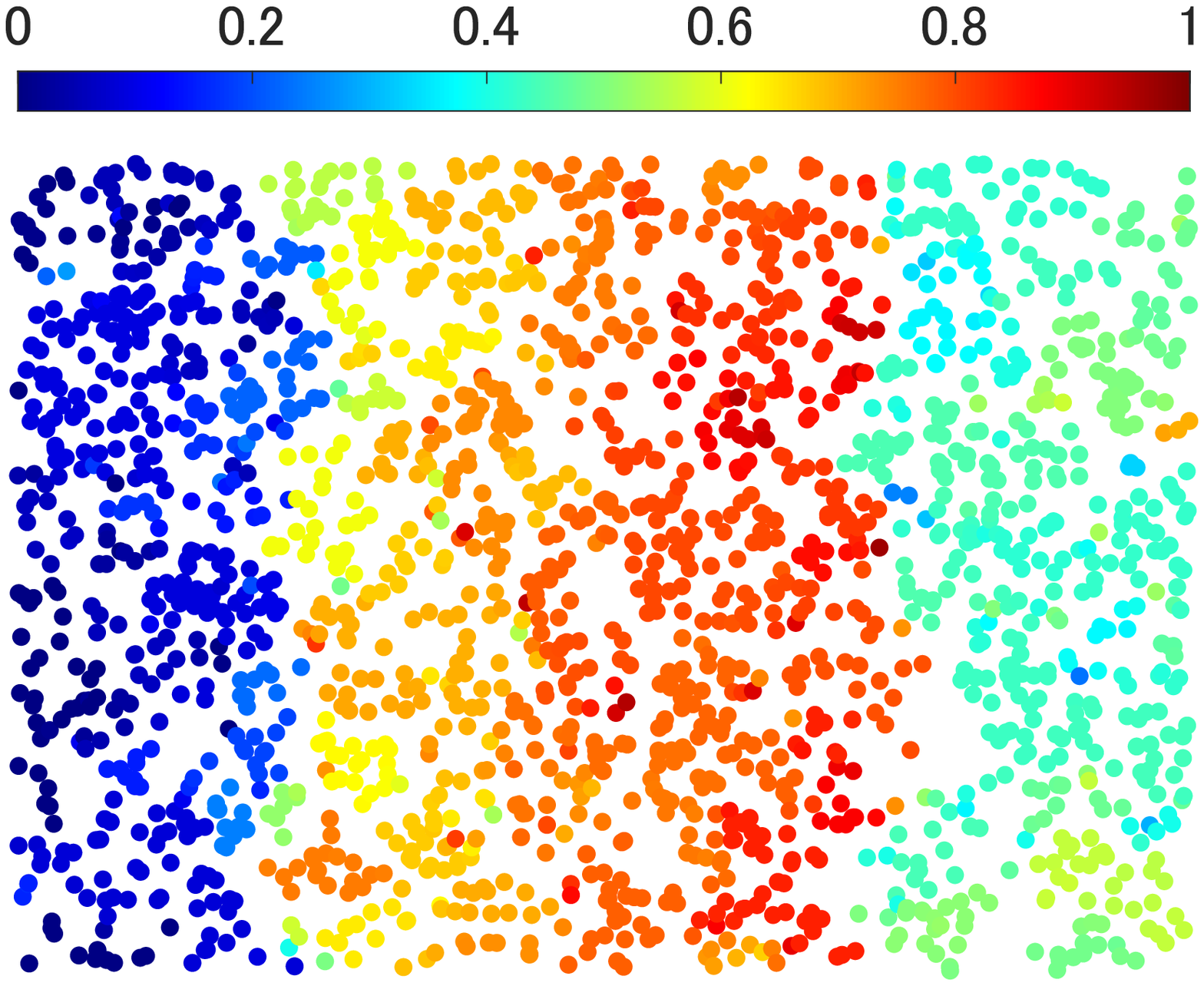}} 
		\end{minipage}
		\begin{minipage}{0.23\hsize}
			\centerline{\includegraphics[width=\hsize]{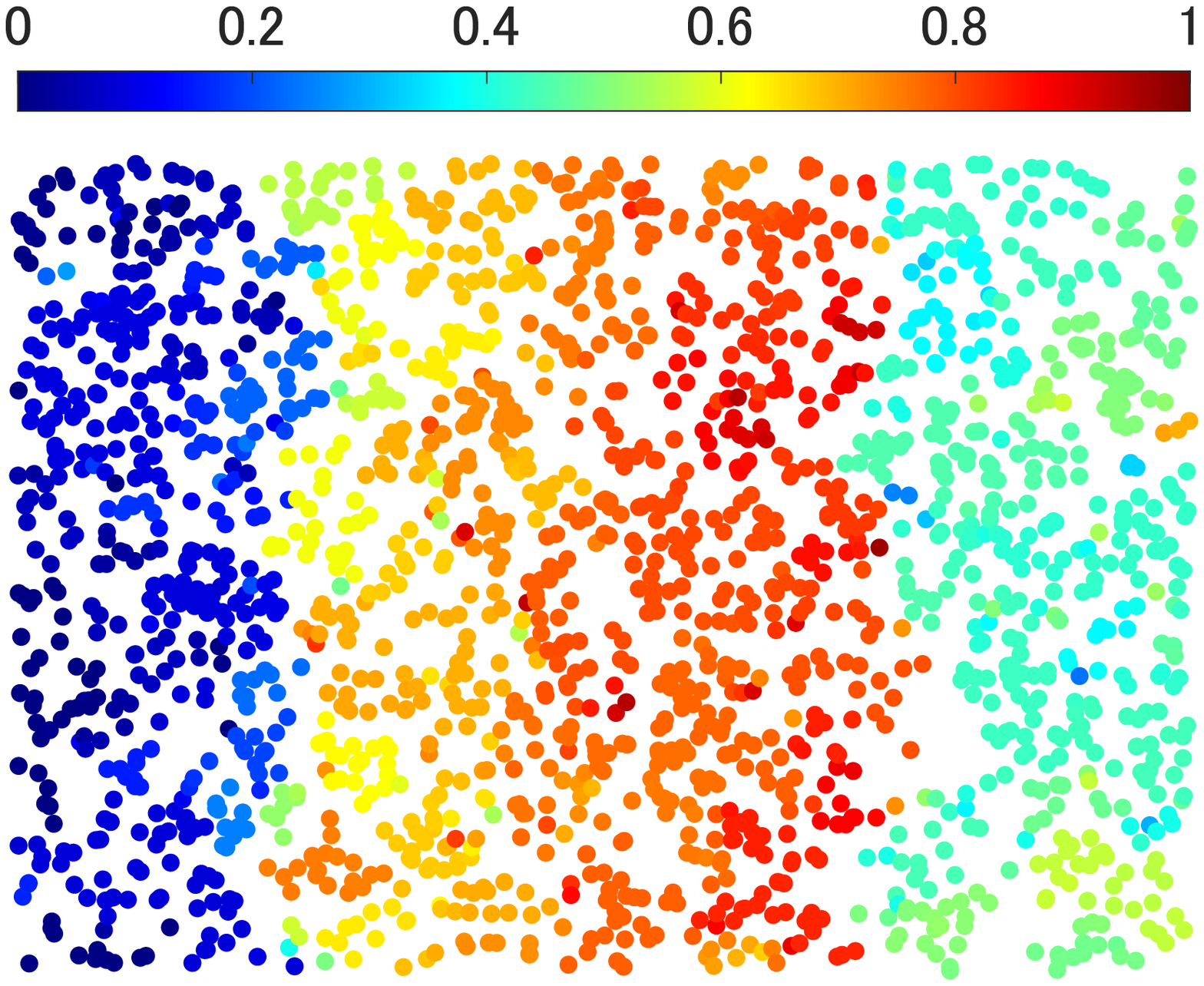}} 
		\end{minipage}
	
		\vspace{1mm}
	
		\begin{minipage}{0.23\hsize}
			\centerline{(a)}
		\end{minipage}
		\begin{minipage}{0.23\hsize}
			\centerline{(b) PSNR=$5.72$ [dB]}
		\end{minipage}
		\begin{minipage}{0.23\hsize}
			\centerline{(c) PSNR=$21.13$ [dB]}
		\end{minipage}
		\begin{minipage}{0.23\hsize}
			\centerline{(d) PSNR=$21.12$ [dB]}
		\end{minipage}
	
		\vspace{1mm}

		\begin{minipage}{0.23\hsize}
			\centerline{\includegraphics[width=\hsize]{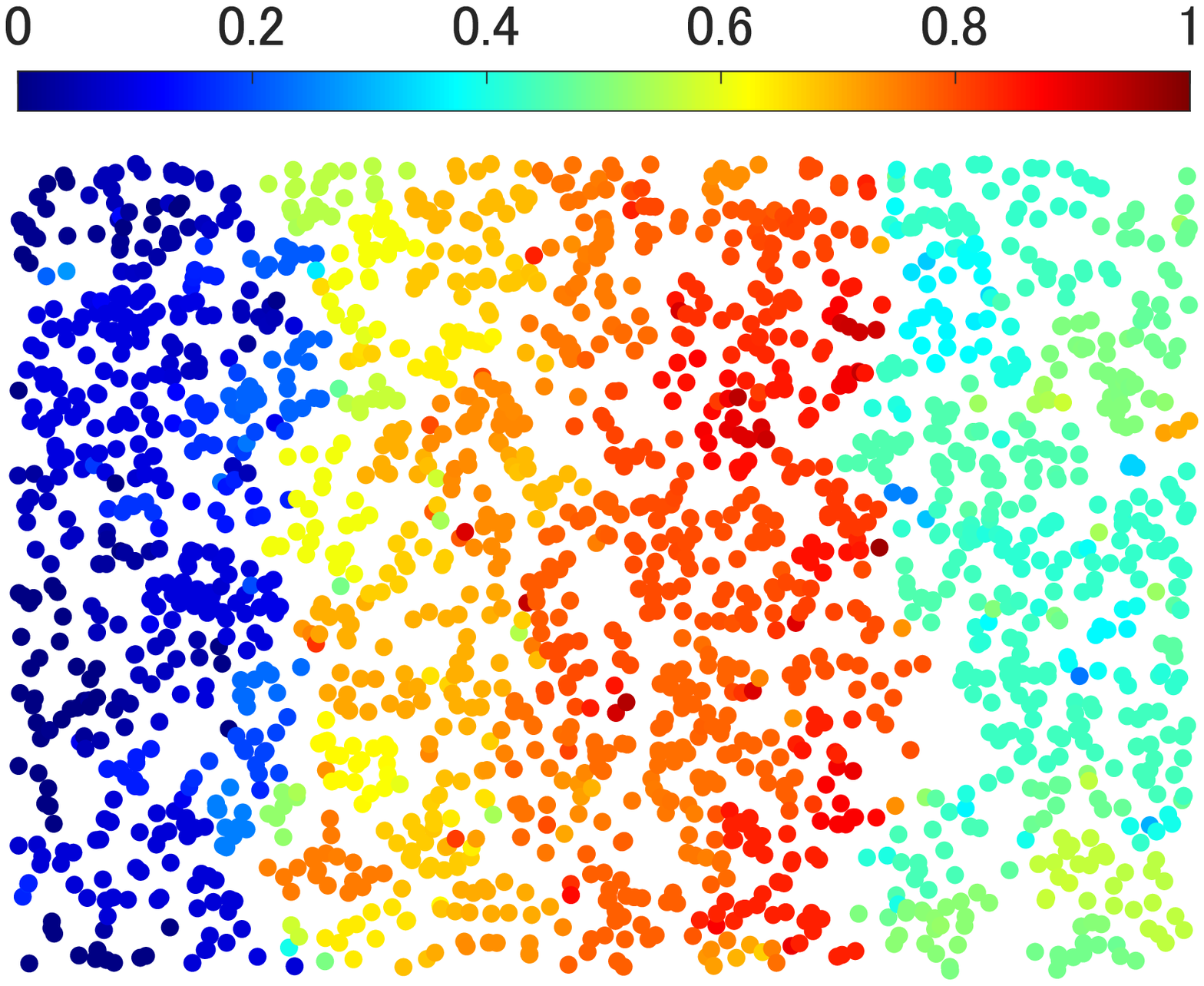}} 
		\end{minipage}
		\begin{minipage}{0.23\hsize}
			\centerline{\includegraphics[width=\hsize]{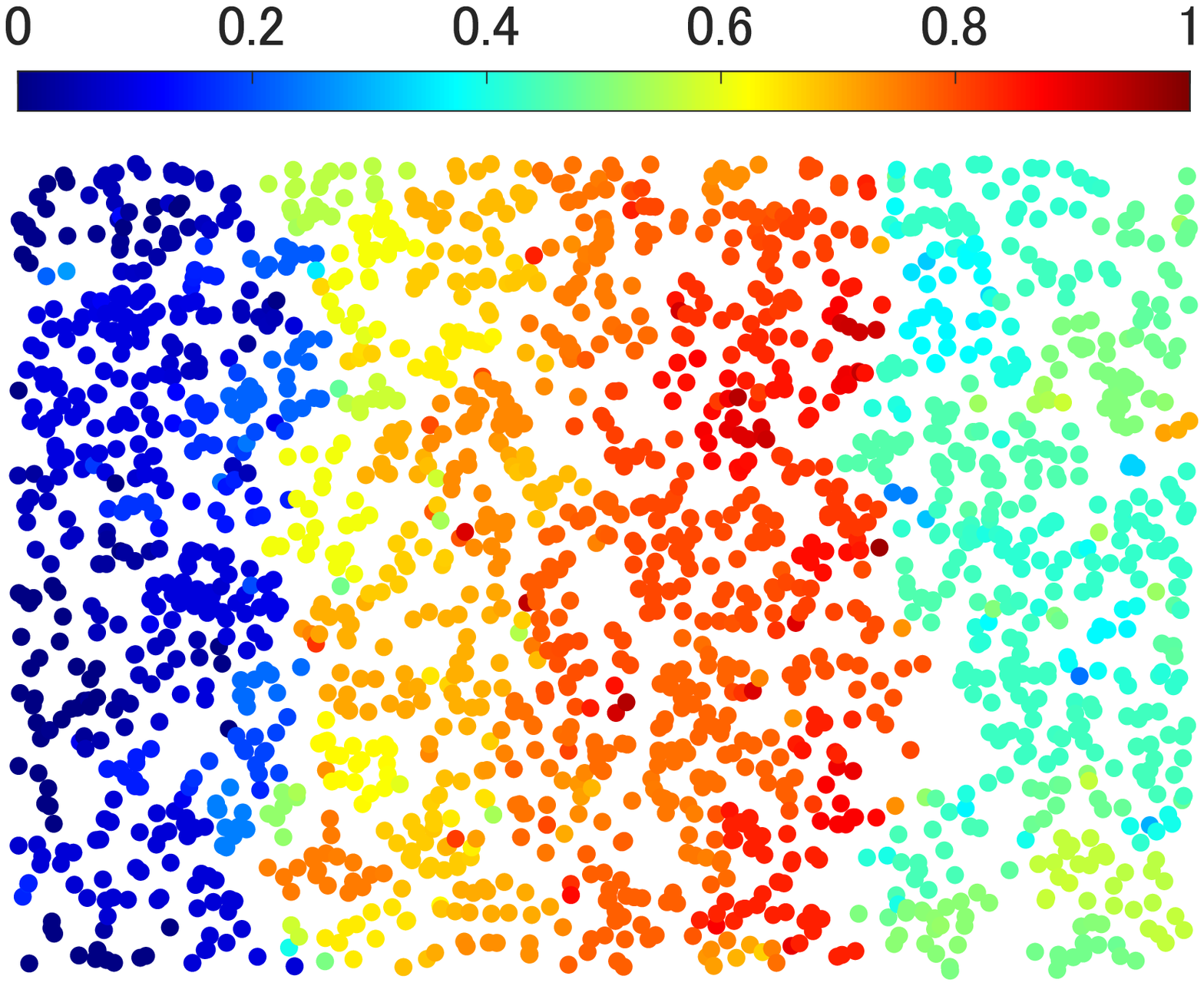}} 
		\end{minipage}
		\begin{minipage}{0.23\hsize}
			\centerline{\includegraphics[width=\hsize]{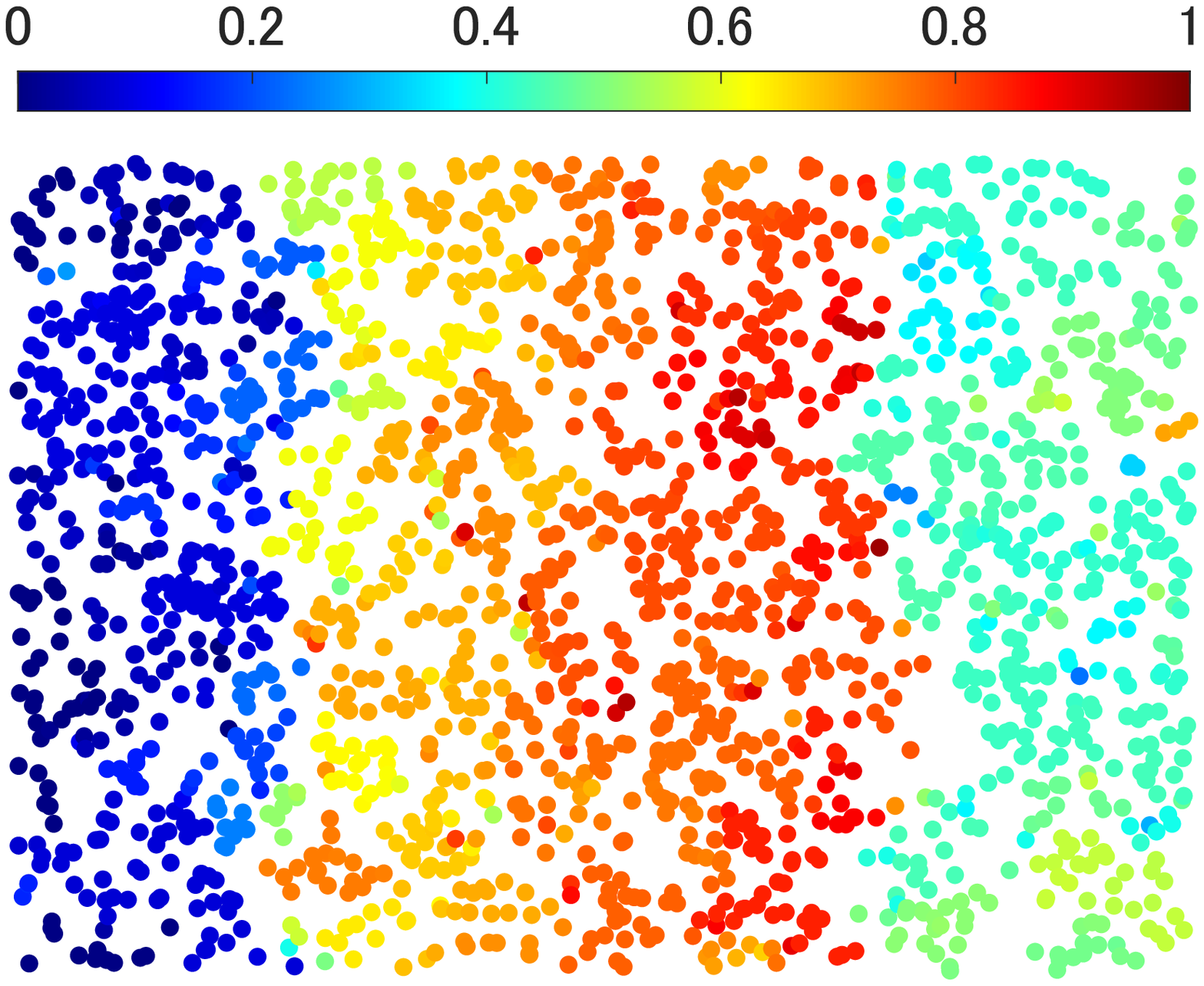}} 
		\end{minipage}
		\begin{minipage}{0.23\hsize}
			\centerline{\includegraphics[width=\hsize]{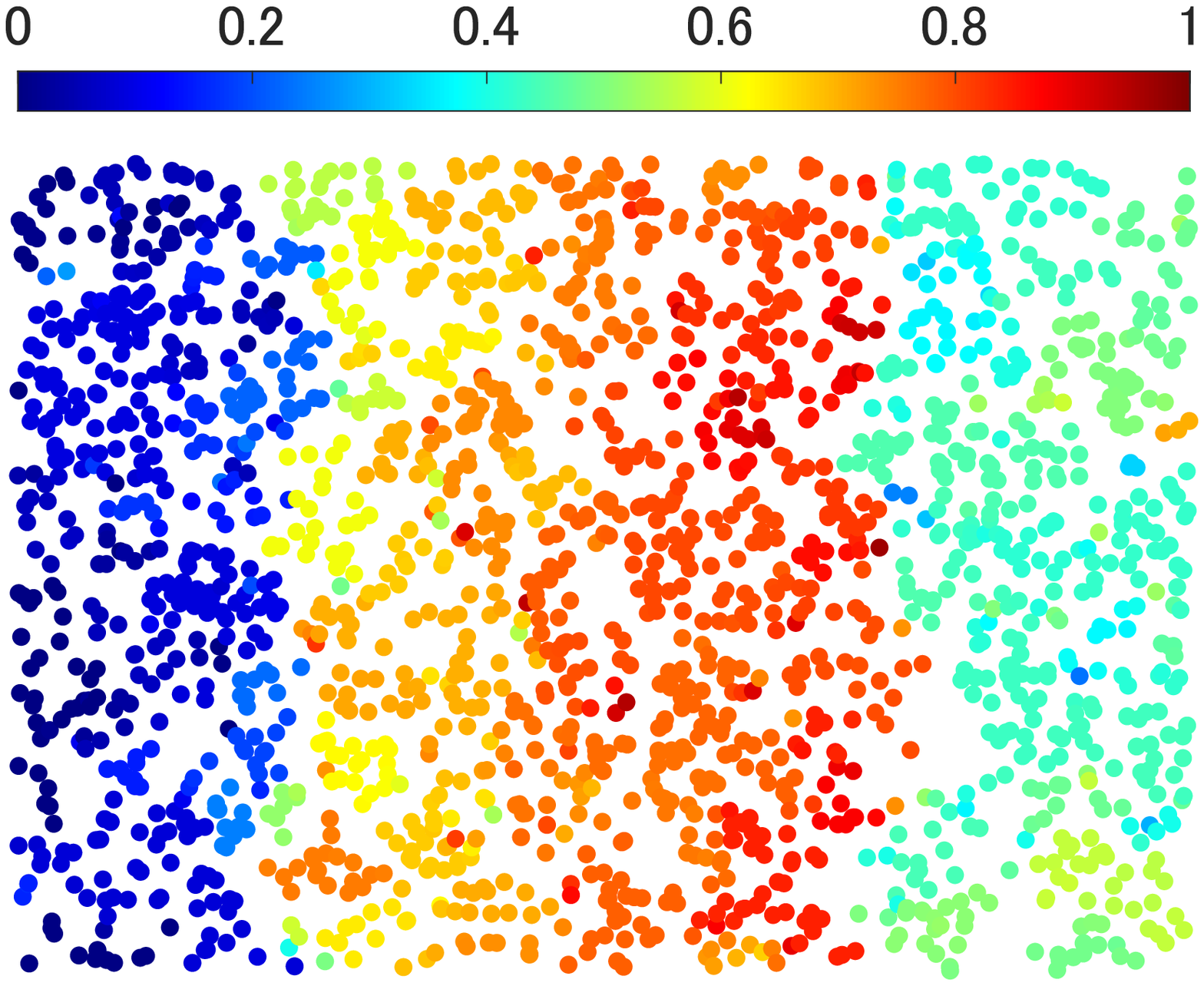}} 
		\end{minipage}
		
		\vspace{1mm}

		\begin{minipage}{0.23\hsize}
			\centerline{(e) PSNR=$21.12$ [dB]}
		\end{minipage}
		\begin{minipage}{0.23\hsize}
			\centerline{(f) PSNR=$21.14$ [dB]}
		\end{minipage}
		\begin{minipage}{0.23\hsize}
			\centerline{(g) PSNR=$21.12$ [dB]}
		\end{minipage}
		\begin{minipage}{0.23\hsize}
			\centerline{(h) PSNR=$21.12$ [dB]}
		\end{minipage}
	\end{center}

	\vspace{-4mm}
	
	\caption{Graph signal recovery results. (a): The ground truth signal. (b): The observed graph signal. (c): The graph signal estimated by P-PDS with \SP~\cite{PDS_1} ($\ParamPDS_{1} = 0.1$). (d): The graph signal estimated by P-PDS with \ASEW~\cite{DP-PDS}. (e): The graph signal estimated by P-PDS with \SPDP~\cite{SPDP} ($\ParamSPDP = 1$). (f): The graph signal estimated by P-PDS with \ONVWOnes (Ours). (g): The graph signal estimated by P-PDS with \ONVWTwos (Ours). (h): The graph signal estimated by P-PDS with \ONVWThrees (Ours).}
	\label{fig:results_GSR}
\end{figure*}

Fig.~\ref{fig:graph_Unm} plots iterations versus RMSE, Residual, and SNR and computational time versus RMSE, Residual, and SNR, respectively. 
In terms of iterations (Figs.~\ref{fig:graph_Unm} (a1), (b1), and (c1)), P-PDS with \SPDPs was very slow in all parameter cases.
P-PDSs with \SPs ($\ParamPDS_{1}=1$), \ONVWTwo, and \ONVWThrees were slightly slow, but P-PDSs with \SPs ($\ParamPDS_{1}=0.1$) and \ASEWs were not.
P-PDSs with \SPs ($\ParamPDS_{1}=0.01$), \SPs ($\ParamPDS_{1}=0.001$), and \ONVWOnes were fast. 
In terms of computational time (Figs.~\ref{fig:graph_Unm} (a2), (b2), and (c2)), P-PDS with \SPs and \ONVWs were similar to the results with respect to iterations. 
P-PDSs with \ASEWs and \SPDPs were very slow because they require the iterative algorithm to calculate the skewed proximity operator in each iteration of P-PDS.
At first glance, the curves generated by P-PDSs with \SPDPs ($\ParamSPDP = 1$, $0.1$, and $0.001$) may appear to converge to different SNRs. 
This is because they take enormous amounts of time to converge (in fact, the convergence times are too enormous to measure). 
Therefore, they do not converge to different SNRs.

Fig.~\ref{fig:results_Unm} shows the unmixing results and the SNR values [dB] obtained by P-PDS with \SPs ($\ParamPDS_{1}=0.001$), \ASEW, \SPDPs ($\ParamSPDP = 0.01$), \ONVWOne, \ONVWTwo, and \ONVWThree. 
The algorithm was run until satisfying the stopping criterion or reaching $10000$ iterations.
We can see that all results are almost the same in terms of the SNR and the visual qualities.

\subsection{Application to Graph Signal Recovery}
Graphs explicitly represent the irregular structures of data~\cite{GSP_Sandryhaila2013,GTV_Shumn_D_2013,GSP_review_Ortega2018}, such as traffic and sensor network data, geographical data, mesh data, and biomedical data.
The signals on the irregular structures are called graph signals.
Similar to classical signal processing, sampling of graph signals~\cite{GSSampling_Tanaka2020} is a leading research topic due to its numerous promising applications, for example, sensor placement, filter bank designs, traffic monitoring, and semi-supervised learning.
In graph signal recovery, which reconstructs original graph signals from sampled graph signals, it is assumed that graph signals have some properties, such as smoothness.
The smoothness of graph signals can be captured by graph total variation type regularizations~\cite{GTV_gilboa_2009,GTGV_Ono2015,GSP_Berger_P_2017}, which have been applied to various graph signal processing tasks~\cite{GTV_appl_LiZ_2017,GTV_appl_LiB_2020}.

\subsubsection{Problem Formulation}
We consider signals on weighted directed graphs $\SymbGraph = (\Vertex, \Edge, \MatWeight)$ with a vertex set $\Vertex = \{1, \ldots, \NumVertex\}$, an edge set $\Edge \subseteq \Vertex \times \Vertex$, and a weighted matrix $\MatWeight \in \RealNumSet^{\NumVertex \times \NumVertex}$.
The value $\MatWeightElem_{i,j}$ is designed to be large if the relation between vertices $i$ and $j$ is strong.
Graph signals are typically assumed to be smooth with respect to the graph $\SymbGraph$.
Based on the assumption, graph signal recovery methods often adopt the graph total variation (GTV)~\cite{GTV_Shumn_D_2013,GSP_Berger_P_2017}:
\begin{equation}
	\GTV{\VarOneGS} := \| \DiffGS\VarOneGS \|_{1, 2} = \sum_{\IndVertex = 1}^{\NumVertex}\|\VarTwoGS_{\IndVertex}\|_{2},
\end{equation}
where $\DiffGS$ is the graph difference operator defined as follows. 
Let $\DiffGS\VarOneGS = [\VarTwoGS_{1}^{\top}, \ldots, \VarTwoGS_{\NumVertex}^{\top}]$, then each $\VarTwoGS_{i}$ consists of the weighted differences between the graph signal value $\VarOneElemGS_{\IndVertex}$ at an $i$th vertex and the graph signal values $\VarOneElemGS_{\IndVertexCor}$ ($\forall \IndVertexCor \in \NeiVertexI{\IndVertex} := \{ k \in \Vertex \: | \: \MatWeightElem_{i,k} \neq 0 \}$) at its connected vertices $\NeiVertexI{\IndVertex}$, i.e., 
\begin{equation}
	[\VarTwoGS_{\IndVertex}]_{\IndVertexCor} := (\VarOneElemGS_{j} - \VarOneElemGS_{\IndVertex})\MatWeightElem_{\IndVertex,\IndVertexCor}, ~ (\forall j \in \NeiVertexI{\IndVertex}).
\end{equation}
By weighting the difference between $\VarOneElemGS_{\IndVertex}$ and $\VarOneElemGS_{\IndVertexCor}$ by $\MatWeightElem_{\IndVertex,\IndVertexCor}$, GTV can capture the graph signal smoothness that the difference of graph signal values is small as the relation of their vertices is strong.

\begin{table}[t]
	\begin{center}
		\caption{The Preconditioners by \ONVWs for Graph Signal Recovery.}
		\label{tab:OVDP_GSR}
		\vspace{-2mm}
		\begin{tabular}{cccc}
			\toprule
			& $\Precon_{1,1}$ 
			& $\Precon_{2,1}$ & $\Precon_{2,2}$ \\
			\cmidrule(lr){2-4}
			
			\vspace{1mm}
			\ONVWOne & $\frac{1}{\NormOp{\DiffGS}^{2} + 1^{2}}\mathbf{I}$ & $\mathbf{I}$ & $\mathbf{I}$ \\
			
			\vspace{1mm}
			\ONVWTwo & $\frac{1}{\NormOp{\DiffGS} + 1}\mathbf{I}$ & $\frac{1}{\NormOp{\DiffGS}}\mathbf{I}$ & $\mathbf{I}$ \\ 
			
			\ONVWThree & $\frac{1}{2}\mathbf{I}$ & $\frac{1}{\NormOp{\DiffGS}^{2}}\mathbf{I}$ & $\mathbf{I}$ \\
			\bottomrule
		\end{tabular}
	\end{center}
	\vspace{-3mm}
\end{table}


\begin{table*}[t]
	\begin{center}
		\caption{The Number of Iterations to Meet the Stopping Criteria. XXX* Means That the Method Requires More Than XXX Iterations.}
		\label{tab:iteration_results}
		\scalebox{0.85}{
			\vspace{-2mm}
			\begin{tabular}{ccccccccccccc}
				\toprule
				& \multicolumn{12}{c}{Methods} \\
				\cmidrule(lr){2-13}
				& \multicolumn{4}{c}{\SP} & \ASEW & \multicolumn{4}{c}{\SPDP} & \textbf{\ONVWOne} & \textbf{\ONVWTwo} & \textbf{\ONVWThree} \\
				\cmidrule(lr){2-5}\cmidrule(lr){7-10}
				& ($\ParamPDS_{1}$ = 1) & ($\ParamPDS_{1}$ = 0.1) & ($\ParamPDS_{1}$ = 0.01) & ($\ParamPDS_{1}$ = 0.001) &  & ($\ParamSPDP$ = 1) & ($\ParamSPDP$ = 0.1) & ($\ParamSPDP$ = 0.01) & ($\ParamSPDP$ = 0.001) & & & \\
				\midrule 
				Mixed noise removal &
				10000* & 4736 & 10000* & 10000* &
				3323 & 
				10000* & 1552 & 10000* & 10000* & 
				5470 & 3325 & 9755 \\
				Unmixing &
				2943 & 395 & 350 & 349 &
				300 & 
				10000* & 10000* & 10000* & 10000* & 
				350 & 525 & 4709 \\
				Graph signal recovery &
				3625 & 806 & 1937 & 10000* &
				448 & 
				194 & 396 & 4129 & 10000* & 
				998 & 1846 & 3546 \\
				Average &
				5523* & 1979 & 4096* & 6783* & 
				1357 & 
				6731* & 3983* & 8043* & 10000* & 
				2273 & 1899 & 6003 \\
				\bottomrule
			\end{tabular}
		}
	\end{center}
	\vspace{-3mm}
\end{table*}


\begin{table*}[t]
	\begin{center}
		\caption{Running time [s] to Meet the Stopping Criteria. XXX* Means That the Method Requires More Than XXX [s].}
		\label{tab:times_results}
		\scalebox{0.85}{
			\vspace{-2mm}
			\begin{tabular}{ccccccccccccc}
				\toprule
				& \multicolumn{12}{c}{Methods} \\
				\cmidrule(lr){2-13}
				& \multicolumn{4}{c}{\SP} & \ASEW & \multicolumn{4}{c}{\SPDP} & \textbf{\ONVWOne} & \textbf{\ONVWTwo} & \textbf{\ONVWThree} \\
				\cmidrule(lr){2-5}\cmidrule(lr){7-10}
				& ($\ParamPDS_{1}$ = 1) & ($\ParamPDS_{1}$ = 0.1) & ($\ParamPDS_{1}$ = 0.01) & ($\ParamPDS_{1}$ = 0.001) &  & ($\ParamSPDP$ = 1) & ($\ParamSPDP$ = 0.1) & ($\ParamSPDP$ = 0.01) & ($\ParamSPDP$ = 0.001) & & & \\
				\midrule 
				Mixed noise removal &
				1000* & 111.86 & 333.72 & 1000* &
				230.40 & 
				1000* & 1000* & 1000* & 1000* & 
				130.70 & 79.35 & 231.78 \\
				Unmixing &
				3.52 & 3.97 & 3.60 & 3.60 &
				64.44 & 
				1000* & 1000* & 1000* & 1000* & 
				11.38 & 5.44 & 46.13 \\
				Graph signal recovery &
				7.54 & 1.65 & 4.00 & 40.61 &
				1000* & 
				607.90 & 450.60 & 1000* & 1000* & 
				2.04 & 3.73 & 7.21 \\
				\bottomrule
			\end{tabular}
		}
	\end{center}
	\vspace{-3mm}
\end{table*}

Consider that an observed graph signal $\ObsGS\in\RealNumSet^{\NumSampVertex}$ is modeled by
\begin{equation}
	\label{eq:obs_model_GSR}
	\ObsGS = \MatSampGS \bar{\GTGS} + \NoiseGS,
\end{equation}
where $\bar{\GTGS}\in\RealNumSet^{\NumVertex}$, $\NoiseGS\in\RealNumSet^{\NumSampVertex}$, and $\MatSampGS\in\{0, 1\}^{\NumSampVertex \times \NumVertex}$ are the true graph signal of interest, random additive noise, and the sampling matrix, respectively.
Based on this observation model, the GTV regularized graph signal recovery problem is formulated as the following convex optimization problem~\cite{GSP_Berger_P_2017}:
\begin{equation}
	\label{prob:form_GSR}
	\min_{\GTGS} 
	\GTV{\DiffGS\GTGS} \: 
	\st \: 
	\MatSampGS \GTGS \in \BallFidelGS.
\end{equation}
The hard constraint guarantees the $\ell_{2}$ data fidelity to the observed signal $\ObsGS$ with the radius $\ParmFidelGS$.

By using the indicator function (see Eq.~\eqref{eq:indicator_function}) of $\BallFidelGS$, Prob.~\eqref{prob:form_GSR} is reduced to Prob.~\eqref{prob:general_form_of_optimization} via the following reformulation:
\begin{align}
	\min_{\GTGS, \VarDualGS_{1}, \VarDualGS_{2}} \: &
	\|\VarDualGS_{1}\|_{1, 2} 
	+ \FuncIndi{\BallFidelGS}{\VarDualGS_{2}} \nonumber \\
	\st \: & 
	\begin{cases}
		\VarDualGS_{1} = \DiffGS\GTGS, \\
		\VarDualGS_{2} = \MatSampGS\GTGS. 
	\end{cases}
	\label{prob:form_GSR_PDS}
\end{align}
Applying Algorithm~\ref{algo:P_PDS} to Prob.~\eqref{prob:form_GSR_PDS}, we can compute an optimal solution of Prob.~\eqref{prob:form_GSR}.
Since the function $\| \cdot \|_{1, 2}$ is not separable for each element of the input variable, an iterative algorithm is needed for the computation of their skewed proximity operators relative to the metric induced by the preconditioners of \ASEWs and \SPDPs in~\eqref{eq:precon_ASEW}.
Here, the preconditioners designed by \ONVWs are given as in Tab.~\ref{tab:OVDP_GSR}.
According to~\cite{GSP_Berger_P_2017}, an upper bound of the operator norm $\NormOp{\DiffGS}$ can be derived by
\begin{equation}
	\NormOp{\DiffGS} \leq 2 \max_{i \in \Vertex} \sum_{j\in\Vertex}(\MatWeightElem_{i,j}^{2} + \MatWeightElem_{j,i}^{2}).
\end{equation}
An upper bound of the norm of the sampling matrix is one, i.e., $\NormOp{\MatSampGS} = 1$.

\subsubsection{Experimental Results}
For \SP, $\PreconOpNormSP$ in~\eqref{eq:SP_expriment} was set as 
\begin{equation}
	\label{eq:SP_GS}
	\PreconOpNormSP = \sqrt{\NormOp{\DiffGS}^{2} + 1},
\end{equation}
because the following inequality holds due to the inequality of the operator norms of block matrices~\cite{norm_inequality}:
\begin{equation}
	\NormOp{\begin{bmatrix} \DiffGS \\ \MatSampGS \end{bmatrix}}^{2} 
	\leq \NormOp{\DiffGS}^{2} + \NormOp{\MatSampGS}^{2}
	\leq \NormOp{\DiffGS}^{2} + 1.
\end{equation}
The preconditioners by \ASEWs in~\eqref{eq:precon_ASEW} for Prob.~\eqref{prob:form_GSR} are
\begin{align}
	[\Precon_{1, 1}]_{i, i} 
	& = \frac{1}{\sum_{j = 1}^{\NumVertex\NumVertex}|\MatWeightElem_{i, j}| + \sum_{k = 1}^{\NumSampVertex}\MatSampElemGS_{k, i}}, \: (\forall i = 1, \ldots, \NumVertex), \nonumber \\
	[\Precon_{2, 1}]_{i, i} 
	& = \frac{1}{2\sum_{j = 1}^{\NumVertex}|\MatWeightElem_{j, i}|}, \: (\forall i = 1, \ldots, \NumVertex\NumVertex), \nonumber \\
	[\Precon_{2, 2}]_{i, i} 
	& = 1, \: (\forall i = 1, \ldots, \NumSampVertex).
	\label{eq:ASP_GSR}
\end{align}
For \SPDP, the preconditioners in~\eqref{eq:precon_SPDP_2} were used since the number of dual variables is two.

We constructed a random sensor graph $\SymbGraph$ by using GSPBox~\cite{perraudin2014gspbox}, then generated a noiseless piece-wise smooth graph signal on the graph with $\NumVertex = 2000$ vertices.
The observed graph signal was obtained by adding white Gaussian noise with $0.1$ of the standard deviation $\StanDivGaussGS$ and by sampling it with $0.2$ of the sampling rate ($\NumSampVertex = 0.2\NumVertex$).
The parameter $\ParmFidelGS$ was set as $\ParmFidelGS = 0.9\StanDivGaussGS\sqrt{\NumSampVertex}$.
For the quantitative evaluation of recovery qualities, we used the Peak Signal-to-Noise Ratio (PSNR):
\begin{equation}
	\mathrm{PSNR}
	:= 10\log_{10}\left(\frac{\NumVertex}{\|\bar{\mathbf{u}}-\mathbf{u}^{(\InnerIter)}\|_{2}^{2}}\right),
\end{equation}

Fig.~\ref{fig:graph_GSR} plots iterations versus RMSE, Residual, and PSNR and computational time versus RMSE, Residual, and PSNR, respectively. 
In terms of iterations (Figs.~\ref{fig:graph_GSR} (a1), (b1), and (c1)), P-PDSs with \SPs ($\ParamPDS_{1}=0.001$) and \SPDPs ($\ParamSPDP = 0.001$) were very slow.
P-PDSs with \SPs ($\ParamPDS_{1}=1$), \SPDPs ($\ParamSPDP = 0.01$), \ONVWThrees were not slow but not fast.
P-PDSs with \SPs ($\ParamPDS_{1}=0.1$), \SPs ($\ParamPDS_{1}=0.01$), \ASEW, \SPDPs ($\ParamSPDP = 1$), \SPDPs ($\ParamSPDP = 0.1$), \ONVWOne, and \ONVWTwos were fast. 
In terms of computational time (Figs.~\ref{fig:graph_GSR} (a2), (b2), and (c2)), P-PDS with \SPs and \ONVWs were similar to the results with respect to iterations. 
P-PDSs with \ASEWs and \SPDPs were very slow because they require the iterative algorithm to calculate the skewed proximity operator.

Fig.~\ref{fig:results_GSR} shows the recovery results and the PSNR values [dB] obtained by P-PDS with \SPs ($\ParamPDS_{1}=0.1$), \ASEW, \SPDPs ($\ParamSPDP = 1$), \ONVWOne, \ONVWTwo, and \ONVWThree. 
The algorithm was run until satisfying the stopping criterion or reaching $10000$ iterations.
We can see that all results are almost the same in terms of the PSNR and the visual qualities.

\subsection{Discussion}
For discussion based on numerical values, we compare the number of iterations (Tab.~\ref{tab:iteration_results}) and running time (Tab.~\ref{tab:times_results}) to satisfy the stopping criteria in Tab.~\ref{tab:convergence_conditions}.

The appropriate parameter for \SPs ($\ParamPDS_{1}$) varied depending on the optimization problem and were  $0.1$ for mixed noise removal, $0.01$ and $0.001$ for unmixing, and $0.1$ and $0.01$ for graph signal recovery. 
If $\ParamPDS_{1}$ is adjusted appropriately, as in the case of the unmixing experiments ($\ParamPDS_{1}$ = $0.01$ and $0.001$), P-PDS with \SPs can converge faster than the automatic preconditioner design methods (\ASEWs and \ONVW). 
However, no parameter results in fast convergence for any optimization problem, and the convergence might be extremely slow, such as at $0.01$ and $0.001$ for mixed noise removal, at $1$ for unmixing, and at $0.001$ for graph signal recovery. 
Therefore, $\ParamPDS_{1}$ needs to be manually adjusted according to each problem.

P-PDS with \ASEWs was the best in terms of the average number of iterations, and P-PDS with PDP ($\ParamSPDP$ is adjusted) resulted in a small number of iterations to converge for both graph signal recovery and mixed noise removal. 
However, for the unmixing experiments, P-PDS with PDP required a more significant number of iterations to converge than P-PDS with SP ($\ParamPDS_{1}$ = $0.01$ and $0.001$) and \ONVW. 
We speculate that this is because the optimization problem of unmixing is relatively complicated; it involves an endmember matrix, while the optimization problems of mixed noise removal and graph signal recovery only include relatively simple difference operators and random sampling matrices in their optimization problems. 
Although P-PDSs with \ASEWs and \SPDPs were fast in the number of iterations, they took a much longer running time to converge. 
This is due to the fact that they require iterative algorithms such as FISTA to compute the skewed proximity operator in each iteration of P-PDS.
Incidentally, since the internal iterations of FISTA vary depending on the task and parameters (e.g., $\tau$), the execution time of P-PDS may be long relative to the number of iterations to convergence. 
For example, P-PDS with \SPDPs ($\tau = 1$) required fewer iterations but a longer running time than P-PDS with \SPDPs ($\tau = 0.1$). 
In addition, P-PDS with \ASEWs took a very long running time per iteration in the graph signal recovery experiment, while it took a short running time in the unmixing experiments.

P-PDSs with \ONVWs achieved good convergence speed in both the number of iterations and the running time thanks to a diagonal preconditioning technique based on the problem structure. 
In addition, they maintain the proximability of the functions, resulting in fast running time.
P-PDS with \ONVWTwos was fast on average in the number of iterations. 
Moreover, P-PDS with \ONVWTwos produced the fastest result in terms of running time for the mixed noise removal experiment. 
P-PDS with \ONVWOnes was faster than P-PDS with \ONVWTwos and \ONVWThrees for the unmixing and graph signal recovery experiments. 
Futhermore, the preconditioners of \ONVWs can be easily calculated in the mixed noise removal case whose optimization problem incorporates the linear operators implemented not as explicit matrices.

These results indicate the following conclusions. 
\begin{itemize}
	\item \SPs and \SPDPs are effective for cases where preconditioners are easily adjusted. In particular, PDP is very effective for the cases where the structure of an optimization problem is simple and the calculation of an inner iteration is efficient.
	\item \ASEWs is applicable to the cases where the structure of an optimization problem is simple, the calculation of an inner iteration is efficient, and the optimization problem only contains linear operators implemented as the represented matrix.
	\item Our \ONVWs can determine effective preconditioners regardless of whether or not the above conditions are satisfied.
	Specifically, for the signal estimation problem that can be handled by \ASEW, our \ONVWs was several hundred times faster than \ASEW.
	\item In addition, P-PDS with our \ONVWs required fewer iterations on average than P-PDSs with \SPs or \SPDP, which require manual adjustments.
\end{itemize}   

\section{Conclusion}
\label{sec:conclusion}
We have proposed \ONVW, which automatically and easily designs preconditioners in a variable-wise manner when a given optimization problem incorporates linear operators represented not as explicit matrices. We also proved the convergence of P-PDS with \ONVW. Applications of our method to three signal estimation tasks have been provided with experimental comparison, where we have shown that our method achieved the fast convergence speed on average and raised the examples of signal processing tasks that \ONVWs is effective to be applied.

\appendix[Proof of Lemma~\ref{lem:matrix_decomp}]
\label{append:proof_lemma1}
\begin{proof}
	Let $r$ be the rank of $\mathbf{A}$ and $\sigma_{1}(\mathbf{A}), \ldots, \sigma_{r}(\mathbf{A})$ be the singular values of $\mathbf{A}$. Then, $\mathbf{A}$ can be decomposed as
	\begin{equation}
		\mathbf{A} = \mathbf{U} \mathbf{\Sigma} \mathbf{V}^{*},
	\end{equation}
	where $\mathbf{U}\in\RealNumSet^{m\times r}$ and $\mathbf{V} \in \RealNumSet^{n\times r}$ satisfy $\mathbf{U}^{*}\mathbf{U} = \mathbf{I}$ and $\mathbf{V}^{*}\mathbf{V} = \mathbf{I}$. 
	Then, we introduce an $r \times r$ unitary matrix $\mathbf{W}$ and define $\mathbf{B}$ and $\mathbf{C}$ as
	\begin{equation}
		\mathbf{B} = \mathbf{U} \mathbf{\Sigma}^{1 - \ParamONP} \mathbf{W}^{*}, 
		\mathbf{C} = \mathbf{W} \mathbf{\Sigma}^{\ParamONP} \mathbf{V}^{*},
	\end{equation}
	where $\mathbf{\Sigma}^{1 - \ParamONP} = \mathrm{diag}(\sigma_{1}(\mathbf{A})^{1 - \ParamONP}, \ldots, \sigma_{r}(\mathbf{A})^{1 - \ParamONP})$ and $\mathbf{\Sigma}^{\ParamONP} = \mathrm{diag}(\sigma_{1}(\mathbf{A})^{\ParamONP}, \ldots, \sigma_{r}(\mathbf{A})^{\ParamONP})$.
	It is clear that $\mathbf{A} = \mathbf{BC}$.
	In turn, we obtain from the definition that
	\begin{align}
		\| \mathbf{Bx} \|_{2}^{2} 
		& = \|\mathbf{\Sigma}^{1 - \ParamONP}\mathbf{W}^{*}\mathbf{x}\|_{2}^{2} \nonumber \\
		& \leq \sigma_{1}(\mathbf{A})^{2 - 2\ParamONP} \|\mathbf{W}^{*}\mathbf{x}\|_{2}^{2} \nonumber \\ 
		& = \sigma_{1}(\mathbf{A})^{2 - 2\ParamONP} \|\mathbf{x}\|_{2}^{2}.
	\end{align}
	Hence
	\begin{equation}
		\NormOp{\mathbf{B}} = \sup_{\mathbf{x} \neq \ZeroElem} \frac{\|\mathbf{Bx}\|_{2}}{\|\mathbf{x}\|_{2}} = \sigma_{1}(\mathbf{A})^{1 - \ParamONP}.
	\end{equation}
	Arguing similarly, $\mathbf{C}$ satisfies $\NormOp{\mathbf{C}} = \sigma_{1}(\mathbf{A})^{\ParamONP}$.
	
	\begin{flushright} $\square$ \end{flushright}
\end{proof}

%





\ifCLASSOPTIONcaptionsoff
  \newpage
\fi



%
\bibliographystyle{IEEEtran}



%

%

\begin{IEEEbiography}[{\includegraphics[width=1in,height=1.25in,clip,keepaspectratio]{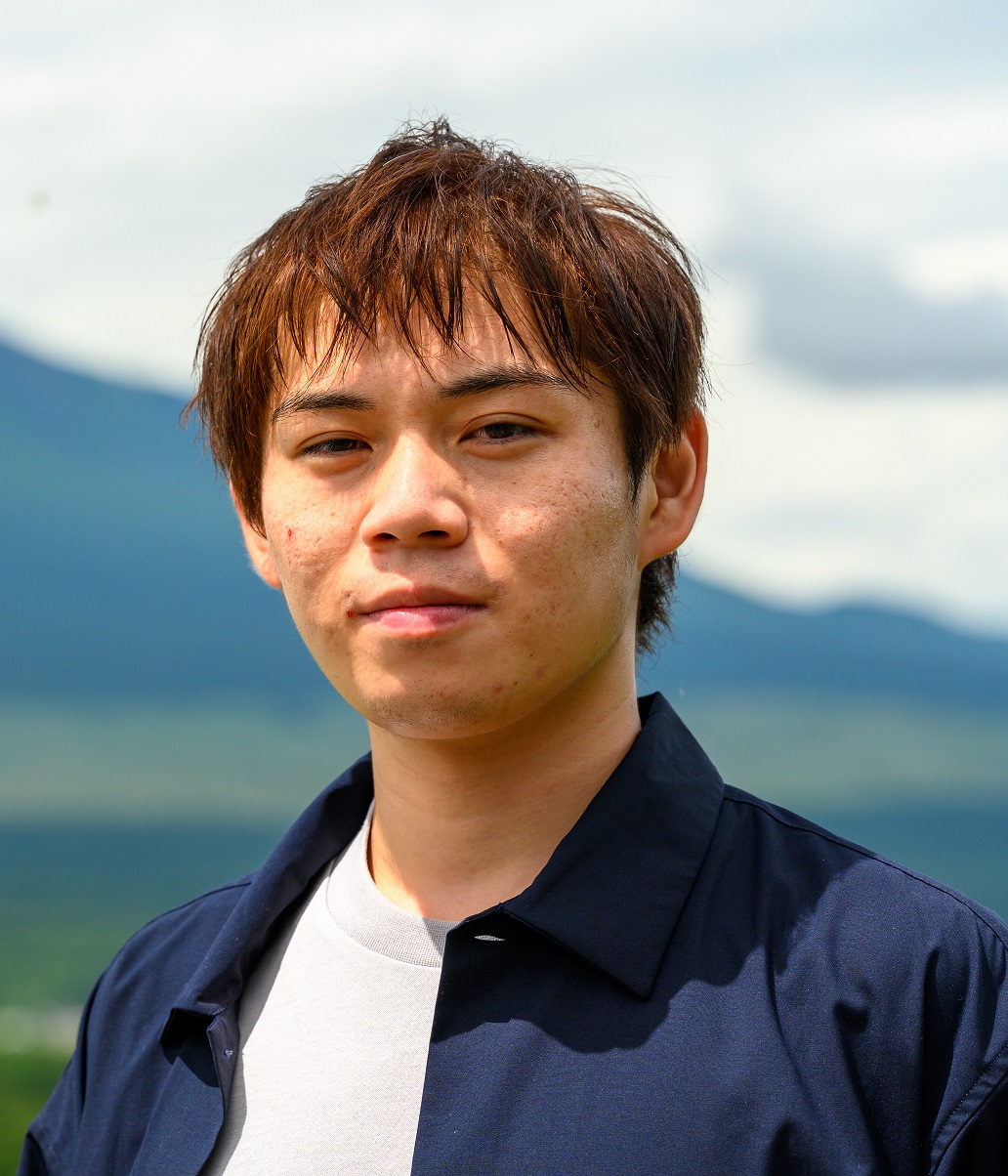}}]{Kazuki Naganuma}
Kazuki Naganuma (S’21) received a B.E. degree and M.E. degree in Information and Computer Sciences in 2020 from the Kanagawa Institute of Technology and from the Tokyo Institute of Technology, respectively.

He is currently pursuing an Ph.D. degree at the Department of Computer Science in the Tokyo Institute of Technology. His current research interests are in signal and image processing and optimization theory.
\end{IEEEbiography}

\begin{IEEEbiography}[{\includegraphics[width=1in,height=1.25in,clip,keepaspectratio]{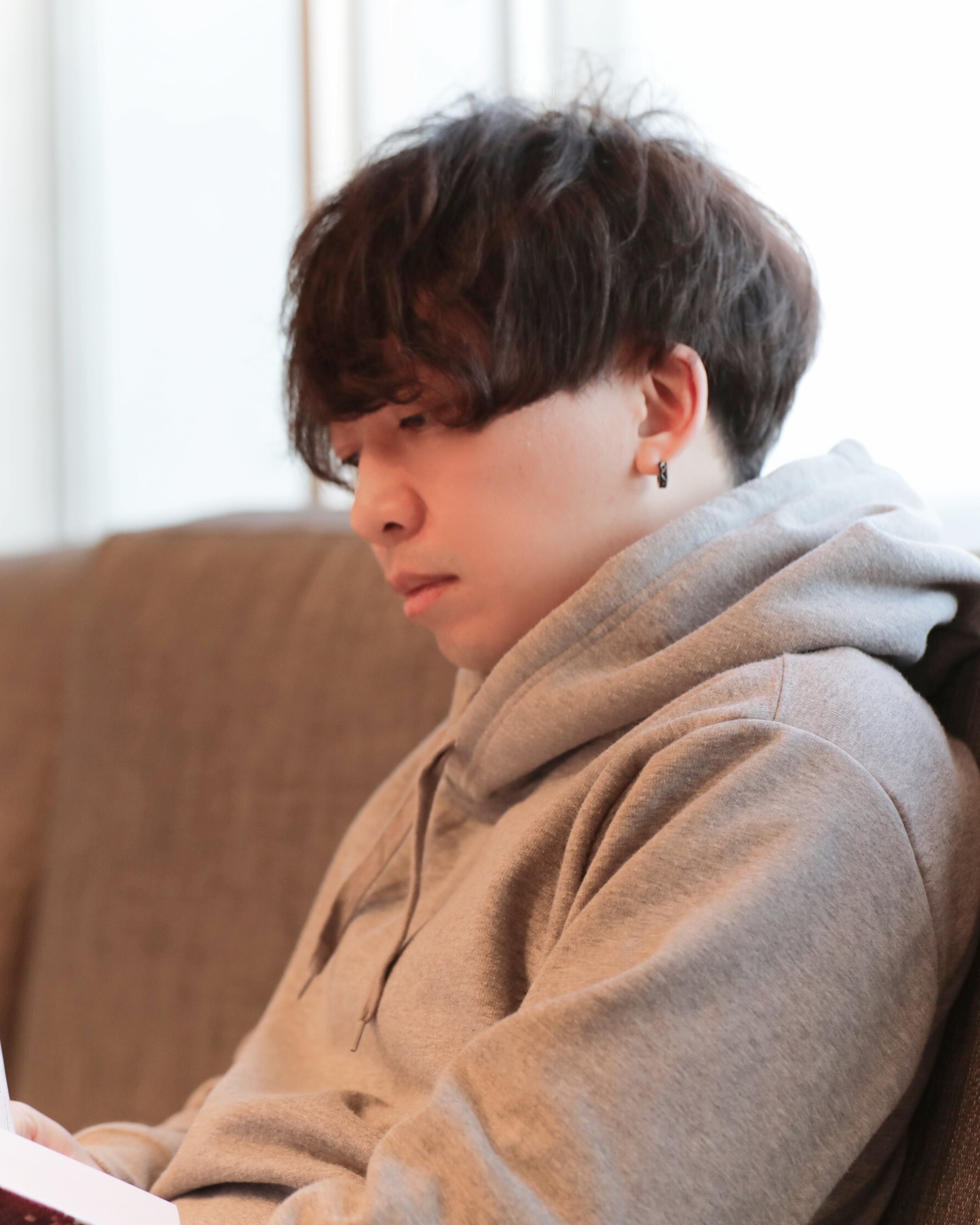}}]{Shunsuke Ono}
(S’11–M’15) received a B.E. degree in Computer Science in 2010 and M.E. and Ph.D. degrees in Communications and Computer Engineering in 2012 and 2014 from the Tokyo Institute of Technology, respectively.

From April 2012 to September 2014, he was a Research Fellow (DC1) of the Japan Society for the Promotion of Science (JSPS). He is currently an Associate Professor in the Department of Computer Science, School of Computing, Tokyo Institute of Technology. From October 2016 to March 2020 and from October 2021 to present, he was/is a Researcher of Precursory Research for Embryonic Science and Technology (PRESTO), Japan Science and Technology Corporation (JST), Tokyo, Japan. His research interests include signal processing, image analysis, remote sensing, mathematical optimization, and data science.

Dr. Ono received the Young Researchers’ Award and the Excellent Paper Award from the IEICE in 2013 and 2014, respectively, the Outstanding Student Journal Paper Award and the Young Author Best Paper Award from the IEEE SPS Japan Chapter in 2014 and 2020, respectively, the Funai Research Award from the Funai Foundation in 2017, the Ando Incentive Prize from the Foundation of Ando Laboratory in 2021, and the Young Scientists’ Award from MEXT in 2022. He has been an Associate Editor of IEEE TRANSACTIONS ON SIGNAL AND INFORMATION PROCESSING OVER NETWORKS since 2019.
\end{IEEEbiography}





\end{document}